\documentclass{article}

\usepackage[english]{babel}
\usepackage{amsfonts,amsmath,amssymb}
\usepackage{graphicx}
\usepackage[letterpaper,top=2cm,bottom=2cm,left=3cm,right=3cm,marginparwidth=1.75cm]{geometry}
\usepackage{enumerate}  

\usepackage{amsmath}
\usepackage{graphicx}
\usepackage[colorlinks=true, allcolors=blue]{hyperref}

\usepackage{amsthm}
\usepackage{natbib}
\usepackage{tikz}
\setcitestyle{authoryear}

\newtheorem{example}{Example}
\newtheorem{theorem}{Theorem}
\newtheorem{lemma}{Lemma}
\newtheorem{assumption}{Assumption}
\newtheorem{proposition}{Proposition}

\newtheorem{remark}{Remark}
\usepackage{mathtools}  
\mathtoolsset{showonlyrefs}

\usepackage{scalerel,stackengine}
\stackMath
\newcommand\reallywidehat[1]{%
\savestack{\tmpbox}{\stretchto{%
  \scaleto{%
    \scalerel*[\widthof{\ensuremath{#1}}]{\kern-.6pt\bigwedge\kern-.6pt}%
    {\rule[-\textheight/2]{1ex}{\textheight}}%
  }{\textheight}%
}{0.5ex}}%
\stackon[1pt]{#1}{\tmpbox}%
}
\usepackage{xcolor}

\usepackage[linesnumbered,ruled,vlined]{algorithm2e}
\usepackage[skins]{tcolorbox}

\SetCommentSty{mycommfont}

\SetKwInput{KwInput}{Input}               
\SetKwInput{KwOutput}{Output}

\title{Calibrated inference: statistical inference that accounts for both sampling uncertainty and distributional uncertainty}

\author{Yujin Jeong and Dominik Rothenh\"ausler}

\begin{document}

\maketitle

\begin{abstract}
How can we draw trustworthy scientific conclusions? One criterion is that a study can be replicated by independent teams. While replication is critically important, it is arguably insufficient. If a study is biased for some reason and other studies recapitulate the approach then findings might be consistently incorrect. It has been argued that trustworthy scientific conclusions require disparate sources of evidence. However, different methods might have shared biases, making it difficult to judge the trustworthiness of a result. We formalize this issue by introducing a ``distributional uncertainty model", wherein dense distributional shifts emerge as the superposition of numerous small random changes. The distributional perturbation model arises under a symmetry assumption on distributional shifts and is strictly weaker than assuming that the data is i.i.d.\ from the target distribution. We show that a stability analysis on a single data set allows us to construct confidence intervals that account for both sampling uncertainty and distributional uncertainty.
\end{abstract}

\section{Introduction}

Statistical inferences can be fragile. 
If we compare two analyses conducted by different data scientists on different data sets, variation can be due to sampling, due to distribution shift, or due to a change in methodology. These issues raise a fundamental question: How can we draw trustworthy scientific conclusions? A common recommendation is to have independent teams attempt to replicate the findings of others. While replication is critically important, it is arguably insufficient. If a study suffers from biases for some reason and replication studies emulate the study, the findings will be consistently incorrect \citep{munafo2018repeating}. To solve this issue, researchers have advocated investigating independent lines of evidence \citep{denzin1970research,freedman1991statistical,rosenbaum2010evidence,munafo2018repeating}. Ideally, these different lines of evidence are susceptible to different biases. Intuitively, if results agree across different methodologies, then a statistical finding is less likely to be an artifact. However, it might be expensive and impractical to ask several researchers to run studies independently. 

Can we emulate this strategy on a single data set? In fact, stability analyses have been advocated by many researchers. To be more precise, it has been recommended to evaluate several reasonable modelling choices for one single data set \citep{leamer1983let,rosenbaum2010evidence,patel2015assessment,steegen2016increasing,yu2020veridical}.  Practitioners often compute multiple estimators for a single target quantity by running differently specified regressions or considering the perturbations induced by various forms of data pre-processing.  If the estimator-to-estimator variability is high, then the analyst has reason to distrust the estimates. 

This practice — computing multiple estimators for a single target quantity and studying their estimator-to-estimator variability — warrants an investigation into its theoretical properties. If the estimator-to-estimator variability is high, it may raise concerns about the reliability of the estimates. However, what criterion tells us whether we should be concerned about such variability? Does this practice come with any guarantees, and if so, which ones? What mathematical problems do we address by examining the estimator-to-estimator variability?  
Often the decision on what qualifies as a “stable result” is left up to the individual judgements of the analyst. One could analyze the estimators with a random effects model, but since the estimators are computed on the same data set, they might share biases. More concretely speaking, since the structure of the biases is generally unknown, it is not clear how to define a random effects model that captures the correlation structure of the estimators' biases.

In order to discuss these questions in a formal framework, we take a distributional perspective.  
We consider a setting where the data is drawn from a ``perturbed" or ``contaminated" distribution, while the goal is to infer some properties of the uncontaminated distribution.
The classical robust statistics literature \citep{huber81} addresses distributional perturbations by investigating the worst-case behavior of a statistical functional over a fixed neighborhood of the model. 
More recently, distributional uncertainty sets based on $f$-divergence have been linked to distributionally robust optimization \citep{ben2013robust,duchi2021statistics}. In such models, it is challenging to choose the appropriate set of distributions, since the size of the perturbation is generally unknown.

Selection bias, confounding variables, or batch effects may be seen as sparse distribution shifts, where the shift affects only specific parts of the data generation process, while other parts stay invariant.

In contrast to sparse distribution shifts, we consider dense distribution shifts where the shift arises as the superposition of many small random changes that affect all parts of the data generating process. Here is an example that illustrates how dense distribution shifts may arise in practice: In clinical trials, distribution shifts can arise from complex, unobservable factors in the underlying population. Consider a scenario where a trial is conducted in 2022. During this time, various factors influence the population's health, such as the severity of the flu that year, local demographic patterns, and other environmental or social variables. When attempting to replicate the trial in a later period, the distribution of these factors may have changed significantly. However, because there are so many contributing variables—many of which are difficult or impossible to measure—the cumulative effect of these changes cannot be precisely quantified. This scenario exemplifies a dense shift: the overall distribution is shaped by distribution shifts in a multitude of nonrandom factors, whose combined effects are so intricate and unpredictable that they can be modeled as effectively random for practical purposes. We model dense distribution shifts as random and symmetric by randomly up-weighting and down-weighting different parts of the target distribution, capturing the inherent unpredictability and complexity of dense shifts.
As another example, in economics, business cycles are often seen as driven by ‘random summation of random causes’ \citep{drautzburg2019}. One potential interpretation of this is that in complex social and economic systems, there might be dense and unpredictable shifts driven by a confluence of broader social, economic, and natural forces. 

Motivated by an empirical phenomenon observed in several real-world data sets, we define a family of random symmetric perturbations.  While the symmetry assumption is strong, it is strictly weaker than assuming the data is i.i.d.\ from the target distribution. Distribution shift observed in the real world may involve a combination of sparse and dense shifts and we may need a hybrid approach in the future. As a first stepping stone, we aim to establish theoretical foundations addressing dense, symmetric distribution shifts. 

Finally, we show that modeling distributional perturbations as random and symmetric has an intriguing consequence: using a stability analysis, it is possible to estimate the strength of the distributional perturbation. Based on an estimate of the distributional perturbation strength, we propose confidence intervals that capture both sampling uncertainty and distributional uncertainty.

\subsection{Related Work}

Considerations of model stability have emerged in Bayesian statistics \citep{box1980sampling,skene1986bayesian}, causal inference \citep{leamer1983let,lalonde1986evaluating,rosenbaum1987role,imbens2015causal} and in discussions about the data science life-cycle \citep{yu2013stability,steegen2016increasing,yu2020veridical}. Using different estimation strategies is commonly recommended to corroborate a causal hypothesis \citep{freedman1991statistical,rosenbaum2010evidence,karmakar2019integrating}. In particular, to evaluate omitted variable bias, it is a common recommendation to consider the between-estimator variation of several adjusted regressions \citep{oster2019unobservable}. Sensitivity analysis bounds the influence of confounders that have been omitted in a regression or matching procedure and has played an influential role in increasing trustworthiness of causal inference from observational data \citep{cornfield1959smoking,rosenbaum1983assessing,vanderweele2017sensitivity}. %
It has been argued that causal mechanisms are expected to lead to stable associations across settings, if the same mechanism is shared across settings. Based on this observation, stability principles have been employed to discover causal relationships based on heterogeneous data sets \citep{peters2016causal,rothenhausler2015backshift,buhlmann2020invariance,pfister2021stabilizing}. Stability principles are heavily used in machine learning, often with the goal of variance reduction. For example, some tree-based methods employ feature bagging, which can be seen as averaging over differently specified prediction models \citep{breiman1996bagging,breiman2001random}.  Dropout in neural networks is another form of algorithm perturbation \citep{srivastava2014dropout}. Distributional uncertainty sets based on $f$-divergences have been linked to distributionally robust optimization \citep{ben2013robust,duchi2021statistics}. In the context of prediction under distribution shift, stability or invariance principles have been employed to learn prediction mechanisms that generalize to new settings \citep{schoelkopf2012,zhang2013domain,rojas2015causal,heinze2021conditional,rothenhausler2021anchor}. Quasi-likelihoods \citep{wedderburn1974quasi} are a way to allow greater variability in the data than what is expected from the model. However, uncertainty quantification in quasi-likelihoods still only deals with sampling uncertainty, while we aim to quantify uncertainty due to both sampling and distributional uncertainty.

\subsection{Outline of The Paper}
In Section~\ref{sec:usual-business}, we will quickly review standard practice for forming confidence intervals. In Section~\ref{sec:setting}, we introduce the setting of the paper and discuss why standard statistical practice does not account for all types of uncertainty in this setting. The setting of our paper arises under a distributional perturbation model described in Section~\ref{sec:setting} and sampling procedures described in the Appendix. We then turn to statistical inference. In Section~\ref{sec:calibrated}, we discuss how to form confidence intervals in our setting. This completes the picture from an inferential viewpoint. 
In Section~\ref{sec:numerical}, we evaluate the performance of the proposed procedure on a simulated example from causal inference. In Section~\ref{sec:real-world} we demonstrate that the proposed procedure can increase the stability of decision-making based on real-world data. We conclude in Section~\ref{sec:discussion}.

\subsection{Standard Approach}\label{sec:usual-business}

Let us consider estimation of the mean $\theta^0 = \mathbb{E}[D]$ of a square-integrable real-valued random variable $D \in \mathcal{D}$, $D \sim \mathbb{P}$. Assume that we are given data $(D_i)_{i=1,\ldots,n} \stackrel{\text{i.i.d.}}{\sim} \mathbb{P}$ with $\text{Var}(D_i) = \sigma^2 \in (0,\infty)$. We can estimate $\sigma^2$ via $\hat \sigma^2 = \frac{1}{n-1} \sum_{i=1}^n (D_i - \overline D)^2$ to form asymptotically valid confidence intervals that means
\begin{equation}\label{eq:ciusual}
    \mathrm{P}( \overline D  - z_{1-\alpha/2} \hat \sigma/\sqrt{n} \le \theta^0  \le \overline D  + z_{1-\alpha/2} \hat \sigma/\sqrt{n} ) \rightarrow 1- \alpha,
\end{equation}
where $z_{1-\alpha/2}$ is the $1-\alpha/2$ quantile of a standard Gaussian random variable. This practice is justified by the central limit theorem which implies
\begin{equation}\label{eq:main}
    \frac{1}{\sqrt{n}} \sum_{i=1}^n (D_i - \mathbb{E}[D]) \stackrel{d}{\xrightarrow{}} \mathcal{N}(0,\text{Var}(D)).
\end{equation}
More generally, for some vector-valued data $D_i \stackrel{\text{i.i.d.}}{\sim} \mathbb{P}$  consider a parametrized model $\{ p_\theta, \theta \in \Omega \}$ of positive probability densities $p_\theta$ with respect to some $\sigma$-finite measure $\mu$. Assume that the parameter space $\Omega$ is an open subset of $\mathbb{R}^d$. We consider the maximum-likelihood estimator 
\begin{equation*}
    \hat \theta = \arg \max \sum_{i=1}^n  \log p_\theta(D_i),
\end{equation*}
for some unknown target parameter $\theta^0(\mathbb{P}) = \arg \max \mathbb{E}[\log p_\theta(D)]$, where $D \sim \mathbb{P}$.
Under regularity assumptions \citep{van2000asymptotic,tsiatis2006semiparametric}, for $n \rightarrow \infty$, 
\begin{equation*}
    \sqrt{n} (\hat \theta - \theta^0) =  \frac{1}{\sqrt{n}} \sum_{i=1}^n -\mathbb{E}[\partial_\theta^2 \log p_{\theta^0}(D)]^{-1} \partial_\theta \log p_{\theta^0}(D_i) + o_p(1) \stackrel{d}{\xrightarrow[]{}} \mathcal{N}(0, \Sigma),
\end{equation*}
where $\Sigma = \mathbb{E}[\partial_\theta^2 \log p_{\theta^0}(D)]^{-1}\text{Var}(\partial_\theta \log p_{\theta^0}(D))\mathbb{E}[\partial_\theta^2 \log p_{\theta^0}(D)]^{-1}$. Thus, based on a consistent estimator $\hat \Sigma \rightarrow \Sigma$, one can form asymptotically valid confidence intervals via
\begin{equation}\label{eq:mle-ci}
    \hat \theta_k \pm z_{1-\alpha/2} \frac{\sqrt{\hat \Sigma_{kk}}}{\sqrt{n}}.
\end{equation}
A similar approach can be used to construct asymptotically valid confidence intervals for $M$-estimators. In the following, we discuss situations in which this approach does not have the desired coverage.

\section{Distributional Uncertainty}\label{sec:setting}

    There are several reasons why the coverage in equation~\eqref{eq:ciusual} might not hold. The main focus of this paper will be violations of \eqref{eq:main} due to what we call distributional uncertainty. Due to distribution shifts, the data scientist might not draw a sample from the target distribution $\mathbb{P}$ but from some $\mathbb{P}^\xi \neq \mathbb{P}$. The data analyst may try to  address the source of bias by re-weighting, regression adjustment, random effect modeling, a bias correction, or other statistical techniques. 
    Our viewpoint is that when using such techniques, it is likely that some residual error remains which we might want to address by scaling the confidence intervals. Ideally, we would like to construct confidence intervals that detect residual errors due to distributional perturbations, and account for them, if necessary. We model the variation due to distributional changes as random. This will allow us to integrate both distributional uncertainty and sampling uncertainty in a natural fashion.

    Let us make this more concrete by returning to the example of estimating the mean. Due to a superposition of small random errors, the data scientist might not draw a sample from the target distribution $\mathbb{P}$. Instead, the data $(D_i)_{i=1,\ldots,n}$ might be drawn i.i.d.\ from some perturbed distribution $\mathbb{P}^\xi \neq \mathbb{P}$, where $\xi$ is a random variable and $\mathbb{P}^\bullet$ is a probability distribution for each fixed $\bullet \in \text{range}(\xi)$. Then the error of the empirical mean can be decomposed:
    \begin{equation}\label{eq:regimes}
       \frac{1}{n} \sum_{i=1}^n D_i - \mathbb{E}[D] = \underbrace{ \frac{1}{n} \sum_{i=1}^n D_i - \mathbb{E}^\xi[D] }_{\substack{\text{variation due to} \\ \text{sampling}}} + \underbrace{\mathbb{E}^\xi[D]- \mathbb{E}[D]}_{\substack{\text{variation due to} \\ \text{distributional perturbation}}}.
    \end{equation}
    Here, $\mathbb{E}$ denotes the expectation under the target distribution $\mathbb{P}$ and $\mathbb{E}^{\xi}$ denotes the expectation under the perturbed distribution $\mathbb{P}^{\xi}$.
    Equation~\eqref{eq:main} usually does not hold in this setting as distributional perturbations induce additional variation. 

    In the following, we focus on the regime where the variation due to sampling and the variation due to distributional perturbations are both of the order $1/\sqrt{n}$. This choice is motivated by observations from real-world data sets, where these variations often appear to be of similar order (see Figure \ref{fig:qqplots}). There may be situations where higher-order bias exists and should be removed when possible. However, even after accounting for all estimable bias components, residual data quality issues can remain. We would like to address these residual errors by scaling confidence intervals.

    \paragraph{\textit{Notation}.}

    Let $\mathbb{P}$ denote an unknown fixed target probability measure on $\mathcal{D}$. For each fixed $n$ the random variable $\xi(n) \in \Xi$ encodes the distributional perturbation. Formally, $\mathbb{P}^\bullet$, $\bullet \in \Xi$, is a stochastic kernel with the target space $\mathcal{D}$.  Conditionally on $\xi(n)$, we draw an i.i.d.\ sample $(D_1^n,\ldots,D_n^n)$ from $\mathbb{P}^{\xi(n)}$. $\xi(n)$ might depend on $n$ but we suppress this in the notation and simply write $\xi$. Similarly, we sometimes suppress the dependence of $(D_1^n,\ldots,D_n^n)$ on $n$ and simply write $(D_1,\ldots,D_n)$. We write $P$ for the marginal distribution of $(D_1,\ldots,D_n,\xi)$. We denote $E$ as the marginal expectation under $P$, $\mathbb{E}$ as the expectation under the target distribution $\mathbb{P}$, and $\mathbb{E}^{\xi}$ as the expectation under $\mathbb{P}^{\xi}$ (conditioned on $\xi$). We write $\text{Var}_P$ for the variance under $P$ and $\text{Var}_\mathbb{P}$ for the variance under $\mathbb{P}$.

    \subsection{Empirical examples}
    What is a reasonable model for the distributional perturbation? We draw inspiration from direct replication studies and the GTEx\footnote{The Genotype-Tissue Expression (GTEx) Project was supported by the Common Fund of the Office of the Director of the National Institutes of Health, and by NCI, NHGRI, NHLBI, NIDA, NIMH, and NINDS. The data set used for the analyses described in this manuscript is version 6 and can be downloaded in the GTEx Portal: \url{www.gtexportal.org}.} gene expression data.

    For the two direct replication studies in \cite{prochazka2022pain}, we assume that the first replication study $D_1,\ldots,D_{n_1}$ is drawn i.i.d.\ from the target distribution $\mathbb{P}$, and the data set of the second replication study consists of $n_2$ samples $(D_1', \dots, D_{n_2}')$ drawn from the perturbed distribution $\mathbb{P}^{\xi}$. Analogously, for GTEx gene expression data (V6),  we randomly selected the tissue ``Liver" and consider the data set from the tissue as $n_2$ samples from the perturbed distribution. Then 
    we define the target distribution as the distribution of gene expression across the remaining tissues. 
    
    From a statistical perspective, it is natural to study the distribution of the standardized mean difference: 
    \begin{equation} \label{eq:ttest}
          \left( \frac{1}{n_1} + \frac{1}{n_2} \right)^{-1/2} \left(\frac{\frac{1}{n_2}\sum_{i=1}^{n_2} \psi(D_i') - \frac{1}{n_1}\sum_{i=1}^{n_1} \psi(D_i) }
       {\hat{\text{sd}}_{\mathbb{P}}(\psi(D))} \right) \text{ for some test function } \psi.
    \end{equation}
    If there were no distributional shifts across different replication studies or different tissues and all data were sampled i.i.d. from $\mathbb{P}^\xi = \mathbb{P}$, for fixed $\psi$ one would expect the ratio to follow roughly a standard Gaussian distribution.

        \begin{figure}[ht]
        \centering
        \includegraphics[scale = 0.4]{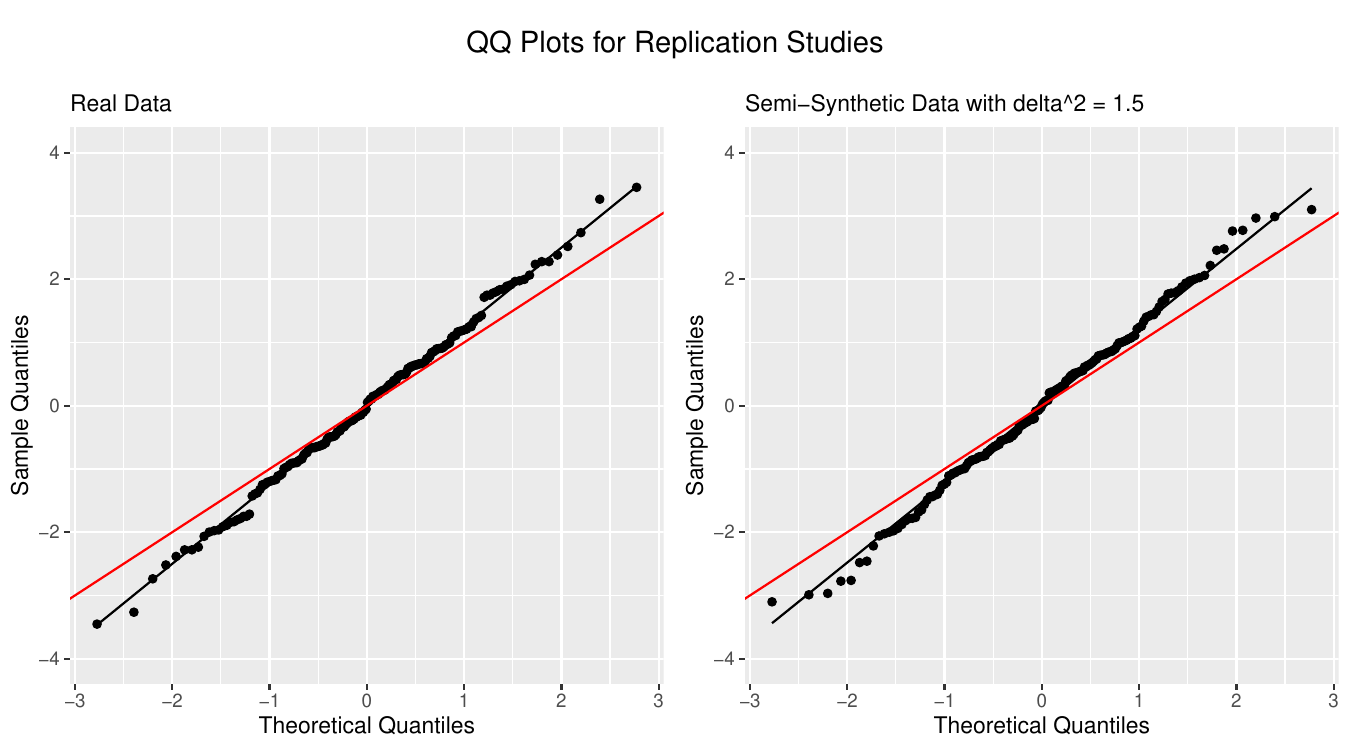}
        \includegraphics[scale = 0.4]{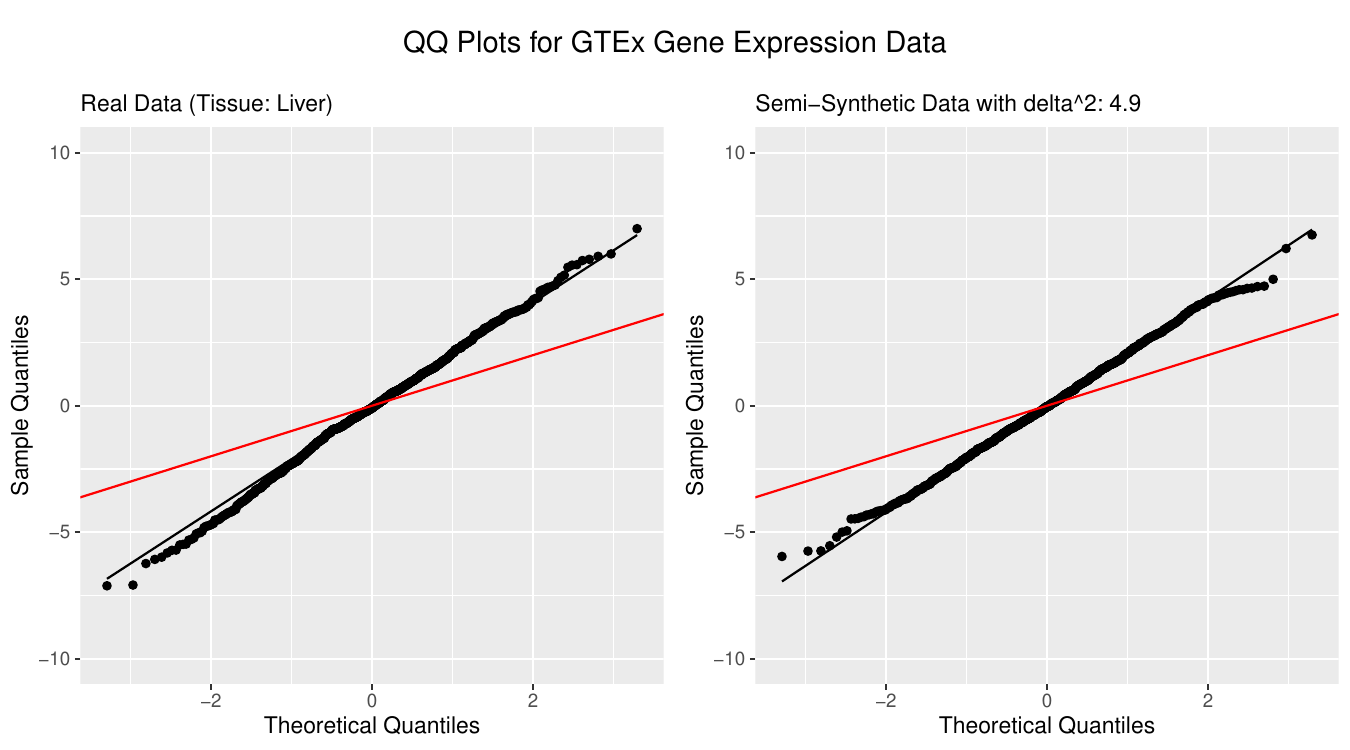}
        \caption{QQ plots of \eqref{eq:ttest} for various test functions $\psi_l$ 
  on replication studies in \cite{prochazka2022pain} (above) and GTEx gene expression data set (below). The QQ plots on the left side are from real-world data sets. The red line represents the expected QQ line if the data were all drawn i.i.d.\ from some (unperturbed) distribution $\mathbb{P}$. Perhaps surprisingly, the standardized means are on a line, indicating that the distribution shift has some structure that we can exploit for estimation and inference. The QQ plots on the right side are computed on the data drawn from our model \eqref{eq:random-reweight} with estimated variance inflation factor $\delta$ from the left side. The simulated shifts on the right closely match the pattern observed on the left side.  
        }
        \label{fig:qqplots}
    \end{figure}

    To investigate the distribution of equation~\eqref{eq:ttest} in practice, we define $\psi_{\ell}$ as follows. For the replication studies in \cite{prochazka2022pain}, we define $\psi_{\ell}$ as either the covariate or the sign-flipped covariate for either the treatment or control group, resulting in 180 test functions. For GTEx gene expression data, the covariate means are standardized in a pre-processing step. Thus, we study cross-products. We randomly select 1000 gene pairs and define $\psi_{\ell}$ as the product of gene-expressions for the $\ell$-th pair. 
    We illustrate the behavior of \eqref{eq:ttest} from replication studies and GTEx data on the left side of Figure \ref{fig:qqplots} using QQ plots. 

    \vspace{2mm}
    
    \textbf{Unexpected: distribution shift on the line.}
    Surprisingly, the QQ plots in Figure~\ref{fig:qqplots} indicate that the statistics in \eqref{eq:ttest} follow a Gaussian distribution. The slopes of the black QQ lines are larger than 1, which indicate excess variation compared to i.i.d.\ sampling from $\mathbb{P}$. This raises the question of whether we can model such (moderate) variance inflation with a statistical model for distribution shift.

    \vspace{2mm}

    \textbf{`Distribution shift on the line' implies isotropic perturbations.}
    One can show that constant variance inflation as in Figure~\ref{fig:qqplots} implies that the shifted distribution arises from randomly re-weighting the original distribution with uncorrelated weights with equal variance.
    To not interrupt the flow of the discussion, we provide the justification for this claim in the Appendix, Section~\ref{sec:appendix-motivation}.
    
    \vspace{2mm}

    \textbf{Isotropic perturbations imply `distribution shift on the line'.}
    In Section~\ref{sec:isotropic-model}, we will introduce a random perturbation model (equation~\eqref{eq:random-reweight}) that is based on randomly re-weighting the target distribution with independent weights with equal variance.

    \vspace{2mm}
    
    As a sanity check, we semi-synthetically sample from this model. If the random perturbation model is reasonable, the semi-synthetic data should exhibit similar patterns as the left-hand side of Figure~\ref{fig:qqplots}. From the target data set, we sample a perturbed data set using 
    the estimated variance inflation factor $\delta$ obtained from the left side of Figure \ref{fig:qqplots}. The QQ plots generated from the semi-synthetic data are presented on the right side of Figure \ref{fig:qqplots}. 
    We see that these plots closely resemble the patterns observed in the QQ plots obtained from real-world data sets.

    \subsection{The Isotropic Perturbation Model}\label{sec:isotropic-model}
    
    In this section, we construct a general isotropic perturbation model for multivariate continuous or discrete random variables 
    and consider the case where the change in measure is a superposition of small incremental changes. 
    We will see that under such a model, we will get a non-standard CLT in the sense that sample means are asymptotically normal, but with a different variance formula compared to the i.i.d.\ case. 
    
    To recap, in a random perturbation model, the data is not directly drawn from the target distribution $\mathbb{P}$, but from some random probability measure $\mathbb{P}^\xi$, where $\mathbb{P}^\xi$ is close to $\mathbb{P}$. The idea is that due to numerous random distributional changes, the actual sampling distribution $\mathbb{P}^\xi$ randomly differs from the target distribution $\mathbb{P}$. Under $\mathbb{P}^\xi$, probabilities of events are slightly up-weighted or down-weighted compared to $\mathbb{P}$.

    We want to construct a random perturbation model that includes many commonly encountered situations such as distributions on $\mathbb{R}^d$ or the (infinite-dimensional) space of continuous functions on $\mathbb{R}$. A result from probability theory shows that any random variable $D$ on a finite or countably infinite dimensional probability space can be written as a measurable function $D  \stackrel{d}{=} h(U)$, where $U$ is a uniform random variable on $[0,1]$.\footnote{For any Borel-measurable random variable $D$ on a Polish (separable and completely metrizable) space $\mathcal{D}$, there exists a Borel-measurable function $h$ such that $D \stackrel{d}{=} h(U)$ where $U$ follows the uniform distribution on $[0,1]$ \citep{dudley2018real}.} Thus, without loss of generality we will construct distributional perturbations for a uniform distribution on $[0,1]$. With the transformation $h(\cdot)$ defined above, this construction generalizes to the general cases by setting
    \begin{equation*}
        \mathbb{P}^{\xi}(D \in \bullet) = \mathbb{P}^{\xi}(h(U) \in \bullet).
    \end{equation*}
    The role of $h(\cdot)$ is mainly to ensure that the probability space is rich enough to be transformed into a uniform random variable. In principle, this construction is not unique. There may be many possible choices of $h(\cdot)$ that result in $D \stackrel{d}{=} h(U)$. However, as we will see below, the asymptotic behaviour of the perturbation model does not depend on the choice of $h$.
  
    Let us now construct the distributional perturbation for a uniform random variable. As discussed in Appendix, Section~\ref{sec:appendix-motivation}, the distribution-shift-on-the-line phenomenon observed in Figure~\ref{fig:qqplots} suggests that we can think about the shifted distribution as arising from randomly re-weighting the original distribution with (almost) uncorrelated weights with equal variance. Thus, we take $m$ bins $I_k = [(k-1)/m, k/m]$ for $k = 1, \dots, m$. Let $W_1, \dots, W_m$ be i.i.d.\ positive random variables with finite variance.   Set $\xi = (W_1,\ldots,W_m)$. We define the randomly perturbed distribution $\mathbb{P}^\xi$ by setting
   \begin{equation}\label{eq:random-reweight} 
    \mathbb{P}^\xi(U \in \bullet) = \sum_k \mathbb{P}(U \in I_k \cap \bullet) \cdot \frac{W_k}{\sum_{k=1}^{m} W_k/m}.
    \end{equation}
     
    Let $m = m(n)$ such that $\frac{n}{m(n)}$ converges to some limit $r \in (0,\infty)$. Note that $\xi$ depends on $m$ and thus also on $n$. Conditionally on $\xi$, let $(D_1^n,\ldots,D_n^n)$ be i.i.d.\ draws from $\mathbb{P}^\xi$. 
   \begin{lemma}[CLT under distributional uncertainty]\label{lemma:random_weight}
    Under the assumptions mentioned above (Section~\ref{sec:isotropic-model}), for any Borel-measurable square-integrable function $\psi : \mathcal{D} \mapsto \mathbb{R}^l$, we have
    \begin{equation}\label{eq:normality}
        \frac{1}{\sqrt{n}}\sum_{i=1}^{n} (\psi(D_i^n) - \mathbb{E}[\psi(D)]) \stackrel{d}{\xrightarrow[]{}} \mathcal{N}(0,  \delta^2 \text{Var}_{\mathbb{P}}(\psi(D))),
    \end{equation}
     with
    \begin{equation*}
        \delta^2 = 1 + \frac{r \text{Var}(W_1)}{E[W_1]^2}.
    \end{equation*}
    In other words, the marginal distribution of $\frac{1}{\sqrt{n}} \sum_{i=1}^n \psi(D_i^n)$ is asymptotically Gaussian with asymptotic variance containing a scaling factor $\delta^2$. 
    \end{lemma}
    
    The proof can be found in the Appendix, Section~\ref{sec:asymp_pert_model}. The reason we consider a triangular array of data sets is that, motivated by the empirical example in Figure~\ref{fig:qqplots}, we consider a setting where the sampling uncertainty and distributional uncertainty are of the same order.

    Unless explicitly mentioned otherwise, in the following we assume that the data scientist has access to one such data set $(D_1^n,\ldots,D_n^n)$ for some large $n$. Note that if the data $(D_1,\ldots,D_n)$ is $\stackrel{\text{i.i.d.}}{\sim} \mathbb{P}$, then equation~\eqref{eq:normality} holds for  $\delta=1$ and $D_i^n = D_i $. Thus, equation~\eqref{eq:normality} is weaker than assuming that the data is drawn i.i.d.\ from $\mathbb{P}$.

    The asymptotic behaviour shown in equation~\eqref{eq:normality} arises 
    not only under the distributional perturbation model but other types of sampling procedures that induce dependence between observations. In Appendix \ref{sec:examplesofmodels}, we discuss other sampling models that give rise to \eqref{eq:normality}.

 In the following, we will discuss how estimators behave asymptotically under equation~\eqref{eq:normality}. It turns out that under some regularity assumptions, maximum likelihood estimators are still consistent and asymptotically normal, but with the scaling factor $\delta^2$ in the variance formula.

 \subsection{Asymptotic Behaviour of M-estimators}
Here we will consider the asymptotic behaviour of estimators $\hat \theta = \arg \min_{\theta \in \Omega} \frac{1}{n} \sum_{i=1}^n L(\theta,D_i^n)$ for a target defined via $\theta^0 = \arg \min_{\theta \in \Omega} \mathbb{E}[L(\theta,D)],$
where $L(\theta,\bullet)$ is a Borel-measurable loss function and $\Omega$ is an open subset of $\mathbb{R}^d$. These estimators include maximum likelihood estimators with $L(\theta,D) = - \log p_\theta(D)$. 

In classical statistical theory, uncertainty quantification is usually based on showing that the estimator is asymptotically Gaussian. Since we have a different two-stage sampling model, one has to verify that a similar approximation -- with a different variance formula -- still holds in our setting. 

First, we will discuss consistency. Instead of aiming for maximal generality, we will adapt a simple consistency proof from the literature. In particular, we will adapt the classical consistency result in \cite{van2000asymptotic}, Section 5.2.1. We expect that other consistency proofs can be adapted similarly. The main difference in the proof is that since the data is not i.i.d.\ from the target distribution  we cannot directly rely on the law of large numbers. The proof can be found in Appendix~\ref{appendix:m-estim}.
\begin{lemma}[Consistency of M-estimators]\label{lemma:asymptotic-consistency-m-estimators}
Consider the M-estimator $$\hat \theta = \arg \min_{\theta \in \Omega} \frac{1}{n} \sum_{i=1}^n L(\theta,D_i^n),$$ and the target $\theta^0  = \arg \min_{\theta \in \Omega} \mathbb{E}[L(\theta,D)]$,
where $\Omega$ is a compact subset of $\mathbb{R}^d$. Furthermore assume that $\theta \mapsto L(\theta,D)$ is continuous and that $ \inf_{ \| \theta - \theta' \|_2 \le \delta} L(\theta, D)$ is square-integrable under $\mathbb{P}$ for every $\delta$ and $\theta'$ and that $\inf_{\theta \in \Omega} L(\theta,D)$ 
 is square integrable. 
We assume that $\mathbb{E}[ L(\theta,D)]$ has a unique minimum.
Then,
\begin{equation*}
    \hat \theta - \theta^0 = o_p(1).
\end{equation*}
\end{lemma}

Now let us turn to asymptotic normality. We will modify the proof in \cite{van2000asymptotic}, Section~5.6. Similarly as above, the main difference in the proof is since the data is not i.i.d.\ from the target distribution we cannot directly rely on the law of large numbers or a standard CLT. The proof can be found in Appendix~\ref{appendix:m-estim}.
\begin{lemma}[Asymptotic normality of M-estimators]\label{lemma:asymptotic-gaussianity-m-estimator}
For each $\theta$ in an open subset of $\Omega$, let $\theta \mapsto \partial_\theta L(\theta,D)$ be twice continuously differentiable in $\theta$ for every $D$.  Assume that the matrix $\mathbb{E}[ \partial_\theta^2 L(\theta^0,D)]$ exists and is nonsingular. Assume that third order partial derivatives of $ \theta \mapsto L(\theta,D)$ are dominated by a fixed function $h(\cdot)$ for every $\theta$ in a neighborhood of $\theta^0$. We assume that $\partial_\theta L(\theta^0,D)$, $\partial_\theta^2 L(\theta^0,D)$ and $h(D)$ are square-integrable under $\mathbb{P}$.  Let $\hat \theta = \arg \min \frac{1}{n} \sum_{i=1}^n L(\theta,D_i^n)$. Assume that $\hat \theta - \theta^0 = o_p(1)$, where $\theta^0$ satisfies the estimating equation $\mathbb{E}[\partial_\theta L(\theta^0,D)] = 0$. Then,
\begin{equation*}
    \sqrt{n}(\hat \theta - \theta^0) =- \frac{1}{\sqrt{n}} \sum_{i=1}^n  \mathbb{E}[ \partial_\theta^2 L(\theta^0,D) ]^{-1} \partial_\theta L(\theta^0,D_i^n)  + o_p(1).
\end{equation*}
In particular, by Lemma~\ref{lemma:random_weight} we have that $\sqrt{n}(\hat \theta - \theta^0)$ converges in distribution to a normal distribution with mean zero and covariance matrix $ \delta^2 \Sigma$, where 
\begin{equation*}
\Sigma = \mathbb{E}[ \partial_\theta^2 L(\theta^0,D) ]^{-1} \mathbb{E}[ \partial_\theta L(\theta^0,D) \partial_\theta L(\theta^0,D)^\intercal] \mathbb{E}[ \partial_\theta^2 L(\theta^0,D) ]^{-1}.
\end{equation*}
\end{lemma}

The upshot is that $M$-estimators are asymptotically unbiased, marginally across both sampling uncertainty and distributional uncertainty. However, the variance formula changes in the sense that there is an (unknown) scaling factor $\delta^2$.

\subsection{The Standard Mode of Inference Fails}

Let us quickly sketch why the standard mode of inference fails. Let's consider the case of estimating the mean $\theta^0 = \mathbb{E}[D]$ via $\hat \theta = \frac{1}{n} \sum_{i=1}^n D_i^n$. One may be tempted to use the standard variance estimate $\hat \sigma_\text{naive}^2/n$, where 
\begin{equation*}
   \hat \sigma^2_\text{naive} = \frac{1}{n} \sum_{i=1}^n \left(D_i^n - \frac{1}{n} \sum_{j=1}^n D_j^n\right)^2.
\end{equation*}
However, a short calculation shows that
\begin{equation*}
  \hat \sigma^2_\text{naive}  = \frac{1}{n} \sum_{i=1}^n \left(D_i^n - \frac{1}{n} \sum_{j=1}^n D_j^n \right)^2 = \frac{1}{n} \sum_{i=1}^n (D_i^n)^2 - \left(\frac{1}{n} \sum_{j=1}^n D_j^n\right)^2 = \text{Var}_\mathbb{P}(D) + o_P(1).
\end{equation*}
Here, we used equation~\eqref{eq:normality} for $\psi(D) = D$ and $\psi(D) = D^2$. However, as shown in Lemma~\ref{lemma:random_weight}, the asymptotic variance of $\hat \theta$ is $\frac{\delta^2}{n} \text{Var}_\mathbb{P}(D)$.
Thus, the standard approach drastically underestimates variance in our setting. If $\delta$ is known, one can simply stretch the confidence intervals discussed in Section~\ref{sec:usual-business} accordingly. However, in general $\delta$ will be unknown and has to be estimated from data. We will discuss the estimation of $\delta$ in Section~\ref{sec:calibrated}.  

\section{Calibrated Inference}\label{sec:calibrated}

 We will now discuss how to estimate $\delta$ and form asymptotically valid confidence intervals for $\theta^0$.
 As discussed earlier, data analysts often have not just one reasonable estimator for a given parameter $\theta^0$, but potentially several reasonable estimators $\hat \theta^1, \ldots,\hat \theta^K$. For example, these estimators can arise from using different specifications in generalized linear models or by running the analysis for subgroups of the observations.

\begin{example}[OLS with several specifications]\label{example:ols} Let us consider a setting in which the data analyst wants to estimate the causal effect of some variable $X_1$ on a target variable $Y$. On observational data, this is often done by invoking suitable assumptions and regressing $Y$ on $X_1$ and a suitable set of covariates. Often, the analyst has several reasonable choices for the set of covariates. Suppose that the data analyst performs ordinary least-squares on $K$ different subsets of $X$ that include $X_1$, denoted by $X^{S_1}, X^{S_2}, \dots, X^{S_K}$. For example, $X^{S_1}$ can be $(X_1, X_2, X_3)$. Now the data analyst has $K$ different regression coefficients of $X_1$, $\hat{\theta}^1, \dots, \hat{\theta}^K$ where
$$
\hat{\theta}^k =  \left(\sum_{i=1}^{n}X_i^{S_k}{(X_i^{S_k})}^\intercal\right)_{1,\bullet}^{-1}\sum_{i=1}^{n}X_i^{S_k}Y_i.
$$
\end{example}
 
 If the empirical variation between the estimators $\hat \theta^1,\ldots,\hat \theta^K$ is low, then the analyst may feel more confident about conclusions drawn from these estimates than if the variation between these estimators is very large. As an example, in \citet{chiappori2012fatter} the authors write ``It is reassuring that the estimates are very similar in the standard and the
 augmented specifications". We will now look at this practice under the isotropic perturbation model. We will see that in this setting it is possible to construct a consistent estimator of $\delta$ and form asymptotically valid confidence intervals that account for both sampling uncertainty and distributional perturbations. 

If the estimators $\hat \theta^k = \hat \theta^k(D_1,\ldots,D_n)$ are M-estimators, by Lemma~\ref{lemma:asymptotic-consistency-m-estimators} and Lemma~\ref{lemma:asymptotic-gaussianity-m-estimator} the estimators are asymptotically linear in the sense that 
\begin{equation}\label{eq:asymplinearity}
    \hat{\theta}^k - \theta^k = \frac{1}{n}\sum_{i=1}^{n} \phi^k(D_i) + o_p(\frac{1}{\sqrt{n}}),
\end{equation}
for some deterministic $\theta^k = \arg \min \mathbb{E}[L^k(\theta,D)]$, where $L^k$ is the loss function of the estimator $\theta^k$. $\phi^k$ is referred to as the influence function of $\hat{\theta}^k$ that is assumed to satisfy  $\mathbb{E}[\phi^k(D)] = 0$ and $\text{Var}_\mathbb{P}(\phi^k(D)) \in (0, \infty)$. 
Since $\phi^k(D)$ is square integrable, by Lemma~\ref{lemma:random_weight} the sequence $\sqrt{n}(\hat{\theta}^k - \theta^k)$ converges in distribution to a normal distribution with mean zero and covariance $\delta^2 \text{Var}_{\mathbb{P}}(\phi^k(D))$. We summarize this behaviour of the estimators as the following assumption for the convenience of reference later. 
\begin{assumption}[Asymptotic linearity]\label{assump:a1}
The estimators $\hat\theta^k$, $k=1,\ldots,K$ are asymptotically linear, that is they satisfy equation \eqref{eq:asymplinearity} for influence functions $\phi^k$ with $\mathbb{E}[\phi^k(D)] = 0$ and $0 < \text{Var}_\mathbb{P}(\phi^k(D)) < \infty$.
\end{assumption}

As discussed above, for the case of $M$-estimators, this assumption can be justified via Lemma~\ref{lemma:asymptotic-consistency-m-estimators} and Lemma~\ref{lemma:asymptotic-gaussianity-m-estimator}. We will now formalize the premise that the data analyst considers each of the $\hat \theta^k$ a reasonable estimator for the parameter of interest, $\theta^0$. 
\begin{assumption}[Agreement]\label{ass:agreement}
We have $\theta^k = \theta^{0}$ for $k=1,\ldots,K$.
\end{assumption}

This assumption must be justified with scientific background knowledge. Intuitively, the assumption states that if both sampling uncertainty and distributional uncertainty were negligible, the estimators would agree. In Section~\ref{sec:numerical}, we discuss in a numerical example how the choice of such estimators can be justified. 
If the data scientist does not believe in asymptotic agreement of the estimators, we present conservative confidence intervals in the Appendix, Section~\ref{sec:altway}.

\subsection{Confidence Intervals}\label{sect:ci}
Now let us turn to constructing confidence intervals for $\theta^0$. 
Assume that the data analyst has access to $K$ different estimators $\hat{\theta}^1, \dots, \hat{\theta}^K$. We assume that these estimators are asymptotically linear for estimating $\theta^0$ with influence functions $\phi^1(D), \dots, \phi^K(D)$, i.e.\ that equation~\eqref{eq:asymplinearity} holds. As discussed above, this can be justified for common estimators using the theory in Section~\ref{sec:setting}. For expository simplicity, for now we assume that their influence functions $\phi^1(D), \dots, \phi^K(D)$ are uncorrelated and have the same variance $\sigma^2 > 0$ under $\mathbb{P}$. Later in the section, we discuss how to construct confidence intervals for general cases where influence functions are possibly correlated and have different variances. Since the estimators are asymptotically unbiased, uncorrelated, and have the same variance, as the final estimate we consider the mean of estimators, $\hat{\theta}^{\text{pooled}} = \frac{1}{K}\sum_k \hat{\theta}^k$. In the following, we will investigate the asymptotic behaviour of this estimator.

By Assumption~\ref{assump:a1}, \ref{ass:agreement} and Lemma~\ref{lemma:random_weight}, for $k = 1, \dots, K$,
\begin{equation*}
  \sqrt{n}(\hat \theta^k - \theta^0)_{k=1,\ldots,K} \stackrel{d}{=} \delta \sigma (Z_k)_{k=1,\ldots,K}+ o_P(1),
\end{equation*}
where $Z_k$ are independent standard normal random variables. Thus,
\begin{equation}\label{eq:pooled}
    \sqrt{n}(\hat{\theta}^{\text{pooled}} - \theta^0) \stackrel{d}{=} \delta \sigma \Bar{Z} + o_P(1).
\end{equation}
On the other hand, define the between-estimator variance 
\begin{equation*}
    \hat{\sigma}_{\text{bet}}^2 = \frac{1}{K-1}\sum_{k=1}^K (\hat{\theta}^k - \hat{\theta}^{\text{pooled}})^2.
\end{equation*}
Then, 
\begin{equation}\label{eq:betestvar}
    \hat{\sigma}_{\text{bet}}^2 \stackrel{d}{=} \frac{\delta^2\sigma^2}{n} \frac{1}{K-1}\sum(Z_k - \Bar{Z})^2  + o_P(1/n) \stackrel{d}{=} \frac{\delta^2\sigma^2}{n} \frac{\chi^2(K-1)}{K-1}  +  o_P(1/n),
\end{equation}
where $\chi^2(K-1)$ is a chi-square random variable with $K-1$ degrees of freedom. Let us assume for a moment that $\sigma^2$ is known to the data scientist. In this case, the data scientist may estimate $\hat \delta^2$ via 
\begin{equation*}
    \hat{\delta}^2 := \frac{n \hat{\sigma}_{\text{bet}}^2}{\sigma^2} \xrightarrow[]{d}   \delta^2 \frac{\chi^2(K-1)}{K-1}. 
\end{equation*}
Combining equations \eqref{eq:pooled} and \eqref{eq:betestvar}, we get
\begin{equation*}
    \frac{\hat{\theta}^{\text{pooled}} - \theta^0}{\hat{\sigma}_{\text{bet}}/\sqrt{K-1}} \xrightarrow[]{d} t(K-1),
\end{equation*}
where $t(K-1)$ is a $t$-distributed random variable with $K-1$ degrees of freedom. Note that $\delta$, $\sigma$ cancel out.

Without direct estimation of $\delta$ or $\sigma$, we have an $1-\alpha$ confidence interval of $\theta^0$: 
\begin{equation}\label{eq:simple-ci}
    \hat{\theta}^{\text{pooled}} \pm t_{K-1, 1-\alpha/2} \frac{\hat{\sigma}_{\text{bet}}}{\sqrt{K-1}},
\end{equation}
where $t_{K-1, 1-\alpha/2}$ is the $1-\alpha/2$ quantile of the $t$-distribution with $K-1$ degrees of freedom. Note that the size of the confidence intervals goes to zero with rate $1/\sqrt{n}$ as $\hat \sigma_\text{bet} = O_P(1/\sqrt{n})$. 

Let us make the argument in \eqref{eq:simple-ci} more general. We will now discuss the case where the estimators $\hat \theta^k$ have potentially different asymptotic variances $\text{Var}_{\mathbb{P}}(\phi^k(D))$. Instead of using $\hat \theta^{\text{pooled}} = \frac{1}{K} \sum_k \hat \theta^k$ as the final estimate, we recommend inverse variance weighting. Thus, we first need to estimate $\text{Var}_{\mathbb{P}}(\phi^k(D))$ consistently. While estimating $\text{Var}_{\mathbb{P}}(\phi^k(D))$ is straightforward under i.i.d.\ sampling, we also have to verify that this works in our model class. We estimate $\text{Var}_{\mathbb{P}}(\phi^k(D))$ using plug-in estimators of the influence function $\hat\phi^k(D)$ as
\begin{equation}\label{eq:var}
    \widehat{\text{Var}}_{\mathbb{P}}(\phi^k(D)) = \frac{1}{n} \sum_{i=1}^{n} \Big(\hat\phi^k(D_i) - \frac{1}{n}\sum_{i=1}^{n}\hat\phi^k(D_i)\Big)^2.
\end{equation}
The following proposition shows that $\widehat{\text{Var}}_{\mathbb{P}}(\phi^k(D))$ is a consistent estimator of $\text{Var}_{\mathbb{P}}(\phi^k(D))$.
\begin{proposition}[Consistency of $\widehat{\text{Var}}_{\mathbb{P}}(\phi^k(D))$]\label{prop:cons_var}
Suppose that the $\phi^k(D)$ has finite fourth moments.
Furthermore, suppose that the estimation of the influence function is consistent in the sense that 
\begin{equation}\label{eq:cons_inf}
    \frac{1}{n}\sum_{i=1}^{n}(\hat\phi^k(D_i) - \phi^k(D_i))^2 = o_p(1). 
\end{equation}
Then, $\widehat{\text{Var}}_{\mathbb{P}}(\phi^k(D))$ defined in \eqref{eq:var} satisfies that
$$
\widehat{\text{Var}}_{\mathbb{P}}(\phi^k(D)) = \text{Var}_{\mathbb{P}}(\phi^k(D)) + o_p(1).
$$
\end{proposition}

\begin{remark}[OLS]
Note that equation \eqref{eq:cons_inf} is expected to hold for the plug-in estimators of the influence function under regularity assumptions. Revisiting Example \ref{example:ols}, a plug-in estimator of the influence function is 
\begin{equation*}
    \hat{\phi}^k(D_i) = (\frac{1}{n}\sum_{j=1}^{n}X_j^{S_k}{(X_j^{S_k})}^\intercal)_{1,\bullet}^{-1} X_i^{S_k}(Y_i  - (X_i^{S_k})^\intercal \hat \theta^{k,OLS} ), %
\end{equation*}
where $\hat \theta^{k,OLS}$ is the OLS estimator computed with covariates $X^{S^k}$. One can now justify equation~\eqref{eq:cons_inf} via Lemma~\ref{lemma:asymptotic-consistency-m-estimators}.
\end{remark}

Now we construct asymptotically valid confidence intervals for $\theta^0$ using $K$ different estimators $\hat{\theta}^1, \dots, \hat{\theta}^K$ that are asymptotically linear for estimating $\theta^0$. In the following theorem with Remark \ref{remark:correlated}, influence functions of $K$ different estimators can be correlated and have different variances.

\begin{theorem} (Asymptotic validity of calibrated confidence interval). \label{theorem:newci} 
Suppose Assumption \ref{assump:a1} and \ref{ass:agreement} hold and the influence functions $\phi^1(D), \dots, \phi^K(D)$ are uncorrelated. 
Let $\hat \theta^W = \sum_{k=1}^K \hat \alpha_k \hat \theta^k$ be the inverse-variance weighted estimator where the weights are 
\begin{equation}\label{eq:weights}
    \hat{\alpha}_k =  \frac{\frac{1}{\widehat{\text{Var}}_{\mathbb{P}}(\phi^k(D))} }{\sum_{j=1}^K \frac{1}{\widehat{\text{Var}}_{\mathbb{P}}(\phi^{j}(D))}},
\end{equation}
with $\widehat{\text{Var}}_{\mathbb{P}}(\phi^k(D)) = \text{Var}_{\mathbb{P}}(\phi^k(D)) + o_p(1) $ for $k = 1, \dots, K$. Let $\hat{\sigma}_{\text{bet}}$ be the weighted between-estimator variance defined as  
\begin{equation*}
   \hat \sigma^2_{\text{bet}} = \sum_k \hat \alpha_k (\hat{\theta}^k - \hat \theta^W)^2.
\end{equation*}
Then for any $\alpha \in (0, 1)$, for fixed $K$ and as $n \rightarrow \infty$ we have 
\begin{equation*}
    \mathrm{P} \left(\theta^0 \in \Big[\hat{\theta}^W \pm t_{K-1,1-\alpha/2} \cdot \frac{\hat{\sigma}_{\text{bet}}}{\sqrt{K-1}} \Big]\right) \xrightarrow[]{} 1-\alpha,
\end{equation*}
where $t_{K-1, 1-\alpha/2}$ is the $1-\alpha/2$ quantile of the $t$ distribution with $K-1$ degrees of freedom. To be clear, here we marginalize over both the randomness due to sampling and the randomness due to the distributional perturbation.
\end{theorem}

\begin{remark}[Correlated estimators]\label{remark:correlated}
    In practice, the components of $(\phi^1(D), \dots, \phi^K(D))$ may be correlated. Then, we can apply a linear transformation to the estimators to obtain uncorrelated estimators that are asymptotically unbiased for $\theta^0$. We define the transformation matrix $T_{ij} = \frac{(\hat \Sigma^{-1/2})_{ij}}{\sum_{j'} (\hat \Sigma^{-1/2})_{ij'} }$, where $\hat \Sigma$ is an estimate of the covariance matrix of $\hat \theta^1,\ldots,\hat \theta^K$.  We can then define $(\hat \eta^1,\ldots, \hat \eta^K)^\intercal =  T \cdot (\hat \theta^1,\ldots, \hat \theta^K)^\intercal$. If $\| \hat \Sigma - \Sigma\|_2 = o_p(1)$ and $\Sigma$ is invertible, then the estimators $\hat \eta^1,\ldots,\hat \eta^K$ also satisfy Assumption~\ref{assump:a1} with influence functions that are pairwise uncorrelated. Furthermore, if the $\hat \theta^k$, $k=1,\ldots,K$ satisfy Assumption~\ref{ass:agreement}, then also $\hat \eta^k$, $k=1,\ldots,K$ satisfy Assumption~\ref{ass:agreement}.  
    \end{remark}

\begin{remark}[Meta-analysis on a single data set]\label{remark:random-effect}
The inverse variance-weighted estimate shares some similarity with a meta-analysis model. In traditional meta-analysis, one accounts for the random distributional variability of estimators obtained across different data sets. In contrast, our method accounts for the variability of multiple estimators obtained on a single data set under the random distribution shift model, where multiple estimators for a single target quantity are subject to shared symmetric distribution shifts.  
Thus, the random distribution shift model justifies a particular ``meta-analysis on a single data set". 
The idea that a study can potentially ``replicate itself" has appeared in several communities, see the discussion in Section~\ref{sec:practical-implications}. In some of this literature, there is an emphasis that estimators should be subject to different biases. Analogously, in our approach, the estimators should have different influence functions, which implies that they will be affected differently by the random distribution shift. 
\end{remark}

Below we present an algorithm box that provides a summary of Theorem \ref{theorem:newci} and Remark \ref{remark:correlated} to construct calibrated confidence intervals.

\begin{algorithm}\label{algoritmbox}
\DontPrintSemicolon
  
  \KwInput{$K$ different estimators $\hat{\theta}^1, \dots, \hat{\theta}^K$ and their estimated influence functions $\hat{\phi}^1, \dots, \hat{\phi}^K$}
  \KwOutput{A calibrated confidence interval for $\theta^0$}
  \If{$\phi_1, \dots, \phi_K$ are correlated}
    {
        Estimate the transformation matrix $T$ as in Remark \ref{remark:correlated}.\\
        Let $(\hat{\theta}^1, \dots, \hat{\theta}^K)^\intercal$ $\leftarrow$ $T \cdot (\hat{\theta}^1, \dots, \hat{\theta}^K)^\intercal$.\\
        Let $(\hat{\phi}^1, \dots, \hat{\phi}^K)^\intercal$ $\leftarrow$ $T \cdot (\hat{\phi}^1, \dots, \hat{\phi}^K)^\intercal$.
    }
    Estimate the weights $\hat{\alpha}_k$ for $k = 1, \dots, K$ by Equation \eqref{eq:weights}.\\
    Compute the inverse-variance weighted estimator $\hat{\theta}^W = \sum_{k=1}^K \hat{\alpha}_k\hat{\theta}^k$.\\
    Compute the weighted between-estimator variance $\hat{\sigma}_{\text{bet}}^2 = \sum_{k=1}^K \hat{\alpha}_k(\hat{\theta}^k- \hat{\theta}^W)^2$.\\
    Return a calibrated confidence interval for $\theta^0$ as
    $$\hat{\theta}^W \pm t_{K-1, 1-\alpha/2} \cdot \frac{\hat{\sigma}_{\text{bet}}}{\sqrt{K-1}}.$$
    
\caption{Constructing Calibrated Confidence Intervals}
\end{algorithm}

In some cases, the data analyst may trust one of the estimators $\hat\theta^k$ more than others. For example, the data analyst may be convinced that $\theta^1 = \theta^0$ but may not be sure whether $\theta^k = \theta^0$ for $k \ge 2$. In this case, it is possible to construct variance estimates that are upwardly biased in the sense that the resulting confidence intervals are expected to be conservative. The data analyst may report the confidence interval for $\theta^0$ using $\hat\theta^1$ instead of $\hat\theta^W$ with $\delta$ estimated by computing the between-estimator variance of the remaining $K-1$ estimators. As a result, they would lose one degree of freedom in their confidence intervals. The details can be found in the Appendix, Section  \ref{sec:altway}.

In the final variance estimate, there are two effects that are counteracting each other. Inverse-variance weighting reduces the variance of the final estimate compared to each of the individual estimators $\hat \theta^k$. On the other hand, the new variance formula accounts for distributional uncertainty and thus potentially inflates the variance.

\subsection{Practical Implications for Stability Analyses}\label{sec:practical-implications}

The proposed model not only leads to a recommendation on how to summarize the between-estimator uncertainty in confidence intervals but also lets us give some additional guidance.

First, note that if all estimators $\hat \theta^k$ have similar influence functions, the proposed method will be unstable since in Remark~\ref{remark:correlated} we invert the estimated covariance matrix. This coincides with the following intuition: Reporting that a large number of extremely similar estimators return similar results does not automatically increase the trustworthiness of a result. To corroborate a hypothesis one should have estimators that are susceptible to different sources of biases. In our model, this corresponds to estimators that are not highly correlated. Ideally, the estimators are independent. Similar arguments have appeared in other parts of the literature. For example, \cite{rosenbaum2021replication} writes: ``An observational study has two evidence factors if it provides two comparisons susceptible to different biases that may be combined as if from independent studies of different data by different investigators, despite using the same data twice".

In practice, it may happen that calibrated confidence intervals are very large compared to traditional sampling-based confidence intervals. Apart from the estimation error and small $K$, there are two possible explanations. 

First, it could be that distributional uncertainty is very large. If distributional uncertainty is much larger than sampling uncertainty, conventional (unadjusted) confidence intervals are of limited value. Similar points have been made in different parts of the literature. For example, \citet{meng2018statistical} argues that as the sample size grows, data quality becomes more important than data quantity and that standard confidence intervals have to be inflated to account for issues of data quality. In this vein, distributional confidence intervals can be used as a 
warning signal that we might be in a regime where data quality issues are more pressing than sampling uncertainty.

 Secondly, it could be that the assumptions are violated (that means $\theta^k \neq \theta^{k'}$ for some $k,k'$). If the assumptions are grossly violated, inference will be more conservative.  This is further detailed in the Appendix, Section  \ref{sec:altway}. If the assumptions are correct, inference will be more precise. In other words, the precision of calibrated inference depends on whether Assumption~\ref{ass:agreement} is satisfied or not.

 If the number of estimators $K$ is very small, then there is an inferential price to pay for estimating distributional uncertainty in terms of power. This is reflected in the degrees of freedom of the $t$-distribution.

\section{Simulation Study}\label{sec:numerical}

In this section, we evaluate the performance of the proposed method via a simulation study. The marginal coverages of calibrated confidence intervals and the lengths of calibrated confidence intervals are evaluated on simulated data sets generated by random perturbation models.
In this simulation, we emulate the situation where a data scientist uses linear regression with an adjustment set to estimate a causal effect.

\quad

\textbf{Setup.} The unperturbed distribution of $D = (X, Y)$ with covariates $X \in \mathbb{R}^5$ and response $Y \in \mathbb{R}$ is generated from the following structural causal model \citep{bollen1989,pearl2009}:

\vspace{3mm}

\begin{minipage}{0.4\linewidth}
\begin{center}
   \quad\quad\quad \includegraphics[width=0.75 \linewidth]{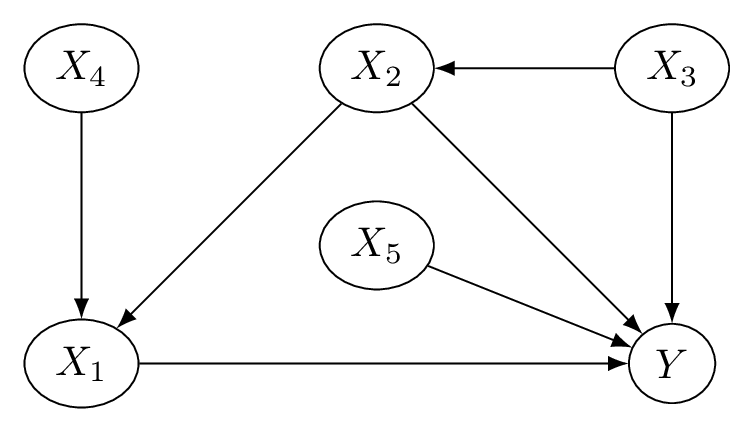} 
\end{center}
\end{minipage}\hfil
\begin{minipage}{0.55\linewidth}
\begin{center}
  \begin{tabular}{l}
    $\epsilon, \epsilon_1, \epsilon_2, X_3, X_4, X_5 \overset{\text{i.i.d}}{\sim} \mathcal{N}(0,1)$, \vspace{0.1cm} \\
    $X_2 \leftarrow X_3 + \epsilon_2$, \vspace{0.1cm}\\
    $X_1 \leftarrow 0.5X_2 + X_4 + \epsilon_1$, \vspace{0.1cm}\\
    $Y \, \, \, \leftarrow X_1 + 0.5X_2 + X_3 + X_5 +  \epsilon$ \\
  \end{tabular}
\end{center}
\end{minipage}

\quad

The goal is to estimate the direct causal effect of $X_1$ on $Y$, which in this setup corresponds to the regression coefficient of $X_1$ in a regression of $Y$ on the set $S = (X_1,X_2)$.  Practitioners often conduct such regressions for different choices of sets $S$ to evaluate the overall stability of the procedure \citep{leamer1983let,oster2019unobservable}.

In this example, the structural causal model can be used to construct multiple valid estimators. We look at the case where the data analyst considers $K=6$ different adjustment sets which all include the confounding variable $X_2$. In this case, $K=6$ different regression-adjusted estimators estimate the same quantity, the direct causal effect of $X_1$ on $Y$, under the unperturbed distribution. We consider following adjustment sets; 
$\{X_1, X_2, X_3\}$, 
$\{X_1, X_2, X_5\}$, $\{X_1, X_2, X_3, X_4\}$, $\{X_1, X_2, X_3, X_5\}$, $\{X_1, X_2, X_4, X_5\}$, $\{X_1, X_2, X_3, X_4, X_5\}$.

We now want to model a random shift between the target and the sampling distribution. We generate randomly perturbed data sets in two ways. First, we adopt the random perturbation model described in Lemma~\ref{lemma:random_weight}. We partition the support of the joint distribution of $X$ and $Y$ into $m^{p+1}$ equal probability bins and perturb the probability of each bin with i.i.d. random weights $Z\cdot W$ where $W \sim \text{Gamma}(1,1)$ and $Z \sim \text{Ber}(1/m^p)$. For sufficiently large $m$, this procedure can be seen as randomly selecting $m$ bins out of $m^{p+1}$ bins and perturbing the probability of each selected bin with i.i.d random weights $W \sim \text{Gamma}(1,1)$. In our simulations, we generate $n$ i.i.d. data points $D_1, \dots, D_n$ from this randomly perturbed distribution. The strength of the perturbation is given as $\delta^2 \approx 1 + 2\cdot n/m$. Secondly, we employ the random perturbation model described in Example~\ref{example:samp} in the Appendix. Here, we sample $m$ data points from the original distribution and let randomly perturbed distribution be the empirical distribution of $m$ samples. The strength of the perturbation is given as $\delta^2 \approx 1 + n/m$.

Our method is carried out for sample sizes $n = 200, 500, 1000$ and for $m = 200, 500, 1000$, which determines the strength of the perturbation, each with $N = 1000$ replicates. In each replicate, we generate $n$ samples from the randomly perturbed distribution, obtain $K = 6$ different regression-adjusted estimators from the perturbed data set, and construct a calibrated $(1-\alpha)$ confidence interval using the inverse-variance weighted estimator according to Algorithm \ref{algoritmbox}. We then evaluate the marginal coverage and length of the calibrated confidence interval and non-calibrated confidence intervals for each regression-adjusted estimator. While the direct estimation of $\hat{\delta^2}$ is not required in our calibrated confidence intervals, we also include simulation results on the accuracy of $\hat{\delta}^2$ in the Appendix \ref{sec:add_simulations}.

\subsection{The Marginal Coverages of Calibrated Confidence Intervals}

The marginal coverages of calibrated confidence intervals and non-calibrated confidence intervals are given in Figure~\ref{fig:coverage}. We see that calibrated confidence intervals have much improved coverage compared to non-calibrated confidence intervals, especially when $n$ is large and $m$ is small as the variance due to distributional perturbations dominates the marginal variance. In Appendix, Section \ref{sec:add_simulations}, we additionally look at the case where the data analyst considers $K=8$ different adjustment sets including two additional sets $\{X_1, X_2\}$ and
$\{X_1, X_2, X_4\}$. In this case, some estimators are highly correlated, meaning that intuitively they are not distinct sources of evidence. This results in slight undercoverage, which highlights our advice that ideally one should use uncorrelated estimators to calibrate inference.

\begin{figure}
    \centering
    \includegraphics[scale = 0.8]{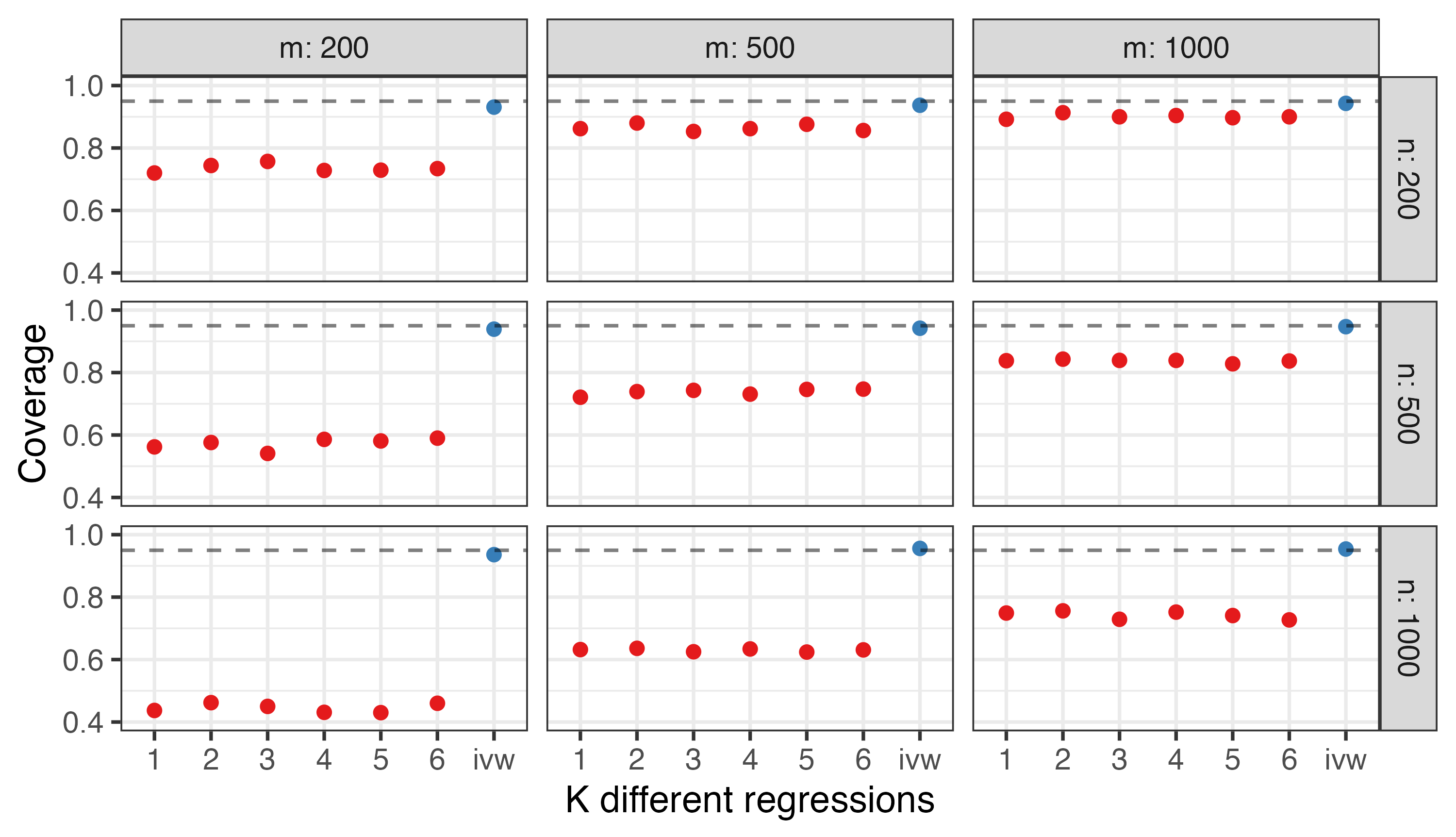}
   \includegraphics[scale = 0.8]{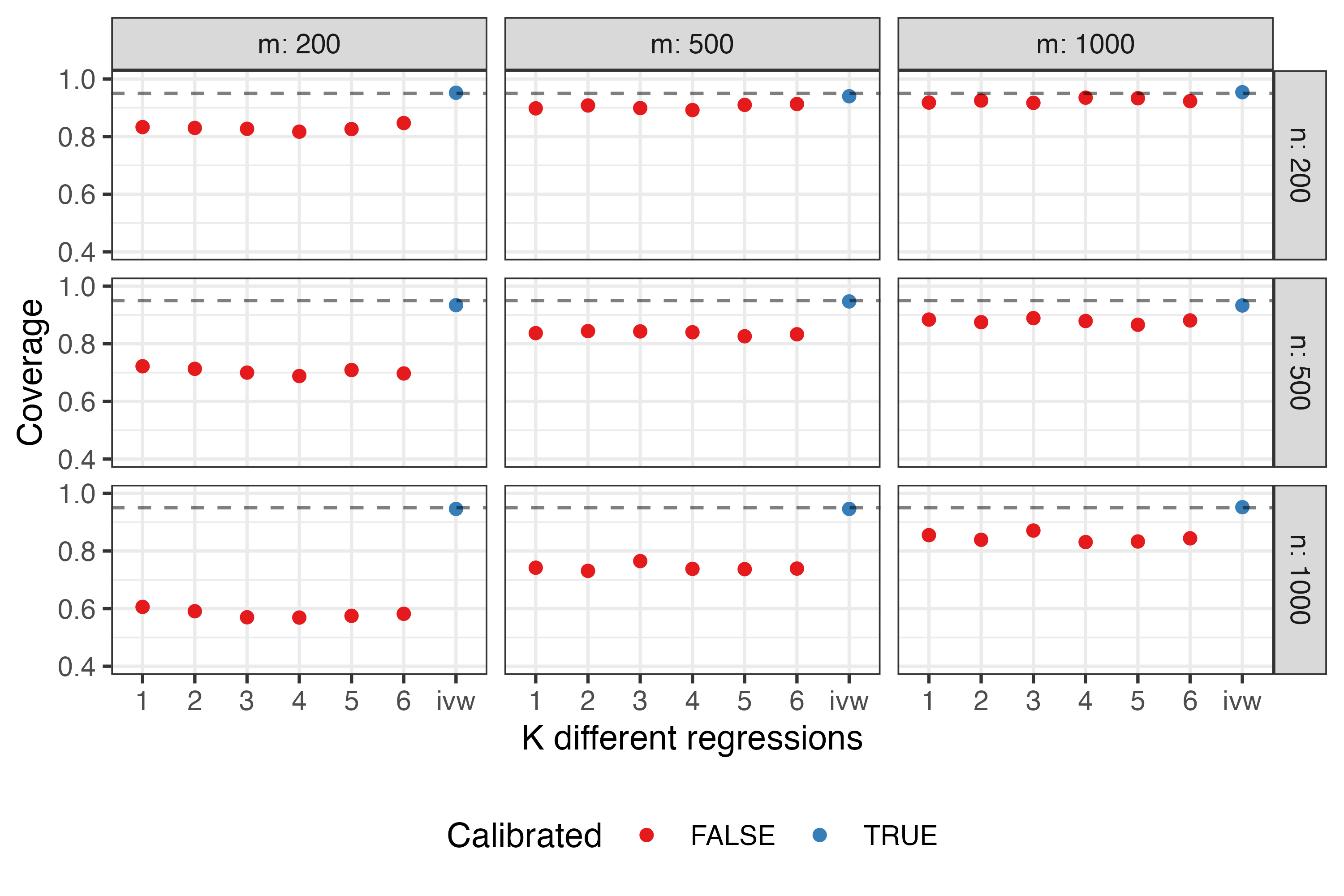}
    \caption{Marginal coverages of calibrated confidence intervals: The panel above shows the results under the perturbation model in Lemma~\ref{lemma:random_weight} and the panel below shows the results under the perturbation model in Example~\ref{example:samp} in the Appendix. Marginal coverages of non-calibrated confidence intervals for each regression-adjusted estimator and calibrated confidence intervals for the inverse-variance weighted estimator are presented for $m = 200, 500, 1000$ and $n = 200, 500, 1000$. The strength of the perturbation is given as $\delta^2 \approx 1 + n/m$. The dashed lines indicate the nominal coverage 0.95.}
    \label{fig:coverage}
\end{figure}

\subsection{The Lengths of Calibrated Confidence Intervals}

The boxplots of lengths of calibrated confidence intervals and non-calibrated confidence intervals are given in Figure \ref{fig:cilengths}. 
Figure \ref{fig:cilengths} indicates that, perhaps surprisingly, calibrated confidence intervals can have even smaller lengths than non-calibrated confidence intervals, despite accounting for both distributional uncertainty and sampling uncertainty. 
This is due to inverse-variance weighting, which reduces the variance of the final estimate in comparison to each of the individual estimators. Note that the proportion of outliers marked as circles in boxplots is typically less than $5\%$ for each boxplot. The distribution of the lengths of calibrated confidence intervals has a heavier tail than that of non-calibrated ones, as the former follows the square root of a scaled chi-square distribution with $K-1$ degrees of freedom.

\begin{figure}
    \centering
    \includegraphics[scale = 0.92]{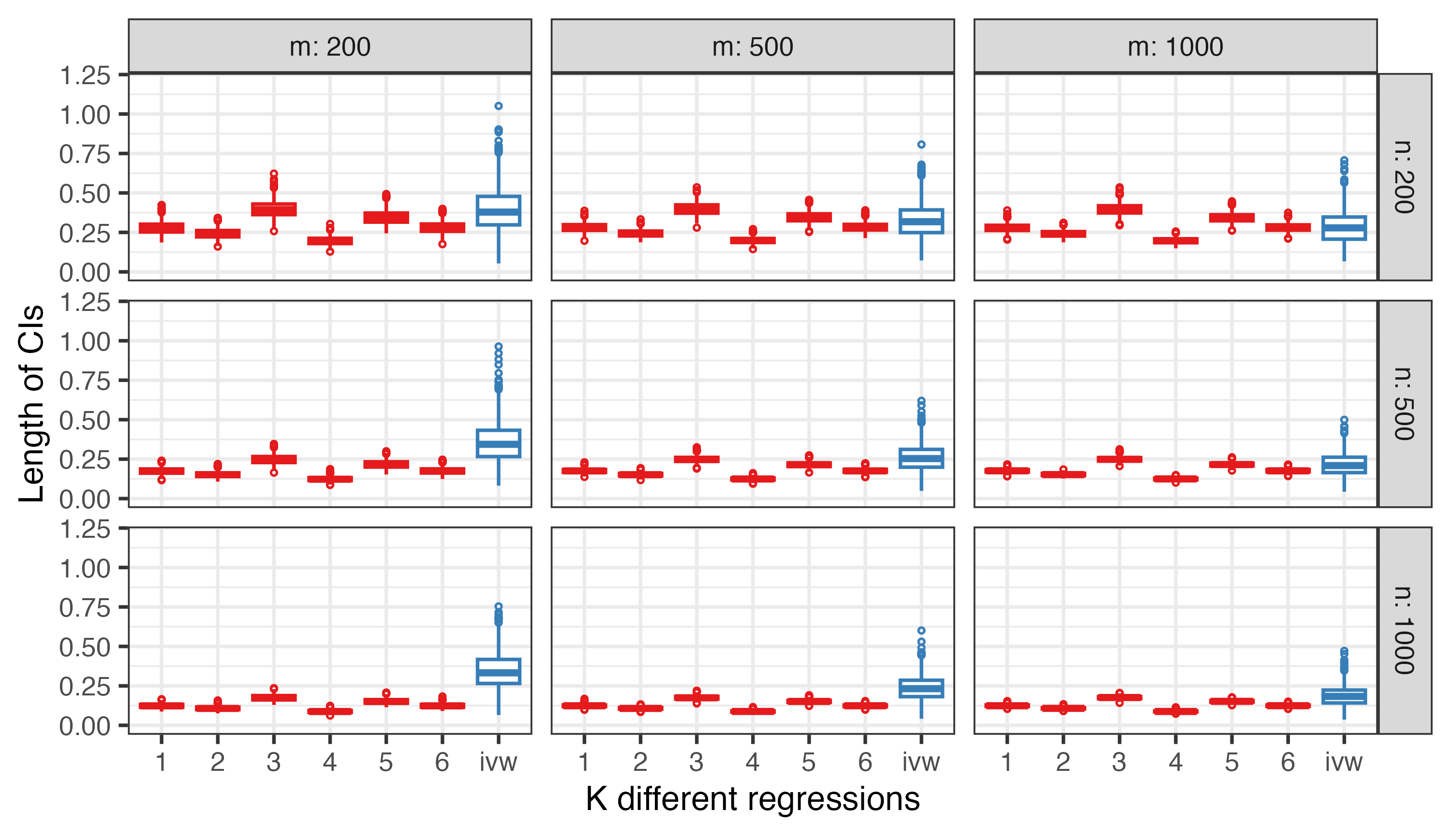}
    \includegraphics[scale = 0.92]{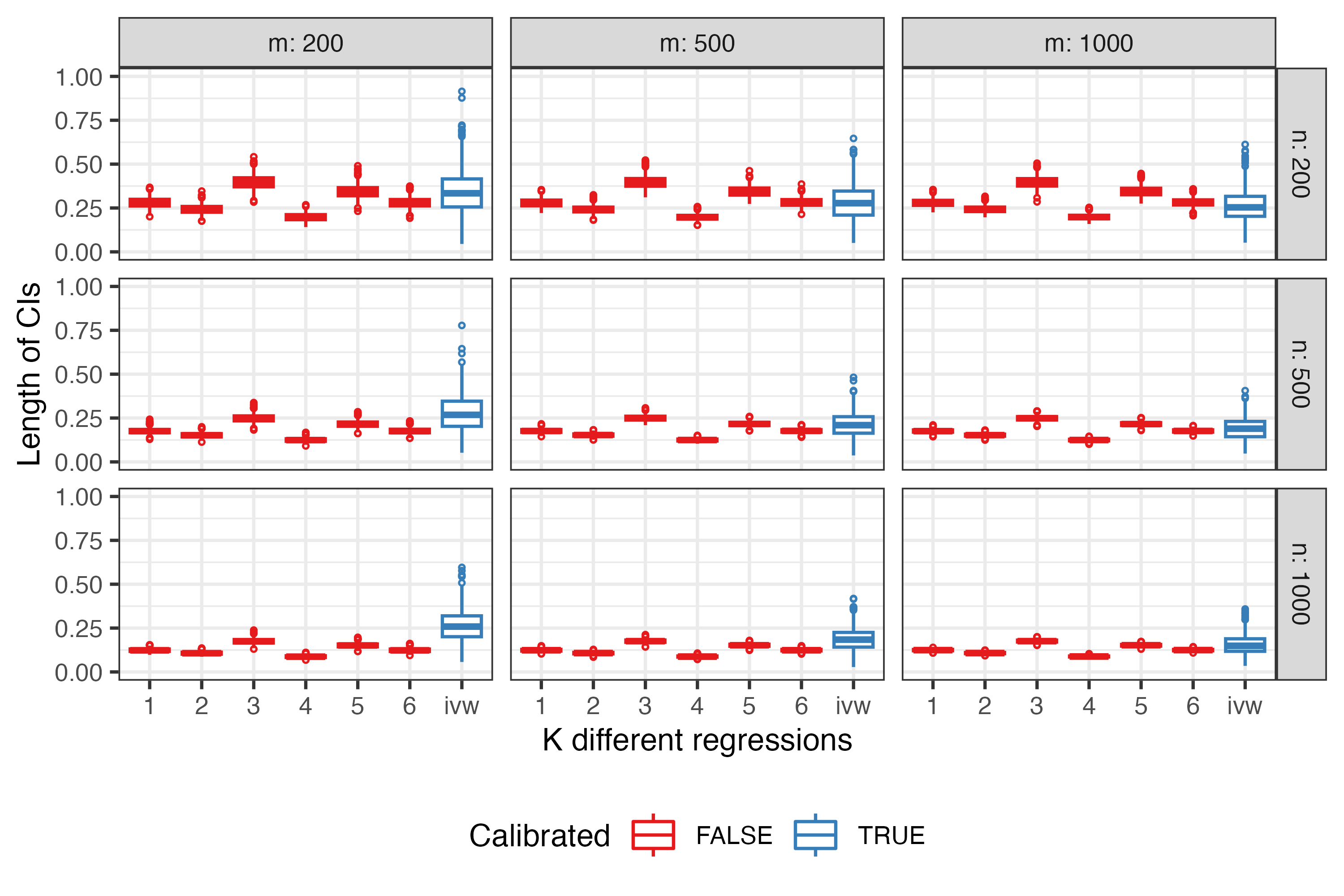}
    \caption{Lengths of calibrated confidence intervals: The panel above shows the results under the perturbation model in Lemma~\ref{lemma:random_weight} and the panel below shows the results under the perturbation model in Example~\ref{example:samp} in the Appendix. Boxplots of lengths of $N = 1000$ non-calibrated confidence intervals for each regression-adjusted estimator and calibrated confidence intervals for the inverse-variance weighted estimator are presented for $m = 200, 500, 1000$ and $n = 200, 500, 1000$.}
    \label{fig:cilengths}
\end{figure}

\section{Real-World Data Analysis}\label{sec:real-world}

Ultimately, the goal of our procedure is to increase stability and trustworthiness of decision-making.
In this section, we demonstrate that our method can improve stability on a real data set. We will see that even in situations without distributional perturbations, the proposed method can increase stability of decision-making. The data set \citep{studentdata} was collected by using school reports and questionnaires to estimate final grades of students in secondary education of two Portuguese schools. The data attributes include student grades, demographic, social and school related features. It is available at the UCI machine learning repository \citep{Dua:2017}. We adopt 20 covariates in the data set. The response $Y$ is the final year grade in Portuguese language. There are 649 students in total.

The goal is to determine the relative importance of $L = 7$ selected binary covariates: 1) parents' cohabitation status, 2) whether the student received extra educational support from the school, 3) whether the student received family educational support, 4) whether the student is in a relationship, 5) whether the student had extra paid classes within the course subject, 6) whether the student's mother had secondary or higher education, and 7) whether the student's father had secondary or higher education. The relative importance is determined by the rank order of the covariates' effect sizes in a linear regression. 

In the simulation setup, we aim to emulate a situation where as baseline the analyst has several reasonable choices to estimate a certain target quantity, and makes these decisions randomly. On the other hand, as a comparison, the analyst aggregates the estimators and conducts uncertainty quantification as proposed above. Ideally, the different estimators  are driven by scientific background knowledge, as in the previous section. Here, for illustration purposes, we investigate the extreme case where background knowledge is very limited, that is, the statistician does not have strong preferences regarding which covariates to include in the regression.

Suppose we are given multiple sets of covariates, all containing the 7 binary covariates of our interest. We consider the following two methods. In method 1, a statistician randomly chooses one of the sets of covariates, performs a linear regression, and ranks the effect sizes of 7 covariates. In method 2, a statistician employs our proposed method. In particular, they perform linear regressions with multiple sets of covariates and for each covariate, calculate an inverse-variance weighted estimator and its effect size in consideration of distributional perturbations as described in Section~\ref{sec:calibrated}. Then, they rank these effect sizes. Note that we use the additional constraint $\hat{\delta} = \max(\hat{\delta}, 1)$ in our implementation.

We evaluate the two methods' stability in ranking effect sizes. To evaluate method $i$, we randomly split the data set into two, perform method $i$ on each split, and compare the rankings resulting from each split. To measure the stability of the ranking, we compute the set similarity measure between  $S_{1,\ell} = \{$Top $\ell$ covariates by the effect size on split 1$\}$ and $S_{2,\ell} = \{$Top $\ell$ covariates by the effect size on split 2$\}$ for each $\ell = 1, \dots, L=7$ as $|S_{1,\ell} \cap S_{2,\ell}|/L$. We repeat this procedure $N = 500$ times and record the average set similarity measure. In each replicate, we randomly generate $K = 10, 20$ sets of covariates that include the 7 covariates of our interest. The results can be found in Table~\ref{table:ranking}. Overall, we see our method (Method 2) improves the stability of the ranking, notably outperforming Method 1 for $\ell=1, 2, 3$. Note that the method 1 gives slightly worse results than random guessing for small $\ell$. One possible explanation is that sample splitting introduces small negative correlations between splits: If a regression coefficient is close to zero on the entire data set and on one split by chance the coefficient is large, then the coefficient is expected to be small on the other split. 

\begin{table}[ht]
\centering
\begin{tabular}{rrrrrrrr}
  \hline
 $\ell$ & 1 & 2 & 3 & 4 & 5 & 6 & 7 \\ 
  \hline
Method 1 (Non-Calibrated, $K = 10$) & 0.102 & 0.203 & 0.407 & 0.648 & 0.817 & 0.898 & 1.000 \\ 
Method 2 (Calibrated, $K = 10$) & 0.210 & 0.296 & 0.449 & 0.658 & 0.828 & 0.912 & 1.000 \\ 
   \hline
\end{tabular}
\\\quad\\\quad\\
\centering
\begin{tabular}{rrrrrrrr}
  \hline
 $\ell$ & 1 & 2 & 3 & 4 & 5 & 6 & 7 \\ 
  \hline
Method 1 (Non-Calibrated, $K = 20$) & 0.090 & 0.203 & 0.417 & 0.659 & 0.817 & 0.893 & 1.000 \\ 
Method 2 (Calibrated, $K = 20$) & 0.235 & 0.313 & 0.445 & 0.679 & 0.845 & 0.912 & 1.000 \\ 
   \hline
\end{tabular}
\caption{The stability of the ranking: The table above shows results with $K = 10$ sets of covariates and the table below shows results with $K = 20$ sets of covariates. Mean over $N=500$ iterations of the computed set similarity measure between $S_{1,\ell}$ and $S_{2,\ell}$ for each $\ell = 1, \dots, 7$ is provided for each method. }
    \label{table:ranking}
\end{table}

Additionally, we compare lengths of calibrated and non-calibrated confidence intervals for each selected binary covariate using the full data set. 
From the results provided in Figure~\ref{fig:CIlength}, one can see that our method is not so conservative given that we are adjusting confidence intervals with a scaling factor $\hat{\delta}$. Moreover, the variance of the length of calibrated confidence intervals tends to decrease as we increase the number of sets of covariates from $K = 10$ to $K = 20$. For a more detailed look at the distribution of $\hat{\delta}$, the histograms of the confidence interval lengths are provided in Section~\ref{sec:add-hist} of the Appendix.

\begin{figure}[h]
     \centering
     \includegraphics[scale = 0.7]{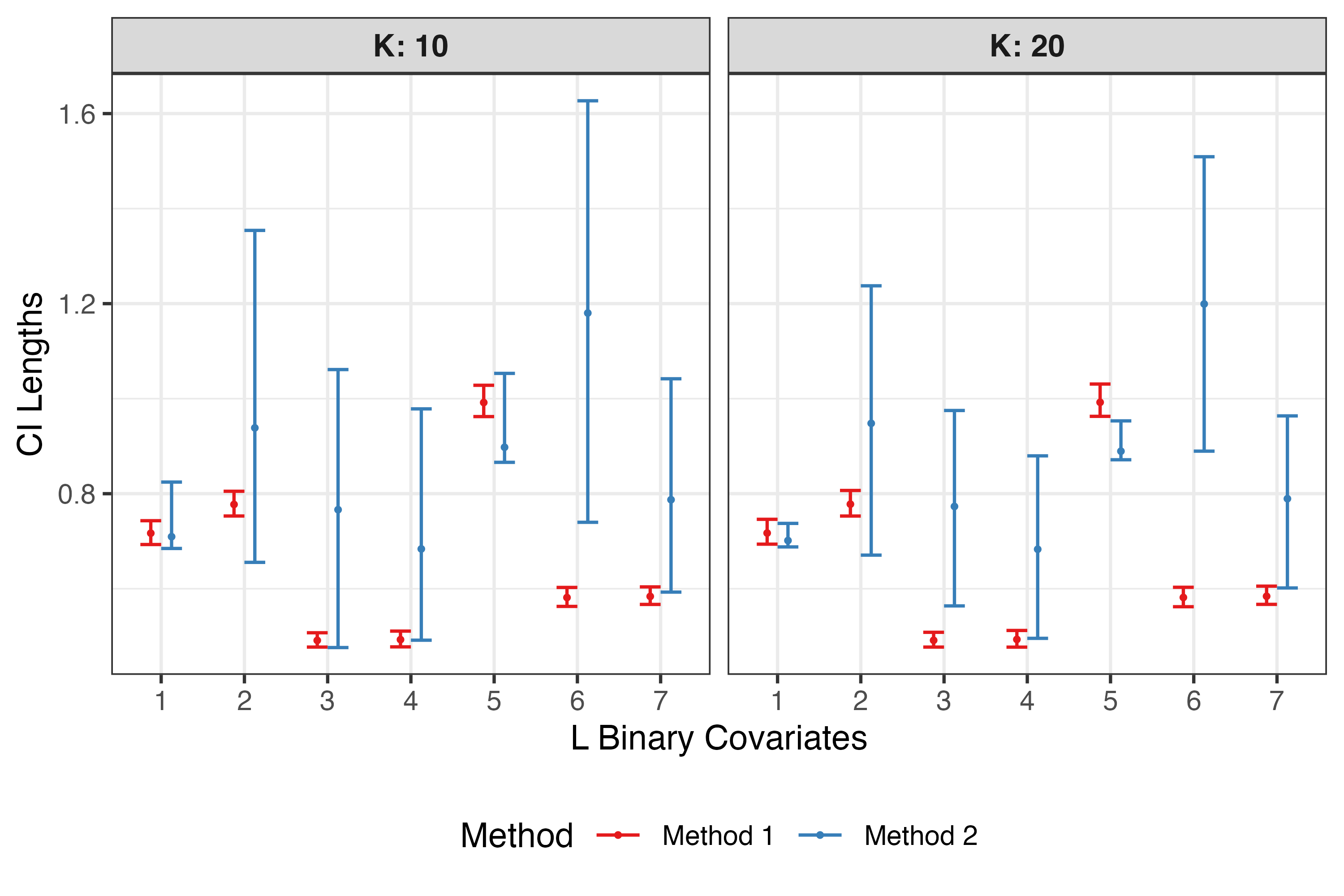}
     \caption{The lengths of calibrated and non-calibrated confidence intervals: Mean, $2.5\%$ and $97.5\%$ quantiles of the lengths of calibrated (Method 2) and non-calibrated (Method 1) confidence intervals for each selected binary covariate over $N=1000$ iterations are provided for $K = 10$ and $K = 20$.}
     \label{fig:CIlength}
 \end{figure}

\section{Discussion}\label{sec:discussion}

In practice, data analysts often compute not just one estimator but multiple estimators for a single target quantity. For example, in causal inference, practitioners often consider multiple strategies to estimate the treatment effect. They could compute multiple regression-adjusted estimators for different choices of adjustment sets (see Example~\ref{example:ols} and Section~\ref{sec:numerical}).  
If they believe that the treatment effect is homogeneous, they can derive several reasonable estimators for different subgroups of observations. 

Often, it is recommended to study the estimator-to-estimator variability between sensible choices of estimators. If the estimator-to-estimator variability is high, then the analyst might have reason to not trust the estimates. In these cases, such stability investigations may be more informative than traditional $p$-values or confidence regions. This warrants an investigation of the theoretical properties of this practice. Does this practice have any guarantees and if so, which? Can we integrate this type of stability analysis into statistical inference?

We study a variant of this procedure from a distributional perspective. The data analyst may have access to multiple estimators, each purportedly estimating the same quantity, as justified by scientific background knowledge. In this context, estimator-to-estimator variability can be leveraged to scale confidence intervals. We show that these scaled confidence intervals account for both sampling uncertainty and distributional uncertainty within an isotropic perturbation model. Such uncertainty quantification seems desirable, especially in settings where the sampling uncertainty is of similar or lower order than other types of uncertainty. 

This isotropic perturbation model is motivated by empirical phenomena in Figure~\ref{fig:qqplots}. It assumes that the distribution shift is a superposition of many small random distributional changes. 

The calibration procedure is not meant to replace existing methods that address confounding or selection bias via bias corrections, regression adjustment, or weighting procedures. Instead, our procedure can be used in conjunction with these methods as a second step that ``calibrates" the confidence intervals.

The isotropic perturbation model is a strong assumption, but it is a weaker assumption than assuming that the data is i.i.d.\ from the target distribution, which is commonly made. Thus, the proposed calibration procedure works under strictly less assumptions on the data generating process than the most common inferential strategy. Instead of relying on i.i.d.\ sampling from $\mathbb{P}$, inference in the proposed model is based on a symmetry assumption and on scientific background knowledge for finding multiple reasonable estimators.

Of course, in practice perturbations might affect parts of the distribution differently. In such cases the proposed method can potentially have over-coverage or under-coverage. Looking forward, it would be desirable to extend the isotropic perturbation model (which has only one single parameter $\delta$) to more flexible models that depend on multiple parameters. Such perturbation models would allow training different uncertainty models for different parts of the distribution, potentially leading to more realistic and flexible uncertainty quantification than existing approaches.

Furthermore, while our method currently operates with a single data set, a promising extension involves exploring scenarios with multiple perturbed data sets. When having access to multiple perturbed data sets, we can model the different data sets arising from the perturbed data generating distributions. Some first discussions about using and modeling multiple perturbed data sets can be found in \cite{rothenhausler2023distributionally} and \cite{bansak2023learning}.

A companion R package, calinf, is available at \url{https://github.com/rothenhaeusler/calinf}. Our package allows to draw data under the distributional uncertainty model and calibrate inference in generalized linear models. We provide an example of calibrated inference where the data analyst computes regression-adjusted estimators for different choices of adjustment sets. If multiple estimators are not available, it is also possible to estimate $\delta$ using other types of scientific background knowledge. On the GitHub page, we discuss an example where the data analyst has background knowledge of population parameters.

\section{Acknowledgments}

We thank the AE, three anonymous reviewers, Bin Yu, Peter B\"uhlmann, Guido Imbens, Xiao-Li Meng, Kevin Guo, Suyash Gupta, Ying Jin, and James Yang for helpful feedback and inspiring discussions. We are grateful for the support of the Stanford Institute for Human-Centered Artificial Intelligence (HAI) and the Dieter Schwarz Foundation.

\newpage
\bibliography{sample}

\begin{thebibliography}{49}
\providecommand{\natexlab}[1]{#1}
\providecommand{\url}[1]{\texttt{#1}}
\expandafter\ifx\csname urlstyle\endcsname\relax
  \providecommand{\doi}[1]{doi: #1}\else
  \providecommand{\doi}{doi: \begingroup \urlstyle{rm}\Url}\fi

\bibitem[Bansak et~al.(2024)Bansak, Paulson, and Rothenh{\"a}usler]{bansak2023learning}
K.~Bansak, E.~Paulson, and D.~Rothenh{\"a}usler.
\newblock Learning under random distributional shifts.
\newblock \emph{To appear in Proceedings of the 27th International Conference on Artificial Intelligence and Statistics (AISTATS)}, 2024.

\bibitem[Ben-Tal et~al.(2013)Ben-Tal, Den~Hertog, De~Waegenaere, Melenberg, and Rennen]{ben2013robust}
A.~Ben-Tal, D.~Den~Hertog, A.~De~Waegenaere, B.~Melenberg, and G.~Rennen.
\newblock Robust solutions of optimization problems affected by uncertain probabilities.
\newblock \emph{Management Science}, 59\penalty0 (2):\penalty0 341--357, 2013.

\bibitem[Bollen(1989)]{bollen1989}
K.~Bollen.
\newblock \emph{Structural Equations with latent variables}.
\newblock John Wiley \& Sons, 1989.

\bibitem[Box(1980)]{box1980sampling}
G.~E.~P. Box.
\newblock Sampling and bayes' inference in scientific modelling and robustness.
\newblock \emph{Journal of the Royal Statistical Society: Series A}, 143\penalty0 (4):\penalty0 383--404, 1980.

\bibitem[Breiman(1996)]{breiman1996bagging}
L.~Breiman.
\newblock Bagging predictors.
\newblock \emph{Machine learning}, 24\penalty0 (2):\penalty0 123--140, 1996.

\bibitem[Breiman(2001)]{breiman2001random}
L.~Breiman.
\newblock Random forests.
\newblock \emph{Machine learning}, 45\penalty0 (1):\penalty0 5--32, 2001.

\bibitem[B{\"u}hlmann(2020)]{buhlmann2020invariance}
P.~B{\"u}hlmann.
\newblock Invariance, causality and robustness.
\newblock \emph{Statistical Science}, 35\penalty0 (3):\penalty0 404--426, 2020.

\bibitem[Chiappori et~al.(2012)Chiappori, Oreffice, and Quintana-Domeque]{chiappori2012fatter}
P.~Chiappori, S.~Oreffice, and C.~Quintana-Domeque.
\newblock Fatter attraction: anthropometric and socioeconomic matching on the marriage market.
\newblock \emph{Journal of Political Economy}, 120\penalty0 (4):\penalty0 659--695, 2012.

\bibitem[Cornfield et~al.(1959)Cornfield, Haenszel, Hammond, Lilienfeld, Shimkin, and Wynder]{cornfield1959smoking}
J.~Cornfield, W.~Haenszel, E.~C. Hammond, A.~Lilienfeld, M.~Shimkin, and E.~Wynder.
\newblock Smoking and lung cancer: recent evidence and a discussion of some questions.
\newblock \emph{Journal of the National Cancer institute}, 22\penalty0 (1):\penalty0 173--203, 1959.

\bibitem[Cortez and Silva(2008)]{studentdata}
P.~Cortez and A.~Silva.
\newblock Using data mining to predict secondary school student performance.
\newblock \emph{Proceedings of 5th Future Business Technology Conference (FUBUTEC 2008)}, pages 5--12, 2008.

\bibitem[Denzin(1970)]{denzin1970research}
N.~K. Denzin.
\newblock \emph{The research act: A theoretical introduction to sociological methods}.
\newblock Aldine Publishing Company, 1970.

\bibitem[Dheeru and Graff(2017)]{Dua:2017}
D.~Dheeru and C.~Graff.
\newblock {UCI} machine learning repository, 2017.
\newblock URL \url{http://archive.ics.uci.edu/ml}.

\bibitem[Drautzburg(2019)]{drautzburg2019}
T.~Drautzburg.
\newblock Why are recessions so hard to predict? random shocks and business cycles.
\newblock Technical report, Federal Reserve Bank of Philadelpha, 2019.

\bibitem[Duchi et~al.(2021)Duchi, Glynn, and Namkoong]{duchi2021statistics}
J.~C. Duchi, P.~W. Glynn, and H.~Namkoong.
\newblock Statistics of robust optimization: A generalized empirical likelihood approach.
\newblock \emph{Mathematics of Operations Research}, 46\penalty0 (3):\penalty0 946--969, 2021.

\bibitem[Dudley(2018)]{dudley2018real}
R.~M. Dudley.
\newblock \emph{Real Analysis and Probability}.
\newblock CRC Press, 2018.

\bibitem[Freedman(1991)]{freedman1991statistical}
D.~Freedman.
\newblock Statistical models and shoe leather.
\newblock \emph{Sociological methodology}, pages 291--313, 1991.

\bibitem[Heinze-Deml and Meinshausen(2021)]{heinze2021conditional}
C.~Heinze-Deml and N.~Meinshausen.
\newblock Conditional variance penalties and domain shift robustness.
\newblock \emph{Machine Learning}, 110\penalty0 (2):\penalty0 303--348, 2021.

\bibitem[Huber(1981)]{huber81}
P.~Huber.
\newblock \emph{Robust Statistics}.
\newblock Wiley, New York, 1981.

\bibitem[Imbens and Rubin(2015)]{imbens2015causal}
G.~Imbens and D.~Rubin.
\newblock \emph{Causal Inference in Statistics, Social, and Biomedical Sciences}.
\newblock Cambridge University Press, 2015.

\bibitem[Karmakar et~al.(2019)Karmakar, French, and Small]{karmakar2019integrating}
B.~Karmakar, B.~French, and D.~S. Small.
\newblock Integrating the evidence from evidence factors in observational studies.
\newblock \emph{Biometrika}, 106\penalty0 (2):\penalty0 353--367, 2019.

\bibitem[LaLonde(1986)]{lalonde1986evaluating}
R.~LaLonde.
\newblock Evaluating the econometric evaluations of training programs with experimental data.
\newblock \emph{The American Economic Review}, pages 604--620, 1986.

\bibitem[Leamer(1983)]{leamer1983let}
E.~Leamer.
\newblock Let's take the con out of econometrics.
\newblock \emph{The American Economic Review}, 73\penalty0 (1):\penalty0 31--43, 1983.

\bibitem[Meng(2018)]{meng2018statistical}
X.~Meng.
\newblock Statistical paradises and paradoxes in big data (i): Law of large populations, big data paradox, and the 2016 us presidential election.
\newblock \emph{The Annals of Applied Statistics}, 12\penalty0 (2):\penalty0 685--726, 2018.

\bibitem[Munaf{\`o} and Smith(2018)]{munafo2018repeating}
M.~Munaf{\`o} and G.~Smith.
\newblock Repeating experiments is not enough.
\newblock \emph{Nature}, 553\penalty0 (7689):\penalty0 399--401, 2018.

\bibitem[Oster(2019)]{oster2019unobservable}
E.~Oster.
\newblock Unobservable selection and coefficient stability: Theory and evidence.
\newblock \emph{Journal of Business \& Economic Statistics}, 37\penalty0 (2):\penalty0 187--204, 2019.

\bibitem[Patel et~al.(2015)Patel, Burford, and Ioannidis]{patel2015assessment}
C.~J. Patel, B.~Burford, and J.~P.~A. Ioannidis.
\newblock Assessment of vibration of effects due to model specification can demonstrate the instability of observational associations.
\newblock \emph{Journal of clinical epidemiology}, 68\penalty0 (9):\penalty0 1046--1058, 2015.

\bibitem[Pearl(2009)]{pearl2009}
J.~Pearl.
\newblock \emph{Causality: Models, reasoning, and inference}.
\newblock Cambridge University Press, 2nd edition, 2009.

\bibitem[Peters et~al.(2016)Peters, B{\"u}hlmann, and Meinshausen]{peters2016causal}
J.~Peters, P.~B{\"u}hlmann, and N.~Meinshausen.
\newblock Causal inference by using invariant prediction: Identification and confidence intervals.
\newblock \emph{Journal of the Royal Statistical Society, Series B}, 78\penalty0 (5):\penalty0 947--1012, 2016.

\bibitem[Pfister et~al.(2021)Pfister, Williams, Peters, Aebersold, and B{\"u}hlmann]{pfister2021stabilizing}
N.~Pfister, E.~G. Williams, J.~Peters, R.~Aebersold, and P.~B{\"u}hlmann.
\newblock Stabilizing variable selection and regression.
\newblock \emph{The Annals of Applied Statistics}, 15\penalty0 (3):\penalty0 1220--1246, 2021.

\bibitem[Prochazka et~al.(2022)Prochazka, Parilakova, Rudolf, Bruk, Jungwirthova, Fejtova, Masaryk, and Vaculik]{prochazka2022pain}
J.~Prochazka, K.~Parilakova, P.~Rudolf, V.~Bruk, R.~Jungwirthova, S.~Fejtova, R.~Masaryk, and M.~Vaculik.
\newblock Pain as social glue: A preregistered direct replication of experiment 2 of bastian et al. (2014).
\newblock \emph{Psychological Science}, 33\penalty0 (3):\penalty0 463--473, 2022.
\newblock \doi{10.1177/09567976211040745}.

\bibitem[Rojas-Carulla et~al.(2018)Rojas-Carulla, Sch{\"o}lkopf, Turner, and Peters]{rojas2015causal}
M.~Rojas-Carulla, B.~Sch{\"o}lkopf, R.~Turner, and J.~Peters.
\newblock Causal transfer in machine learning.
\newblock \emph{Journal of Machine Learning Research}, 19\penalty0 (36):\penalty0 1--34, 2018.

\bibitem[Rosenbaum(1987)]{rosenbaum1987role}
P.~Rosenbaum.
\newblock The role of a second control group in an observational study.
\newblock \emph{Statistical Science}, 2\penalty0 (3):\penalty0 292--306, 1987.

\bibitem[Rosenbaum(2010)]{rosenbaum2010evidence}
P.~Rosenbaum.
\newblock Evidence factors in observational studies.
\newblock \emph{Biometrika}, 97\penalty0 (2):\penalty0 333--345, 2010.

\bibitem[Rosenbaum(2021)]{rosenbaum2021replication}
P.~Rosenbaum.
\newblock \emph{Replication and Evidence Factors in Observational Studies}.
\newblock CRC Press, 2021.

\bibitem[Rosenbaum and Rubin(1983)]{rosenbaum1983assessing}
P.~Rosenbaum and D.~Rubin.
\newblock Assessing sensitivity to an unobserved binary covariate in an observational study with binary outcome.
\newblock \emph{Journal of the Royal Statistical Society: Series B (Methodological)}, 45\penalty0 (2):\penalty0 212--218, 1983.

\bibitem[Rothenh{\"a}usler and B{\"u}hlmann(2023)]{rothenhausler2023distributionally}
D.~Rothenh{\"a}usler and P.~B{\"u}hlmann.
\newblock Distributionally robust and generalizable inference.
\newblock \emph{Statistical Science}, 38\penalty0 (4):\penalty0 527--542, 2023.

\bibitem[Rothenh{\"a}usler et~al.(2015)Rothenh{\"a}usler, Heinze, Peters, and Meinshausen]{rothenhausler2015backshift}
D.~Rothenh{\"a}usler, C.~Heinze, J.~Peters, and N.~Meinshausen.
\newblock Backshift: Learning causal cyclic graphs from unknown shift interventions.
\newblock In \emph{Advances in Neural Information Processing Systems 29 (NIPS)}, pages 1513--1521, 2015.

\bibitem[Rothenh{\"a}usler et~al.(2021)Rothenh{\"a}usler, Meinshausen, B{\"u}hlmann, and Peters]{rothenhausler2021anchor}
D.~Rothenh{\"a}usler, N.~Meinshausen, P.~B{\"u}hlmann, and J.~Peters.
\newblock Anchor regression: Heterogeneous data meet causality.
\newblock \emph{Journal of the Royal Statistical Society: Series B (Statistical Methodology)}, 83\penalty0 (2):\penalty0 215--246, 2021.

\bibitem[Sch\"{o}lkopf et~al.(2012)Sch\"{o}lkopf, Janzing, Peters, Sgouritsa, Zhang, and Mooij]{schoelkopf2012}
B.~Sch\"{o}lkopf, D.~Janzing, J.~Peters, E.~Sgouritsa, K.~Zhang, and J.~Mooij.
\newblock On causal and anticausal learning.
\newblock In \emph{Proceedings of the 29th International Conference on Machine Learning ({ICML})}, pages 1255--1262, 2012.

\bibitem[Skene et~al.(1986)Skene, Shaw, and Lee]{skene1986bayesian}
A.~Skene, J.~Shaw, and T.~Lee.
\newblock Bayesian modelling and sensitivity analysis.
\newblock \emph{Journal of the Royal Statistical Society}, 35\penalty0 (2):\penalty0 281--288, 1986.

\bibitem[Srivastava et~al.(2014)Srivastava, Hinton, Krizhevsky, Sutskever, and Salakhutdinov]{srivastava2014dropout}
N.~Srivastava, G.~Hinton, A.~Krizhevsky, I.~Sutskever, and R.~Salakhutdinov.
\newblock Dropout: a simple way to prevent neural networks from overfitting.
\newblock \emph{The Journal of Machine Learning Research}, 15\penalty0 (1):\penalty0 1929--1958, 2014.

\bibitem[Steegen et~al.(2016)Steegen, Tuerlinckx, Gelman, and Vanpaemel]{steegen2016increasing}
S.~Steegen, F.~Tuerlinckx, A.~Gelman, and W.~Vanpaemel.
\newblock Increasing transparency through a multiverse analysis.
\newblock \emph{Perspectives on Psychological Science}, 11\penalty0 (5):\penalty0 702--712, 2016.

\bibitem[Tsiatis(2006)]{tsiatis2006semiparametric}
A.~Tsiatis.
\newblock \emph{Semiparametric Theory and Missing Data}.
\newblock Springer, 2006.

\bibitem[Van~der Vaart(2000)]{van2000asymptotic}
A.~Van~der Vaart.
\newblock \emph{Asymptotic Statistics}, volume~3.
\newblock Cambridge university press, 2000.

\bibitem[VanderWeele and Ding(2017)]{vanderweele2017sensitivity}
T.~VanderWeele and P.~Ding.
\newblock Sensitivity analysis in observational research: introducing the e-value.
\newblock \emph{Annals of Internal Medicine}, 167\penalty0 (4):\penalty0 268--274, 2017.

\bibitem[Wedderburn(1974)]{wedderburn1974quasi}
R.~W.~M. Wedderburn.
\newblock Quasi-likelihood functions, generalized linear models, and the gauss—newton method.
\newblock \emph{Biometrika}, 61\penalty0 (3):\penalty0 439--447, 1974.

\bibitem[Yu(2013)]{yu2013stability}
B.~Yu.
\newblock Stability.
\newblock \emph{Bernoulli}, 19\penalty0 (4):\penalty0 1484--1500, 2013.

\bibitem[Yu and Kumbier(2020)]{yu2020veridical}
B.~Yu and K.~Kumbier.
\newblock Veridical data science.
\newblock \emph{Proceedings of the National Academy of Sciences}, 117\penalty0 (8):\penalty0 3920--3929, 2020.

\bibitem[Zhang et~al.(2013)Zhang, Sch{\"o}lkopf, Muandet, and Wang]{zhang2013domain}
K.~Zhang, B.~Sch{\"o}lkopf, K.~Muandet, and Z.~Wang.
\newblock Domain adaptation under target and conditional shift.
\newblock In \emph{International Conference on Machine Learning}, pages 819--827, 2013.

\end{thebibliography}

\newpage

\appendix

\section*{Appendix}

In Section~\ref{sec:pert-model}, we discuss additional properties of the isotropic perturbation model. Section~\ref{sec:proofs} contains the proofs. Section~\ref{sec:altway} discusses how to form robust confidence intervals if the data analyst trusts one of the estimators $\hat \theta^k$ more than others. Section~\ref{sec:add_simulations} presents additional simulation results. 

\section{Properties of The Isotropic Perturbation Model}\label{sec:pert-model}

Recall that conditionally on $\xi$, the data $(D_i)_{1\le i \le n}$ are drawn i.i.d.\ from the perturbed distribution $\mathbb{P}^\xi(D = \bullet)$, where $\xi$ is an unobserved random variable. Note that an estimator $\hat \theta = \hat \theta(D_1,\ldots,D_n)$ for some parameter $\theta^0(\mathbb{P})$ now has two sources of uncertainty: the uncertainty due to sampling and the uncertainty due to the random perturbation.
\begin{equation*}\label{eqn:decompose}
    \hat{\theta} -  \theta^0(\mathbb{P}) = \underbrace{\hat{\theta} - \theta(\mathbb{P}^\xi)}_{\substack{\text{variation due to} \\ \text{sampling}}} + \underbrace{\theta(\mathbb{P}^\xi) - \theta^0(\mathbb{P})}_{\substack{\text{variation due to} \\ \text{random perturbation}}}
\end{equation*}
We refer to the second component as \textit{distributional uncertainty}. In this section we will study such distributional perturbation models in more detail. 
In Section~\ref{sec:appendix-motivation}, we provide additional insights into the motivation behind the random distributional perturbation model. 
 In Section~\ref{sec:unique}, we will show that under a strong symmetry assumption, there exists only one class of perturbation models that is characterized by a one-dimensional parameter $\delta_\text{dist}$. %
 In Section~\ref{sec:outlook}, we will sketch an extension of the random perturbation model that allows different parts of the distribution to be affected by different perturbations. 

 \subsection{Additional Insights into the Isotropic Perturbation Model}\label{sec:appendix-motivation}

Here we present additional insights into the weight-based distributional perturbation model.  We draw inspiration from real-world examples presented in Figure \ref{fig:qqplots} to construct the random perturbation model. A priori, there may be several mathematical random perturbation models leading to Figure \ref{fig:qqplots}. To simplify the discussion, in the following we will ignore sampling uncertainty. First, we will show that constant variance inflation implies random weights that are (almost) uncorrelated for disjoint events. Then, in the discrete setting, we will show that random weights imply constant variance inflation.

First, we study models that give rise to the constant variance inflation observed in Figure~\ref{fig:qqplots}.  To be precise, as working assumption we assume that for all square-integrable functions $\psi(D)$ under $\mathbb{P}$, 
    \begin{equation}\label{eq:var-inflation}
        \text{Var}_P(\mathbb{E}^{\xi}[\psi(D)]) = \delta_{\text{dist}}^2 \text{Var}_{\mathbb{P}}(\psi(D)),
    \end{equation}
for some variance inflation factor $\delta_\text{dist}^2$. As before, $\mathbb{P}$ refers to the unperturbed distribution of $D$ and $P$ refers to the marginal distribution of the perturbation and the observed data, $(\xi, D_1,\ldots,D_n)$. Using equation~\eqref{eq:var-inflation}, for all square-integrable functions $\psi(D)$, $\psi'(D)$,
    \begin{align*}
        &2 \text{Cov}_P(\mathbb{E}^{\xi}[\psi(D)], \mathbb{E}^{\xi}[\psi'(D)]) \\
        &=  \text{Var}_P(\mathbb{E}^{\xi}[\psi(D)] + \mathbb{E}^{\xi}[\psi'(D)]) - \text{Var}_P(\mathbb{E}^{\xi}[\psi(D)]) - \text{Var}_P(\mathbb{E}^{\xi}[\psi'(D)])  \\
        & = \delta_{\text{dist}}^2 (\text{Var}_{\mathbb{P}}(\psi(D) +\psi'(D)) -\text{Var}_{\mathbb{P}}(\psi(D)) - \text{Var}_{\mathbb{P}}(\psi'(D))) \\
        & = 2 \delta_{\text{dist}}^2 \text{Cov}_{\mathbb{P}}(\psi(D), \psi'(D)).
    \end{align*}  
    Thus, for disjoint $D$-measurable events $A$ and $B$ with $\mathbb{P}(A) = \mathbb{P}(B) = 1/K$, 
    \begin{equation*}
        \text{Cov}_{P}(\mathbb{P}^{\xi}[A], \mathbb{P}^{\xi}[B]) = \text{Cov}_{P}(\mathbb{E}^{\xi}[1_A], \mathbb{E}^{\xi}[1_B]) = \delta_{\text{dist}}^2 \text{Cov}_{\mathbb{P}}(1_A, 1_B) = - \frac{\delta_{\text{dist}}^2}{K^2} ,
    \end{equation*}
    \begin{equation*}
         \text{Var}_P(\mathbb{P}^{\xi}[A]) = \text{Var}_P(\mathbb{E}^{\xi}[1_A]) = \delta_{\text{dist}}^2 \text{Var}_{\mathbb{P}}(1_A) = \delta_{\text{dist}}^2\frac{1}{K}\left(1-\frac{1}{K}\right) .
    \end{equation*}
    Thus, $\mathbb{P}^{\xi}[A]$ and $\mathbb{P}^{\xi}[B]$ have the same variance and are marginally uncorrelated (ignoring lower order terms). Moreover, the right hand sides depend only on $\delta_{\text{dist}}^2$ and $K$. This inspires us to construct a random perturbation model by initially partitioning the probability space into $K$ disjoint bins with equal probability and then adjusting the probability of each partition with random weights constructed by positive i.i.d. random variables. As discussed in Section~2.2, empirical means are asymptotically Gaussian as the partitioning becomes finer, no matter how exactly the probability space was partitioned. 

    We will now go in the reverse direction. We will show that random re-weighting implies equation~\eqref{eq:var-inflation} in a simple discrete model. We will consider the simple example of a discrete uniform distribution $\mathbb{P}(D = k) = \frac{1}{K}$ for $k=1,\ldots,K$. Let $W_1, \dots, W_K$ be i.i.d. positive random variables with finite variance. We define the randomly perturbed distribution $\mathbb{P}^{\xi}$ by setting
    \begin{equation}\label{eq:random-pert-discrete}
         \mathbb{P}^\xi(D=k) = \frac{\xi_k}{K},
    \end{equation} 
    where 
    $$
    \xi_k := \frac{W_k}{\sum_{k=1}^K W_k/K}.
    $$
    Note that since $W_1,\ldots,W_K$ are positive i.i.d. random variables, the random perturbations $\xi_1, \dots, \xi_K$ are exchangeable non-negative random variables that sum to $\sum \xi_k = K$. Then for $ 1 \leq k_1 \neq k_2 \leq K$, we have
    \begin{align*}
        \text{Cov}_{P}(\mathbb{P}^{\xi}[D = k_1], \mathbb{P}^{\xi}[D = k_2]) = - \frac{\delta_{\text{dist}}^2}{K^2}\\
        \text{Var}_P(\mathbb{P}^{\xi}[D = k_1]) = \delta_{\text{dist}}^2\frac{1}{K}\left(1-\frac{1}{K}\right)
    \end{align*}
    where $\delta_{\text{dist}}^2 :=  \frac{1}{K-1} \text{Var}(\xi_1)$. We used that since $\text{Var}(\sum_k \xi_k) = \text{Var}(K) = 0$, we have $ \frac{1}{K-1} \text{Var}(\xi_1) = - \text{Cov}(\xi_1, \xi_2)$. Moreover, for any function $\psi : \{1,\ldots,K \} \rightarrow \mathbb{R}$,
    \begin{align*}
        \text{Var}(\mathbb{E}^\xi[\psi(D)])
        & = \text{Var} (\xi_1) \sum_k \frac{\psi(k)^2}{K^2} + \text{Cov}(\xi_1, \xi_2) \sum_{k_1 \neq k_2} \frac{\psi(k_1) \psi(k_2)}{K^2} \\
        & = \delta_{\text{dist}}^2 \left((K-1) \sum_k \frac{\psi(k)^2}{K^2} - \sum_{k_1 \neq k_2} \frac{\psi(k_1)\psi(k_2)}{K^2}\right)\\ 
        & = \delta_{\text{dist}}^2 \left( \frac{1}{K} \sum_k \psi(k)^2 - \frac{1}{K^2} \left( \sum_{k}  \psi(k) \right)^2 \right) \\
        & = \delta_\text{dist}^2 \text{Var}_{\mathbb{P}}(\psi(D)).
    \end{align*}
    Thus, the random re-weighting model \eqref{eq:random-pert-discrete} implies equation~\eqref{eq:var-inflation}.    

\subsection{Uniqueness of Distributional Perturbation Model}\label{sec:unique}

We see that under the distributional perturbation model introduced earlier, the variance of the perturbation is proportional to the variance in the unperturbed distribution. This raises the question whether there are other ``symmetric" random perturbation schemes that do not satisfy the variation inflation property
in equation~\eqref{eq:var-inflation}. The following result gives a negative answer to this question.
We will see that under a symmetry assumption, there exists only one type of perturbation model, which is equivalent to the one in equation~\eqref{eq:var-inflation}. Roughly speaking, the symmetry assumption states that two events that have equal probability under $\mathbb{P}$ are perturbed in a similar fashion. In the following, we write $\mathbb{Q}$ for the marginal distribution of $(D,\xi)$, where first the perturbation $\xi$ is drawn and then $D \sim \mathbb{P}^\xi$. The proof of the following result can be found in Section~\ref{sec:proof-theorem-2}.

\begin{theorem}[Characterization of isotropic perturbation models]\label{theorem:char}
    Let $(D,\xi) \sim \mathbb{Q}$ and assume that there exists a function $h(\bullet)$ such that $h(D)$ is uniformly distributed on $[0,1]$. Let $\mathbb{P}^\xi = \mathbb{Q}( \bullet | \xi)$ and let $\mathbb{P}$ denote the marginal distribution of $D$ under $\mathbb{Q}$. Assume that for any $D$-measurable events $A$ and $B$ with $\mathbb{P}(A) = \mathbb{P}(B)$,
    \begin{equation}\label{eq:symmetry}
        \mathrm{Var}({\mathbb{P}^\xi(A)}) = \mathrm{Var}({\mathbb{P}^\xi(B)}).
    \end{equation}
    Furthermore, assume that for every sequence of $D$-measurable events $A_j$ with $\mathbb{P}(A_j) \rightarrow 0$,
    \begin{equation*}
        \mathrm{Var}({\mathbb{P}^\xi(A_j)}) \rightarrow 0.
    \end{equation*}
    Then for any $\psi(D) \in L^2(\mathbb{P})$
\begin{equation}\label{eq:pert-model}
            \mathrm{Var}(\mathbb{E}^\xi[\psi(D)] ) = \delta_\text{dist}^2 \mathrm{Var}_{\mathbb{P}}(\psi(D)),
\end{equation}
 for some fixed $\delta_\text{dist} \ge 0$. 
\end{theorem}
The assumption that $h(D)$ exists is satisfied for any probability space that includes a continuous random variable. Thus, this is a regularity assumption that makes sure that the probability space is ``rich enough". Let us discuss what this result means for the behaviour of empirical means.  Let $D_1,\ldots,D_n$ be i.i.d.\ drawn from $\mathbb{P}^\xi$. Then, for all square-integrable functions $\psi(D) \in L^2(\mathbb{P})$ marginally across sampling uncertainty and distributional uncertainty we have
\begin{align*}
  \text{Var}_P \left( \frac{1}{n} \sum_{i=1}^n \psi(D_i) \right) & =  \left(\frac{1}{n}  + \delta_\text{dist}^2 - \frac{\delta_\text{dist}^2}{n}\right) \text{Var}_{\mathbb{P}}(\psi(D)) \\
  & = \frac{\delta^2}{n} \text{Var}_{\mathbb{P}}(\psi(D)),
\end{align*}
with $\delta^2 = 1 + n \delta_\text{dist}^2 - \delta_\text{dist}^2$. Since $\mathbb{P}^\xi = \mathbb{Q}(\bullet | \xi)$ we also have $E_P[\frac{1}{n} \sum_{i=1}^n \psi(D_i)] = \mathbb{E}[\psi(D)]$.

There are two major assumptions in this theorem. The first assumption says that two events that have the same probability are perturbed in the same fashion. This can be seen as a symmetry assumption. The second assumption says that events that have a small probability are only perturbed by a small amount. This can be seen as a regularity assumption.

Then, up to a one-dimensional parameter $\delta$, the variance of functions is uniquely determined. This means that using strong symmetry we have reduced the problem of estimating an infinite-dimensional perturbation model to a one-dimensional quantity $\delta$. Note that the statement in Theorem~\ref{theorem:char} is slightly weaker than Lemma~\ref{lemma:random_weight}, since it is only a statement about variances and not about the asymptotic distribution of $\frac{1}{n} \sum_{i=1}^n \psi(D_i)$. 

In practice, some researchers might object to the symmetry assumption in equation~\eqref{eq:symmetry}. It turns out that the perturbation model can be generalized. 
In the following section, we will give a brief outlook of how perturbation models can be used to perturb different parts of a distribution differently.

\subsection{Beyond Isotropic Distributional Perturbations}\label{sec:outlook}

The discussion in Section~\ref{sec:unique} shows that under a strong symmetry assumption, up to an unknown scale factor $\delta$, there exists only one type of perturbation model. However, in practice there might be a situation where one does not expect a perturbation to affect all parts of the distribution in the same way. Consider $D=(X,Y)$. For example, one might expect that the distribution of $X$ is perturbed between settings but that the measurement error is invariant. This may lead one to want to model a situation where $p(x)$ is perturbed but $p(y|x)$ is not perturbed. Under appropriate regularity conditions on $\psi$ we have
\begin{align*}
    \mathbb{E}^\xi[\psi(X, Y)] - \mathbb{E}[\psi(X, Y)] & = \mathbb{E}^\xi[\mathbb{E}[\psi(X, Y)|X]] - \mathbb{E}[\mathbb{E}[\psi(X, Y)|X]] \\& \stackrel{d}{\approx} \mathcal{N}(0, \delta_\text{dist}^2 \text{Var}_{\mathbb{P}}(\mathbb{E}[\psi(X, Y)|X])).
\end{align*}
     If $\delta_\text{dist}$ is known or can be estimated, this allows us to adjust variance and confidence intervals to account for uncertainty both due to sampling and distributional perturbations, similarly as in Section~\ref{sec:calibrated}.

\section{Proofs}\label{sec:proofs}

\subsection{Auxiliary Results and Proof of Lemma~\ref{lemma:random_weight}}\label{sec:asymp_pert_model}

\textbf{Notation}: We write $\mathbb{P}$ for the target distribution and $\mathbb{P}^\xi$ for the randomly perturbed distribution from which we draw $n$ i.i.d.\ data samples $(D_i)_{i=1,\ldots,n}$. 
In both examples $\xi$ can be seen as a random variable that encodes the perturbations. 
The expectation of $f(D_1, \dots, D_n)$ over the joint distribution of $(\xi, D_1, \dots, D_n)$ can be written as $E_{\xi}[\mathbb{E}^\xi(f(D_1, \dots, D_n)]$ where $E_{\xi}$ means we take the expectation over $\xi$ and $\mathbb{E}^\xi$ means that we take the expectation over $(D_1,\ldots,D_n)$, conditionally on $\xi$.

\subsubsection{Auxiliary results}

Let us first state an auxiliary lemma that will turn out helpful for proving Lemma \ref{lemma:random_weight}.

\begin{lemma}\label{lemma:random_pert_1}
Let the assumptions of Lemma~\ref{lemma:random_weight} hold. For the sequence of random variables $\xi =\xi(n)$, for any bounded  $\psi(\bullet)$ we have that
\begin{equation*}
    \mathbb{E}^\xi[\psi(D)] - \mathbb{E}[\psi(D)] \stackrel{d}{=} \gamma_n \sqrt{\text{Var}_{\mathbb{P}}(\psi(D))} \cdot Z + o_p(\gamma_n),
\end{equation*}
where $Z$ follows a standard normal distribution and $\gamma_n^2 = \frac{\text{Var}(W)}{m(n)E[W]^2}$. Here we write $m(n)$ to make it explicit that $m$ grows with $n$. 
\end{lemma}

\begin{proof}
Let $\phi = \psi \circ h$.
Without loss of generality, assume that $\mathbb{E}[\phi(U)] = 0.$ Note that
\begin{equation*}
    \sqrt{m}(\mathbb{E}^\xi[\phi(U)] - \mathbb{E}[\phi(U)]) = \frac{\sqrt{m}\sum_{k=1}^{m}\int_{x \in I_k} \phi(x) dx \cdot (W_k - E[W])}{\sum_{k=1}^{m}W_k/m}.  %
\end{equation*}
Let 
\begin{equation*}
    Y_{m,k} := \sqrt{m}\int_{x \in I_k} \phi(x) dx \cdot (W_k - E[W]).
\end{equation*}
First, note that 
\begin{equation}\label{eq:cltcond1}
    E[Y_{m,k}] = 0
\end{equation}
for all $k$. As the second step, we want to show that 
\begin{equation}\label{eq:cltcond2}
    \sum_{k=1}^{m}E[Y_{m,k}^2] = \text{Var}(W)\cdot m\sum_{k=1}^{m}\left(\int_{x \in I_k} \phi(x) dx\right)^2 \xrightarrow[]{} \text{Var}(W) \cdot \text{Var}_{\mathbb{P}}(\phi(U)).
\end{equation}
For any $f \in L^2([0,1])$, define $\Pi_m(f)$ as
\begin{equation*}
    \Pi_m(f)(x) = \sum_{k=1}^{m} \left(m\int_{x \in I_k} f(x) dx\right) \cdot I(x \in I_k).
\end{equation*}
Then, we have
\begin{align*}
    \Bigg|m\sum_{k=1}^{m}\left(\int_{x \in I_k} \phi(x) dx\right)^2  - \text{Var}_{\mathbb{P}}(\phi(U))\Bigg| &= ||\phi - \Pi_m(\phi)||_2^2 \xrightarrow[]{} 0.
\end{align*}
as $m$ goes to infinity. This is because any bounded function can be approximated by a sequence of step functions of the form $\sum_{k=1}^{m}b_kI(x \in  {I_k})$.
Next we will show that for any $\epsilon > 0$,
\begin{equation}\label{eq:cltcond3}
    g_m(\epsilon) = \sum_{k=1}^{m}E[Y_{m,k}^2 ; |Y_{m,k}| \geq \epsilon] \xrightarrow[]{}0. 
\end{equation}
This is implied by the dominated convergence theorem as
\begin{align*}
&\sum_{k=1}^{m}E[Y_{m,k}^2 ; |Y_{m,k}| \geq \epsilon ] \\
&\quad\quad \leq  \sum_{k=1}^{m}\left(\int_{x\in I_k} \phi^2(x) dx\right) E[(W_k-E[W])^2 I(||\phi||_{\infty} |W_k-E[W]|/\sqrt{m} \geq \epsilon)]\\
    & \quad\quad = ||\phi||_2^2E[(W-E[W])^2 I(||\phi||_{\infty} |W-E[W]|/\sqrt{m} \geq \epsilon)]  \xrightarrow[]{} 0 . 
\end{align*}
Combining equations~\eqref{eq:cltcond1}, \eqref{eq:cltcond2}, and \eqref{eq:cltcond3}, we can apply Lindeberg's CLT.  With Slutsky's theorem, we have
\begin{align*}
    \sqrt{m}(\mathbb{E}^{\xi}[\phi(U)] - \mathbb{E}[\phi(U)]) =\frac{\sum_{k=1}^m Y_{m,k}}{\sum_{k=1}^m W_k/m} \stackrel{d}{\xrightarrow[]{}} \mathcal{N}(0, \text{Var}(W)\text{Var}_{\mathbb{P}}(\phi(U))/ E[W]^2). 
\end{align*}
This completes the proof. 
\end{proof}

\begin{lemma}
\label{lemma:random_pert_2}
Let the assumptions of Lemma~\ref{lemma:random_weight} hold. Assume that for a sequence of random variables $\xi = \xi(n)$ there exists a sequence $\gamma_n$ with limit $\delta^2 = \lim_n(1+n\gamma_{n}^2) < \infty$ such that for any bounded $\psi(\bullet)$ we have
\begin{equation}\label{eq:assump-rand}
    \mathbb{E}^\xi[\psi(D)] - \mathbb{E}[\psi(D)] \stackrel{d}{=} \gamma_n \sqrt{\text{Var}_{\mathbb{P}}(\psi(D))} \cdot Z + o_p(\gamma_n),
\end{equation}
where $Z$ follows a standard normal distribution.
Then, for any bounded $\psi(\bullet)$, it holds that 
 \begin{equation*}
    \frac{1}{\sqrt{n}}\sum_{i=1}^{n} (\psi(D_i^n) - \mathbb{E}[\psi(D)]) \stackrel{d}{\xrightarrow[]{}} \mathcal{N}(0, \delta^2 \text{Var}_{\mathbb{P}}(\psi(D))).
\end{equation*}
\end{lemma}

\begin{proof}
In the proof, we suppress the dependence of $\xi$ on $n$. We want to show that for any $x$, 
\begin{equation*}
    E_{\xi}\Bigg[\mathbb{P}^{\xi} \left( \frac{1}{\sqrt{n}}\sum_{i=1}^{n} (\psi(D_i^n) - \mathbb{E}[\psi(D)]) \leq x \cdot \sqrt{\delta^2\text{Var}_{\mathbb{P}}(\psi(D))}\right) \Bigg] = \Phi(x) +o(1),
\end{equation*}
where $\Phi$ is the cdf of a standard Gaussian random variable. Let us define
\begin{equation*}
    Y_{n} = x\cdot\delta\cdot \frac{\sqrt{\text{Var}_{\mathbb{P}}(\psi(D))}}{\sqrt{\text{Var}_{\mathbb{P}^{\xi}}(\psi(D))}}   - \frac{\sqrt{n}(\mathbb{E}^{\xi}[\psi(D)] - \mathbb{E}[\psi(D)])}{ \sqrt{\text{Var}_{\mathbb{P}^{\xi}}(\psi(D))}},
\end{equation*}
where $\text{Var}_{\mathbb{P}^{\xi}}(\psi(D))$ denotes the variance of $\psi(D)$ where $D \sim \mathbb{P}^{\xi}$. Then,
\begin{equation*}
    E_{\xi}\Bigg[\mathbb{P}^{\xi} \left( \frac{1}{\sqrt{n}}\sum_{i=1}^{n} (\psi(D_i^n) - \mathbb{E}^{\xi}[\psi(D)]) + \sqrt{n}(\mathbb{E}^{\xi}[\psi(D)] - \mathbb{E}[\psi(D)]) \leq x \cdot \sqrt{\delta^2\text{Var}_{\mathbb{P}}(\psi(D))}\right) \Bigg]
\end{equation*}
\begin{equation}\label{eq:ast}
    = E_{\xi}\Bigg[\mathbb{P}^{\xi} \left( \frac{1}{\sqrt{n}}\sum_{i=1}^{n} \frac{\psi(D_i^n) - \mathbb{E}^{\xi}[\psi(D)]}{\sqrt{\text{Var}_{\mathbb{P}^{\xi}}(\psi(D))}} \leq Y_{n} \right) \Bigg].
\end{equation}
\quad\\
We define $g_{n}(y; \xi)$ as
\begin{equation*}
    g_{n}(y ; \xi) =  \mathbb{P}^{\xi} \left( \frac{1}{\sqrt{n}}\sum_{i=1}^{n} \frac{\psi(D_i^n ) - \mathbb{E}^{\xi}[\psi(D)]}{\sqrt{\text{Var}_{\mathbb{P}^{\xi}}(\psi(D))}}\leq y \right).
\end{equation*}
By Berry–Esseen, it holds that
\begin{equation*}
    \sup_y \Bigg|\mathbb{P}^{\xi} \left(
    \frac{1}{\sqrt{n}} \sum_{i=1}^{n}\frac{ \psi(D_i)- \mathbb{E}^{\xi}[\psi(D)]}{\sqrt{\text{Var}_{\mathbb{P}^{\xi}}(\psi(D))}}  \leq y\right) - \Phi(y)\Bigg|  \leq \frac{C\mathbb{E}^{\xi}|{\psi(D)}^3|}{(\mathbb{E}^{\xi}|{\psi(D)}^2|)^{3/2} \sqrt{n}},
\end{equation*}
for all $n$. Invoking equation~\eqref{eq:assump-rand} for $\psi^2$ and $\psi^3$, we have that $\mathbb{E}^{\xi}|{\psi(D)}^3|/(\mathbb{E}^{\xi}|{\psi(D)}^2|)^{3/2}$ converges in probability to $\mathbb{E}|{\psi(D)}^3|/(\mathbb{E}|{\psi(D)}^2|)^{3/2} < \infty$ as $n \rightarrow \infty$. Then the right-hand side of the above inequality converges in probability to 0 as $n \xrightarrow[]{} \infty$, which implies that
    \begin{equation*}
    \sup_y |g_{n}(y; \xi) - \Phi(y)| \xrightarrow[]{p} 0.
\end{equation*}
Using this result, 
\begin{align*}
    \eqref{eq:ast}  = E_{\xi}[g_{n}(Y_{n})] & = 
     E_{\xi}[g_{n}(Y_{n})] - E_{\xi}[\Phi(Y_{n})] + E_{\xi}[\Phi(Y_{n})] \\
     & \leq E_{\xi}[\sup_y|g_{n}(y) - \Phi(y)|]+ E_{\xi}[\Phi(Y_{n})] \\ 
     & = E_{\xi}[\Phi(Y_{n})] + o(1).
\end{align*}
Here, we used the dominated convergence theorem. Using equation~\eqref{eq:assump-rand}, $\text{Var}_{\mathbb{P}^{\xi}}(\psi(D)) \xrightarrow[]{p} \text{Var}_{\mathbb{P}}(\psi(D))$. Then, we have
$$Y_{n} \xrightarrow[]{d}  \delta x - \sqrt{\delta^2-1} Z,$$
where $Z$ is a standard Gaussian random variable. Since $\Phi$ is bounded and continuous, by Portmanteau Lemma, we get
\begin{equation*}
    \lim_{n \xrightarrow[]{} \infty} E_{\xi}[\Phi(Y_{n})] = E[\Phi(\delta x - \sqrt{\delta^2-1} Z)] = \Phi(x).
\end{equation*}
This completes the proof.
\end{proof}

\subsubsection{Proof of Lemma~\ref{lemma:random_weight}}

Now let us show that the Lemma~\ref{lemma:random_weight} holds.

\begin{proof}
Without loss of generality for notational simplicity we restrict ourselves to the case $l =1$, i.e. $\psi: \mathcal{D} \mapsto \mathbb{R}$. As before we write $\phi(U) = \psi \circ h (U)$.
For any $\psi \in L^2(\mathbb{P})$ and for any $\epsilon >0$, there exits a bounded function $\psi^B$ such that $\mathbb{E}[\psi(D)] = \mathbb{E}[\psi^B(D)]$ and $||\psi - \psi^B||_{L^2(\mathbb{P})} < \epsilon$. Note that
\begin{align*}
     \frac{1}{\sqrt{n}}\sum_{i=1}^{n}(\psi(D_i^n) - \mathbb{E}[\psi(D)]) &=   \frac{1}{\sqrt{n}}\sum_{i=1}^{n}(\psi(D_i^n) - \psi^B(D_i^n)) \label{eq:eq1}\tag{a} \\
    &  + \frac{1}{\sqrt{n}}\sum_{i=1}^{n}(\psi^B(D_i^n) - \mathbb{E}[\psi^B(D)]). \label{eq:eq2}\tag{b}
\end{align*}
Without loss of generality, let's assume that $\mathbb{E}[\psi(D)] = 0$. Note that 
\begin{align*}
    \eqref{eq:eq1} =\frac{1}{\sqrt{n}}\sum_{i=1}^{n}\left( (\psi-\psi^B)(D_i^n) - \mathbb{E}^{\xi}[(\psi - \psi^B)(D)]\right)
    + \sqrt{n}\mathbb{E}^{\xi}[(\psi - \psi^B)(D)].
\end{align*}
The marginal variance of its first part is
\begin{align*}
    E[\text{Var}(\frac{1}{\sqrt{n}}\sum_{i=1}^{n}\left( (\psi-\psi^B)(D_i^n) - \mathbb{E}^{\xi}[(\psi - \psi^B)(D)]\right)|\xi)] &\leq E_{\xi}[\mathbb{E}^{\xi}[(\psi-\psi^B)^2(D)]] \\ &= \mathbb{E}[(\psi-\psi^B)^2(D)] <\epsilon^2.
\end{align*}
Let's look at the second part of \eqref{eq:eq1}. Recall that we write $\phi(U) = \psi \circ h (U)$. Note that for any $\phi \in L^2([0,1])$ such that $\mathbb{E}[\phi(U)] = 0$, 
\begin{align*}
   \sqrt{m}(\mathbb{E}^{\xi}[\phi(U)]) &= \frac{\sqrt{m}\sum_{k=1}^{m}\int_{x \in I_k} \phi(x) dx \cdot (W_k - E[W])}{\sum_{k=1}^{m}W_k/m}.
\end{align*}
With $\phi = (\psi - \psi^B)\circ h$, the variance of the numerator is bounded as
\begin{align*}
    \text{Var}(W) \sum_{k=1}^{m}m\left(\int_{x \in I_k}\phi(x)dx\right)^2 
    & \leq \text{Var}(W) \sum_{k=1}^{m}\int_{x \in I_k}\phi^2(x)dx \\
    & = \text{Var}(W) \mathbb{E}[\phi^2(U)] \\
    & < \text{Var}(W) \cdot \epsilon^2,
\end{align*}
where the first inequality holds by Jensen's inequality with $m \int_{x\in I_k} dx = 1$. Therefore, 
\begin{equation*}
\sqrt{n}\mathbb{E}^{\xi}[(\psi - \psi^B)(D)] = \sqrt{r} \cdot \frac{\sqrt{m}\sum_{k=1}^{m}\int_{x \in I_k} \phi(x) dx \cdot (W_k - E[W])}{E[W]} + s_n
\end{equation*}
where $s_n$ is $o_p(1)$. Combining results, we have that $E[\eqref{eq:eq1} - s_n]=0$ and $$\text{Var}_P(\eqref{eq:eq1} - s_n) \leq C  \cdot \epsilon^2$$ for some constant $C$. Now we want to show that for any $x$, 
\begin{equation*}
    \lim_{n \xrightarrow[]{}\infty}P\left(\frac{1}{\sqrt{n}}\sum_{i=1}^{n}\frac{(\psi(D_i^n) - \mathbb{E}[\psi(D)])}{\delta\sqrt{\text{Var}_{\mathbb{P}}(\psi(D))}} \leq x\right) = \Phi(x)
\end{equation*}
where $\Phi(x)$ is the cdf of a standard Gaussian random variable. Note that
\begin{align*}
    & P\left(\frac{1}{\sqrt{n}}\sum_{i=1}^{n}\frac{(\psi(D_i^n) - \mathbb{E}[\psi(D)])}{\delta\sqrt{\text{Var}_{\mathbb{P}}(\psi(D))}} \leq x\right)\\
    & \leq P\left(\frac{\eqref{eq:eq2}}{\delta\sqrt{\text{Var}_{\mathbb{P}}(\psi(D))}} \leq x + 2\eta\right)+ P \left( \frac{|\eqref{eq:eq1}|}{\delta\sqrt{\text{Var}_{\mathbb{P}}(\psi(D))}} > 2\eta\right) \\
    & \leq  P\left(\frac{\eqref{eq:eq2}}{\delta\sqrt{\text{Var}_{\mathbb{P}}(\psi(D))}} \leq x + 2\eta\right)+ P \left( \frac{|\eqref{eq:eq1} - s_n|}{\delta\sqrt{\text{Var}_{\mathbb{P}}(\psi(D))}} > \eta\right) + P \left( \frac{|s_n|}{\delta\sqrt{\text{Var}_{\mathbb{P}}(\psi(D))}} > \eta\right) \\
    & \leq P\left(\frac{\eqref{eq:eq2}}{\delta\sqrt{\text{Var}_{\mathbb{P}}(\psi(D))}} \leq x + 2\eta\right) + \frac{C \cdot \epsilon^2}{\eta^2 \delta^2 \text{Var}_{\mathbb{P}}(\psi(D))} + P \left( \frac{|s_n|}{\delta\sqrt{\text{Var}_{\mathbb{P}}(\psi(D))}} > \eta\right),
\end{align*}
where the last inequality holds by Chebyshev's inequality. With Lemma \ref{lemma:random_pert_1} and Lemma \ref{lemma:random_pert_2}, we have that \begin{align*}
    \eqref{eq:eq2} & \xrightarrow[]{d} \delta N(0, \text{Var}_{\mathbb{P}}(\psi^B(D))).
\end{align*}
Note that
\begin{equation*}
    \left(\sqrt{\text{Var}_{\mathbb{P}}(\psi(D))} - \sqrt{\text{Var}_{\mathbb{P}}(\psi^B(D))}\right)^2 \leq \mathbb{E}[(\psi - \psi^B)^2(D)] \leq \epsilon^2,
\end{equation*}
and thus
\begin{equation*}
    1 - \epsilon \cdot \frac{1}{\sqrt{\text{Var}_{\mathbb{P}}(\psi^B(D))}} \leq \frac{\sqrt{\text{Var}_{\mathbb{P}}(\psi(D))}}{\sqrt{\text{Var}_{\mathbb{P}}(\psi^B(D))}} \leq 1 + \epsilon \cdot \frac{1}{\sqrt{\text{Var}_{\mathbb{P}}(\psi^B(D))}}.
\end{equation*}
Then, we get that
\begin{align*}
     &  \lim\sup_{n\xrightarrow[]{}\infty} P\left(\frac{1}{\sqrt{n}}\sum_{i=1}^{n}\frac{(\psi(D_i^n) - \mathbb{E}[\psi(D)])}{\delta\sqrt{\text{Var}_{\mathbb{P}}(\psi(D))}} \leq x\right) - \Phi(x)\\
     & \quad \leq \Phi \left(\left(1 + \epsilon \cdot \frac{1}{\sqrt{\text{Var}_{\mathbb{P}}(\psi^B(D))}}\right)(x + 2\eta)\right) - \Phi(x) + \frac{C \cdot \epsilon^2}{\eta^2 \delta^2 \text{Var}_{\mathbb{P}}(\psi(D))}.
\end{align*}
Similarly, we can show that 
\begin{align*}
    &\lim\inf_{n \xrightarrow[]{}\infty}P\left(\frac{1}{\sqrt{n}}\sum_{i=1}^{n}\frac{(\psi(D_i^n) - \mathbb{E}[\psi(D)])}{\delta\sqrt{\text{Var}_{\mathbb{P}}(\psi(D)}} \leq x\right) - \Phi(x)\\ 
    & \quad \geq \Phi \left(\left(1 - \epsilon \cdot \frac{1}{\sqrt{\text{Var}_{\mathbb{P}}(\psi^B(D))}}\right)(x - 2\eta)\right) - \Phi(x) - \frac{C \cdot \epsilon^2}{\eta^2 \delta^2 \text{Var}_{\mathbb{P}}(\psi(D))}.
\end{align*}
Note that results hold for arbitrary $\eta>0$ and $\epsilon > 0$. Let $\eta = \sqrt{\epsilon}$.
Then for any $x$, 
\begin{equation*}
    \lim_{n \xrightarrow[]{}\infty}P\left(\frac{1}{\sqrt{n}}\sum_{i=1}^{n}\frac{(\psi(D_i^n) - \mathbb{E}[\psi(D)])}{\delta\sqrt{\text{Var}_{\mathbb{P}}(\psi(D))}} \leq x\right) - \Phi(x) = 0.
\end{equation*}
This completes the proof. 
\end{proof}

\subsection{Examples of Variance Inflation Induced by Non-i.i.d.\ sampling }\label{sec:examplesofmodels}

In the following examples we discuss how Assumption~\ref{assump:a1} with $\delta \neq 1$ arises in non-standard sampling settings.  For simplicity, we start with an artificial example: sampling with replacement from an unknown subpopulation. 

\begin{example}[Sampling with replacement from an unknown subpopulation] \label{example:samp}

    Assume that $D_1', \dots, D_m'$ drawn i.i.d.\ from $\mathbb{P}$. Set $\xi = (D_1',\ldots,D_m')$. We define the randomly perturbed distribution $\mathbb{P}^\xi$ as the empirical measure 
    \begin{equation*}
        \mathbb{P}^\xi(D \in \bullet) = \frac{1}{m}\sum_{i=1}^{m}1_{D_i' \in \bullet}.
    \end{equation*}
    Let $n \rightarrow \infty$ and assume that $m(n)$ is a sequence of integers such that $\frac{n}{m(n)}$ converges to some limit $r \in (0,\infty)$. Conditionally on $\xi$, let $ (D_1^n,\ldots,D_n^n)$ be i.i.d.\ draws from $\mathbb{P}^\xi$. Then equation~\eqref{eq:normality} holds for any $\psi(\bullet)$ with finite second moment with
    \begin{equation*}
     \delta^2 = 1 + r.
    \end{equation*}
    \end{example}

\begin{proof}
Suppose that $D_1', \dots, D_m'$  are drawn from $\mathbb{P}$ for some sequence $m = m(n)$. Let $\mathbb{P}^\xi$ denote the empirical measure of $D_1', \dots, D_m'$. Then by the CLT, for any $\psi(\bullet)$ with finite second moment,
\begin{align*}
    \mathbb{E}^\xi[\psi(D)] - \mathbb{E}[\psi(D)] & = \frac{1}{m}\sum_{i=1}^{m}\psi(D_i') -\mathbb{E}[\psi(D)] \\ & \stackrel{d}{=} \gamma_n \sqrt{\text{Var}_{\mathbb{P}}(\psi(D))} \cdot Z  + o_p(\gamma_n)
\end{align*}
where $Z$ follows a standard normal distribution and $\gamma_n^2 = 1/m(n)$. By applying Lemma \ref{lemma:random_pert_2} and following the proof of Lemma \ref{lemma:random_weight}, for any $\psi(\bullet)$ with finite second moment, we get
  \begin{equation*}
        \frac{1}{\sqrt{n}}\sum_{i=1}^{n} (\psi(D_i^n) - \mathbb{E}[\psi(D)] )\stackrel{d}{\xrightarrow[]{}} \mathcal{N}(0,  \delta^2 \text{Var}_{\mathbb{P}}(\psi(D))),
    \end{equation*}
where $\delta^2 = 1 + r$.
\end{proof}    

Sampling with replacement from a finite population might seem a bit artificial. The next example shows that a similar conclusion holds if we sample clusters, where units in a single cluster are highly correlated, and units between clusters are independent. If the cluster structure is known, one can use clustered standard errors. However, in general the dependence structure might be unknown. 
\begin{example}[Sampling clusters with unobserved membership] \label{example:clusters}
Here, we consider a setting where some observations are associated, but where the overall dependence structure is unknown. This is similar to the previous setting, but there are no ties in the data set. Consider $\mathbb{P}$ a probability distribution with positive density over a compact subset of $\mathbb{R}^p$.
Draw i.i.d.\ observations $D_1',\ldots,D_{m}'$ from $\mathbb{P}$. Set $\xi = (D_1',\ldots,D_{m}')$. Conditionally on $\xi$, let $ (D_1^n,\ldots,D_n^n)$ be i.i.d.\ draws from $\mathbb{P}^\xi$, where
\begin{equation*}
  \mathbb{P}^\xi(D \in \bullet) = \frac{1}{m} \sum_{i=1}^{m} \mathbb{P}( D \in \bullet | \| D_i' - D \|_2 \le \epsilon_n),
\end{equation*} 
 where $\epsilon_n > 0$ is a deterministic sequence with $\epsilon_n = o(1/\sqrt{n})$.  Furthermore, let $n \rightarrow \infty$ and assume that $m = m(n)$ is a sequence of integers such that $\frac{n}{m(n)}$ converges to some limit $r \in (0,\infty)$.  Then, equation~\eqref{eq:normality} holds for any bounded Lipschitz continuous $\psi(\bullet)$ with
 \begin{equation*}
     \delta^2 = 1 + r.
 \end{equation*}
\end{example}   

\begin{proof}
Using Lipschitz continuity and $\epsilon_n = o(1/\sqrt{n})$, we have
\begin{equation*}
\frac{1}{\sqrt{n}} \sum_{i=1}^n (\psi(D_i^n) - \mathbb{E}[\psi(D)]) = \frac{1}{\sqrt{n}} \sum_{i=1}^{n}  (\psi(D_i'') -  \mathbb{E}[\psi(D)]) + o_p(1),
\end{equation*}
where the $D_i''$ are drawn with replacement from $D_1',\ldots,D_m'$. We can now invoke Example~\ref{example:samp}.
\end{proof}

\subsection{Proof of Proposition \ref{prop:cons_var}}

\begin{proof}
Note that
\begin{align*}
    & \widehat{\text{Var}}_{\mathbb{P}}(\phi^k(D))
 = \underbrace{\frac{1}{n}\sum_{i=1}^{n} \left(\hat\phi^k(D_i) -\phi^k(D_i) -  \frac{1}{n}\sum_{i=1}^{n}\hat\phi^k(D_i) + \frac{1}{n}\sum_{i=1}^{n}\phi^k(D_i) \right)^2 }_{(i)}
    \\
    & \quad  +  \underbrace{\frac{2}{n}\sum_{i=1}^{n}\left(\hat\phi^k(D_i) -\phi^k(D_i) -  \frac{1}{n}\sum_{i=1}^{n}\hat\phi^k(D_i) + \frac{1}{n}\sum_{i=1}^{n}\phi^k(D_i) \right)\left(\phi^k(D_i) - \frac{1}{n}\sum_{i=1}^{n}\phi^k(D_i)\right)}_{(ii)} \\
    & \quad  + \underbrace{\frac{1}{n}\sum_{i=1}^{n}\left(\phi^k(D_i)- \frac{1}{n}\sum_{i=1}^{n}\phi^k(D_i)\right)^2}_{(iii)}.
\end{align*} 
As the $\phi^k$ has finite fourth moments, we can use Lemma~\ref{lemma:random_weight} to obtain (iii) $= \text{Var}_{\mathbb{P}}(\phi^k(D)) + o_p(1)$. Then by Cauchy-Schwartz inequality and Jensen's inequality,
\begin{align*}
    \text{(i)} &\leq \frac{2}{n}\sum_{i=1}^{n}\left(\hat\phi^k(D_i) -\phi^k(D_i)\right)^2 + 2\left(\frac{1}{n}\sum_{i=1}^{n}\left(\hat\phi^k(D_i) -  \phi^k(D_i)\right)\right)^2 \\
    & \leq \frac{4}{n}\sum_{i=1}^{n}\left(\hat\phi^k(D_i) -\phi^k(D_i)\right)^2.
\end{align*}
Since our influence function estimators are consistent, (i) $= o_p(1)$. Then again by Cauchy-Schwartz inequality, (ii) $= o_p(1)$. Combining results, we get
$$
\widehat{\text{Var}}_{\mathbb{P}}(\phi^k(D)) = \text{Var}_{\mathbb{P}}(\phi^k(D)) + o_p(1).
$$
This completes the proof.
\end{proof}

\subsection{Proof of Theorem~\ref{theorem:newci}}

\begin{proof}
By Assumption \ref{assump:a1}, \ref{ass:agreement} and Lemma~\ref{lemma:random_weight},
\begin{equation*}
 \begin{pmatrix}
   \sqrt{n}(\hat \theta^1 - \theta^0)\\
    \vdots \\
    \sqrt{n}(\hat \theta^K - \theta^0)
    \end{pmatrix} =    \frac{1}{\sqrt{n}}\sum_{i=1}^{n} \begin{pmatrix}
   \phi^1(D_i) \\
    \vdots \\
    \phi^K(D_i)
    \end{pmatrix} + o_p(1) = \delta\textbf{Z} + o_p(1),
\end{equation*}
where $\textbf{Z} = (Z_1, \dots, Z_K)^\intercal \sim \mathcal{N}(\mathbf{0},\text{diag}(\text{Var}(\phi^1), \dots, \text{Var}(\phi^K)))$. As $n\xrightarrow[]{}\infty$,  using that $\sum  \alpha_k = 1$ and $\hat{\alpha}_k = \alpha_k + o_p(1)$,
\begin{equation}\label{eq:asympestim}
    \sqrt{n}(\hat{\theta}^W - \theta^0) =  \sqrt{n} \sum_{k=1}^K \hat \alpha_k (\hat \theta^k - \theta^0) + o_p(1) = \delta \sum_{k=1}^K \alpha_k Z_k +o_p(1) \stackrel{d}{\xrightarrow[]{}} \delta \mathcal{N}(0, \alpha),
\end{equation}
where $\alpha = \frac{1}{\sum_{k=1}^K \frac{1}{\text{Var}_{\mathbb{P}}(\phi^k(D))}}$. 
By a similar calculation, we have that
\begin{equation*}
    n\hat \sigma^2_{bet} = \delta^2 \sum_{k=1}^K \alpha_k ( Z_k - \sum_{j=1}^K \alpha_j Z_j )^2 + o_p(1).
\end{equation*}
Thus,
\begin{equation}\label{eq:asympsbet}
    n\hat \sigma^2_{bet} = \delta^2 \alpha L_K + o_p(1),
\end{equation}
where
\begin{equation*}
    L_K = \frac{1}{\alpha} \sum_{k=1}^K \alpha_k ( Z_k - \sum_{j=1}^K \alpha_j Z_j )^2.
\end{equation*}
We will now show that
\begin{equation}\label{eq:todo}
    L_K \sim \chi^2(K-1) \quad \perp \quad  \frac{1}{\sqrt{\alpha}}\sum_{j=1}^{K}\alpha_j Z_j \sim \mathcal{N}(0, 1).
\end{equation}
First, note that
\begin{equation*}
  L_K   = \sum_{k=1}^{K}\left(Z_k \frac{\sqrt{\alpha_k}}{\sqrt{\alpha}} - \frac{\sqrt{\alpha_k}}{\sqrt{\alpha}} \sum_{j=1}^{K}\alpha_jZ_j\right)^2.
\end{equation*}
With this definition,
\begin{align*}
    L_K 
    &=   \sum_{k=1}^{K} \left( \frac{\sqrt{\alpha_k}}{\sqrt{\alpha}} Z_k - \frac{\sqrt{\alpha_k}}{\sqrt{\alpha}} \sum_{j=1}^{K}\alpha_j Z_j\right)^2 \\
    &= \sum_{k=1}^{K} \left( \frac{\sqrt{\alpha_k}}{\sqrt{\alpha}} Z_k - \sqrt{\alpha_k} \sum_{j=1}^{K} \sqrt{\alpha_j} \frac{\sqrt{\alpha_j}}{\sqrt{\alpha}} Z_j\right)^2 \\
    &=  \sum_{k=1}^K (\tilde Z_k - w_k \sum_{j=1}^K w_j \tilde Z_j )^2,
\end{align*}
where $ \tilde Z_k :=  Z_k \sqrt{\alpha_k} /\sqrt{\alpha} =  Z_k/\sqrt{\text{Var}_\mathbb{P}(\phi^k)}$ are i.i.d.\ standard normal and $  w_k := \sqrt{\alpha_k}$. Please note that $\sum_k w_k^2 = 1$. Thus, we can write this equation
\begin{equation*}
    L_K = \| \tilde Z - w ( w \cdot \tilde Z ) \|_2^2 = \| (\text{Id} - \Pi ) \tilde Z \|_2^2,
\end{equation*}
where $\Pi$ projects on the one-dimensional subspace spanned by $w$. 
Let $b_1,\ldots,b_{K-1}$ be an orthonormal basis of the span of $\Pi$. Then, by rotational invariance of the $\ell_2$ norm,
\begin{equation*}
  L_K =  \| (\text{Id} - \Pi ) \tilde Z \|_2^2 = \sum_{k=1}^{K-1}  (b_k \cdot \tilde Z)^2.
\end{equation*}
Furthermore, since the $b_k$ are orthogonal to each other $b_k \cdot \tilde Z$ are independent standard Gaussians. Thus, $L_K$ follows a $\chi^2(K-1)$ distribution. Furthermore, since the $b_k$ are orthogonal to $w$,  $L_K$ is independent of
\begin{equation*}
    \sum_{k=1}^K w_k \tilde Z_k.
\end{equation*}
Furthermore, by definition
\begin{equation*}
  \frac{1}{\sqrt{\alpha}}  \sum_{k=1}^K \alpha_k Z_k = \sum_{k=1}^K w_k \tilde Z_k \sim \mathcal{N}(0,1)
\end{equation*}
and thus $L_K$ is independent of $ \frac{1}{\sqrt{\alpha}} \sum_{k=1}^K \alpha_k Z_k$.
Therefore, \eqref{eq:todo} holds. Using \eqref{eq:todo} with \eqref{eq:asympsbet} and \eqref{eq:asympestim}, we get
\begin{equation*}
    \frac{\hat{\theta}^W - \theta^0}{\hat \sigma_{bet}/\sqrt{K-1}} = \frac{\delta \sum_{k=1}^{K} \alpha_k Z_k }{ \sqrt{\alpha} \delta \sqrt{L_K/(K-1)}} + o_P(1) = \frac{ \frac{\sum_k \alpha_k Z_k}{\sqrt{\alpha}}  }{\sqrt{L_K/(K-1)}} +o_P(1) = T_{K-1} +o_P(1),
\end{equation*}
where $T_{K-1}$ is a t-distributed random variable with $K-1$ degrees of freedom. This completes the proof.
\end{proof}

\subsection{Proof of Theorem~\ref{theorem:char}}\label{sec:proof-theorem-2}

\begin{proof}
In this proof, if not specified otherwise, all variances and covariances are meant with respect to $\mathbb{Q}$, that is marginally over both the variation in $D$ and $\xi$.
We will directly work with $U = h(D)$. Define
\begin{equation*}
    f(x) = \text{Var}(\mathbb{P}^\xi(U \in [0,x))).
\end{equation*}
Let $A$ and $B$ be two disjoint subsets of $[0,1]$. Define $a=\mathbb{P}(U \in A)$ and $ b= \mathbb{P}(U \in B )$. Then,
\begin{align*}
  f(a+b)  &= \text{Var}(\mathbb{P}^\xi(U \in A \cup B)) \\
  & = \text{Var}(\mathbb{P}^\xi(U \in A)) + \text{Var}(\mathbb{P}^\xi(U \in B)) + 2 \text{Cov}(\mathbb{P}^\xi(U \in A), \mathbb{P}^\xi(U \in B)) \\
&= f(a) + f(b) + 2 \text{Cov}(\mathbb{P}^\xi(U \in A) ,\mathbb{P}^\xi(U \in B)).
\end{align*}
Thus, for any two disjoint sets $A$ and $B$,
\begin{equation*}
    \text{Cov}(\mathbb{P}^\xi(U \in A), \mathbb{P}^\xi(U \in B)) = \frac{f(a+b) - f(a)- f(b)}{2}.
\end{equation*}
Define
\begin{equation}\label{eq:decompos1}
    g(a,b) = \frac{f(a+b) -f(a)-f(b)}{2}.
\end{equation}
Let us first show that $f$ is continuous. Let $a_n \rightarrow a$, $a_n \ge a$. Then,
\begin{equation*}
    f(a_n) - f(a) = f(a) + 2g(a_n - a, a) + f(a_n - a) - f(a) = 2 g(a_n-a,a) + f(a_n - a).
\end{equation*}
By Cauchy-Schwartz,
\begin{equation*}
    g(a_n-a,a) \le \sqrt{f(a_n -a) f(a)}.
\end{equation*}
By assumption, $f(a_n -a) \rightarrow 0$. Thus, $f(a_n)\rightarrow f(a)$. The case $a_n \rightarrow a$, $a_n \le a$ can be treated analogously. Thus, $f(\bullet)$ and $g(\bullet,\bullet)$ are continuous.

Partition the probability space into disjoint $D$-measurable events $A_i$, $i=1,\ldots,n$ with $P(U \in A_i) = 1/n$. Then,
\begin{equation*}
   0=  \text{Var}(\mathbb{P}^\xi(\cup A_i ) - \mathbb{P}( \cup A_i))  = n f(1/n) + n(n-1) g(1/n,1/n).
\end{equation*}
Thus,
\begin{equation}\label{eq:decompos2}
    g(1/n,1/n) = -1/(n-1) f(1/n).
\end{equation}
We will now show that $f(x) = x (1-x) \delta_\text{dist}^2$ for $x = 1/2^k$, where $\delta_\text{dist}^2 = 4 f(1/2)$. This will show that up to the constant $\delta_\text{dist}^2 = 4 f(1/2)$, $f$ and $g$ are uniquely defined. First, we will show this equality for $x=1/4$.
\begin{equation*}
    f(1/2) = f(1/4) + f(1/4) + 2 g(1/4,1/4) = 2 f(1/4) - 2/3 f(1/4).
\end{equation*}
Thus,
\begin{equation*}
    f(1/2) = 4/3 f(1/4).
\end{equation*}
Rearranging,
\begin{equation*}
   f(1/4) =  1/4 (1 - 1/4)  4 f(1/2) = x (1-x) \delta_\text{dist}^2
\end{equation*}
for $x=1/4$. Induction step: Assume that
\begin{equation*}
    f(x) = x(1-x) \delta_\text{dist}^2
\end{equation*}
for $x = 1/2^k$. Now we want to show that
\begin{equation*}
        f(x/2) = x/2 ( 1- x/2) \delta_\text{dist}^2.
\end{equation*}
To this end, using \eqref{eq:decompos1} and \eqref{eq:decompos2},
\begin{equation*}
    f(x) = f(x/2) + f(x/2) - 2/(2/x -1 ) f(x/2).
\end{equation*}
Thus,
\begin{equation*}
    f(x) = (2 - 2x/(2-x)) f(x/2) = (4-2x -2x)/(2-x) f(x/2) =  (4-4x)/(2-x) f(x/2).
\end{equation*}
By induction assumption,
\begin{equation*}
    f(x/2) =  (2-x)/(4-4x) x (1-x) \delta_\text{dist}^2 = x/2 (1-x/2) \delta_\text{dist}^2.
\end{equation*}
Thus, by induction for all $x=1/2^k$
\begin{equation*}
    f(x) = x (1-x) \delta_\text{dist}^2.
\end{equation*}
Now we want to show that for any $k$ and $j \le 2^k$ and $ x=j/2^k$
\begin{equation*}
    f(x) = x (1-x) \delta_\text{dist}^2.
\end{equation*}
For any $k$ and $j$ with $1 \le j \le 2^k$, using the definition of $f$  and \eqref{eq:decompos2},
\begin{align*}
    &f(j/2^k) = j f(1/2^k) - j(j-1)/(2^k-1) f(1/2^k) =  (j 2^k - j^2) /(2^k - 1) f(1/2^k) \\
    &= (j 2^k - j^2) /(2^k - 1) 1/2^k (1- 1/2^k) \delta_\text{dist}^2 = (j 2^k - j^2) 1/2^k 1/2^k \delta_\text{dist}^2 = j/2^k (1-j/2^k) \delta_\text{dist}^2.
\end{align*}
Thus, for all $k$ and $j \le 2^k$, and $x = j/2^k$,
\begin{equation*}
   f(x) = x(1-x) \delta_\text{dist}^2. 
\end{equation*}
Using continuity of $f$, for all $x \in [0,1]$
\begin{equation*}
    f(x) = x(1-x) \delta_\text{dist}^2.
\end{equation*}
We will now derive an explicit formula for $g$. For any $k$ and $j,j'$ with $j+j' \le 2^k$,
\begin{equation*}
    g(j/2^k,j'/2^k) = j j' g(1/2^k, 1/2^k) = - j j' /(2^k -1) f(1/2^k) = -j j' /(2^k -1) 1/2^k (1- 1/2^k) \delta_\text{dist}^2
\end{equation*}
\begin{equation*}
    = - j j' 1/2^k 1/2^k \delta_\text{dist}^2 = - j/2^k j'/2^k \delta_\text{dist}^2.
\end{equation*}
By continuity, for all $x \ge 0$, $y \ge 0$ with $x+ y \le 1$,
\begin{equation*}
    g(x,y) = - xy \delta_\text{dist}^2.
\end{equation*}
Now assume that for some $D$-measurable disjoint sets $A_i$ and some constants $y_i$,
\begin{equation*}
    \psi(D) = \sum 1_{A_i} y_i.
\end{equation*}
Then,
\begin{equation*}
    \text{Var}( \mathbb{E}^\xi[\psi(D)] - \mathbb{E}[\psi(D)]) = \sum_i y_i^2 f(P( A_i))  + \sum_{i \neq j} y_i y_j g(\mathbb{P}( A_i),\mathbb{P}( A_j)).
\end{equation*}
To simplify, let's write $p_i = P(A_i)$. Using explicit formulas for $f$ and $g$,
\begin{equation}\label{eq:first}
  \text{Var}( \mathbb{E}^\xi[\psi(D)] - \mathbb{E}[\psi(D)])  =\sum_i  \delta_\text{dist}^2 p_i (1-p_i) - \sum_{i \neq j}  \delta_\text{dist}^2 y_i y_j p_i p_j .
\end{equation}
On the other hand,
\begin{equation*}
    \delta_\text{dist}^2 \text{Var}_\mathbb{P}( \psi(D)) = \delta_\text{dist}^2 (\sum_i  p_i(1-p_i) y_i^2 + \sum_{i \neq j} \text{Cov}(1_{A_i},1_{ A_j}) y_i y_j).
\end{equation*}
As the sets are disjoint,  $\text{Cov}(1_{A_i},1_{ A_j}) = - p_i p_j$. Thus,
\begin{equation}\label{eq:second}
 \delta_\text{dist}^2 \text{Var}_\mathbb{P}( \psi(D))   = \delta_\text{dist}^2( \sum_i  p_i(1-p_i) y_i^2 + \sum_{i \neq j} -p_i p_j y_i y_j).
\end{equation}
Combining equation~\eqref{eq:first} with equation~\eqref{eq:second},
\begin{equation*}
     \text{Var}( \mathbb{E}^\xi[\psi(D)] - \mathbb{E}[\psi(D)]) =  \delta_\text{dist}^2 \text{Var}_\mathbb{P}( \psi(D)) .
\end{equation*}
By measure-theoretic induction, this result is extended to any $\psi(D) \in L^2(\mathbb{P})$.
\end{proof}

\subsection{Asymptotic Behaviour of \textit{M}-estimators}\label{appendix:m-estim}

\subsubsection{Proof of Lemma~\ref{lemma:asymptotic-consistency-m-estimators}}

\begin{proof}
The proof follows \cite{van2000asymptotic}, Theorem 5.14 with $m_\theta(D) = - L(\theta,D)$. 

Fix some $\theta$ and let $U_{\ell} \downarrow \theta$ be a decreasing sequence of open balls around $\theta$ of diameter converging to zero. Write $m_{U}(D)$ for $\sup_{\theta \in U}m_{\theta}(D)$. The sequence $m_{U_{\ell}}$ is decreasing and greater than $m_{\theta}$ for every $\ell$. As $\theta \xrightarrow[]{} m_{\theta}(D)$ is continuous, by monotone convergence theorem, we have $\mathbb{E}[m_{U_{\ell}}(D)] \downarrow \mathbb{E}[m_{\theta}(D)]$. 

For $\theta \neq \theta^0$, we have $\mathbb{E}[m_\theta(D)] < \mathbb{E}[m_{\theta^0}(D)]$.  Combine this with the preceding paragraph to see that for every $\theta \neq \theta^0$ there exits an open ball $U_{\theta}$ around $\theta$ with $\mathbb{E}[m_{U_{\theta}}(D)] < \mathbb{E}[m_{\theta^0}(D)]$. For any given $\epsilon > 0$, let the set $B = \{\theta \in \Omega : ||\theta - \theta^0||_2 \geq \epsilon\}$. The set $B$ is compact and is covered by the balls $\{U_{\theta}: \theta \in B\}$. Let $U_{\theta_1},\dots, U_{\theta_p}$ be a finite sub-cover. Then, 
\begin{align}
    \sup_{\theta \in B} \frac{1}{n} \sum_{i=1}^n m_{\theta}(D_i^n) & \leq \sup_{j = 1, \dots, p} \frac{1}{n} \sum_{i=1}^n m_{U_{\theta_j}}(D_i^n) 
    \\ &= \sup_{j = 1, \dots, p} \mathbb{E}[ m_{U_{\theta_j}}(D)] + o_p(1) < \mathbb{E}[m_{\theta^0}(D)] + o_p(1) \label{eq:pf-lemma2},
\end{align}
where for the equality we apply Lemma~\ref{lemma:random_weight} with $\psi(D) = m_{U_{\theta_j}}(D)$ for all $j=1,\ldots,p$. 

If $\hat{\theta} \in B$, then 
\begin{equation*}
    \sup_{\theta \in B} \frac{1}{n} \sum_{i=1}^n m_{\theta}(D_i^n) \geq \frac{1}{n} \sum_{i=1}^n m_{\hat{\theta}}(D_i^n)  \geq \frac{1}{n}\sum_{i=1}^n m_{\theta^0}(D_i^n) - o_p(1),
\end{equation*}
where the last inequality comes from the definition of $\hat{\theta}$.
Using Lemma~\ref{lemma:random_weight} with $\psi(D) = m_{\theta^0}(D)$, we have
\begin{equation*}
   \frac{1}{n}\sum_{i=1}^n m_{\theta^0}(D_i^n) - o_p(1) = \mathbb{E}[m_{\theta^0}(D)] - o_p(1).
\end{equation*}
Therefore, 
\begin{equation*}
    \{\hat{\theta} \in B\} \subset \Bigg\{ \sup_{\theta \in B} \frac{1}{n} \sum_{i=1}^n m_{\theta}(D_i^n) \geq \mathbb{E}[m_{\theta^0}(D)] - o_p(1) \Bigg\}.
\end{equation*}
By the equation \eqref{eq:pf-lemma2}, the probability of the event on the right hand side converges to zero as $n \xrightarrow[]{} \infty$. This completes the proof. 
\end{proof}

\subsubsection{Proof of Lemma~\ref{lemma:asymptotic-gaussianity-m-estimator}}

\begin{proof}
The proof follows \cite{van2000asymptotic}, Theorem 5.41 with $ \Psi_n(\theta) = \frac{1}{n} \sum_{i=1}^n \partial_\theta L(\theta,D_i^n)$ and $\dot{\Psi}_n (\theta) = \frac{1}{n} \sum_{i=1}^n \partial_\theta^2 L(\theta,D_i^n)$. By Taylor's theorem there exists a (random vector) $\tilde{\theta}$ on the line segment between $\theta^0$ and $\hat{\theta}$ such that 
\begin{equation*}
    0 = \Psi_n(\hat{\theta}) = \Psi_n(\theta^0) + \dot{\Psi}_n(\theta^0) (\hat{\theta} - \theta^0) + \frac{1}{2}(\hat{\theta} - \theta^0)^{\intercal} \Ddot{\Psi}_n(\tilde \theta) (\hat{\theta} - \theta^0).
\end{equation*}
By Lemma~\ref{lemma:random_weight} with $\psi(D) = \partial_\theta^2 L(\theta^0,D)$, we have
\begin{equation}\label{eq:lemma3-pf-1}
    \dot{\Psi}_n(\theta^0)  = \frac{1}{n} \sum_{i=1}^n \partial_\theta^2 L(\theta^0,D_i^n) = \mathbb{E}[ \partial_\theta^2 L(\theta^0,D) ] + o_P(1).
\end{equation}
By assumption, there exists a ball $B$ around $\theta^0$ such that $\theta \xrightarrow[]{} \partial_\theta^3 L(\theta,D)$ is dominated by a fixed function $h(\cdot)$ for every $\theta \in B$. The probability of the event $\{\hat{\theta} \in B\}$ tends to 1. On this event
\begin{equation*}
    \|\Ddot{\Psi}_n(\tilde \theta)\| = \Big|\Big| \frac{1}{n} \sum_{i=1}^n \partial_\theta^3 L(\tilde \theta,D_i^n)\Big|\Big| \leq \frac{1}{n} \sum_{i=1}^n h(D_i^n).
\end{equation*}
Using Lemma~\ref{lemma:random_weight} with $\psi(D) = h(D)$, we get
\begin{equation}\label{eq:lemma3-pf-2}
       \| \Ddot{\Psi}_n(\tilde \theta) \|   \le \frac{1}{n} \sum_{i=1}^n h(D_i^n) = O_P(1).
\end{equation}
Combining \eqref{eq:lemma3-pf-1} and \eqref{eq:lemma3-pf-2} gives us
\begin{align*}
    -\Psi_n(\theta^0) & = \left(\mathbb{E}[ \partial_\theta^2 L(\theta^0,D) ] + o_P(1) + \frac{1}{2}(\hat{\theta} - \theta^0)\text{ }O_P(1)\right)(\hat{\theta} - \theta^0) \\ &= \left(\mathbb{E}[ \partial_\theta^2 L(\theta^0,D) ] + o_P(1) \right)(\hat{\theta} - \theta^0).
\end{align*}
The probability that the matrix $\mathbb{E}[ \partial_\theta^2 L(\theta^0,D) ] + o_P(1) $ is invertible tends to 1.  Multiplying the preceding equation by $\sqrt{n}$ and applying $(\mathbb{E}[ \partial_\theta^2 L(\theta^0,D) ] + o_P(1))^{-1} $ left and right complete the proof. 

\end{proof}

\section{Robust Calibrated Inference}\label{sec:altway}

In some cases, the data analyst may trust one of the estimators $\hat\theta^k$ more than others. For example, the data analyst may be convinced that $\theta^1 = \theta^0$ but may not be sure whether $\theta^k = \theta^0$ for $k \ge 2$. In this case, the data analyst may report the confidence interval for $\theta^0$ using $\hat\theta^1$ instead of $\hat\theta^W$ with $\delta$ estimated by looking at the between-estimator variance of the remaining $K-1$ estimators. Now we present how to build asymptotically valid confidence intervals in such cases. 

\begin{theorem} (Asymptotic validity of calibrated confidence interval).
Suppose Assumption \ref{assump:a1} holds for $k =1,\ldots,K$ and the influence functions $\phi^1(D), \dots, \phi^K(D)$ are uncorrelated. Suppose $\theta^1 = \theta^0$ but $\theta^k$ may not be $\theta^0$ for $k \ge 2$. Furthermore assume that we have consistent estimates of the variances of influence functions such that $\widehat{\text{Var}_{\mathbb{P}}(\phi^k(D))} = \text{Var}_{\mathbb{P}}(\phi^k(D)) + o_p(1) $ for $k = 1, \dots, K$. 
Let $\hat \theta^W = \sum_{k=2}^K \hat \alpha_k \hat \theta^k$ be the inverse-variance weighted estimator of $K-1$ estimators where the weights are 
\begin{equation*}
    \hat{\alpha}_k =  \frac{\frac{1}{\widehat{\text{Var}_{\mathbb{P}}(\phi^k(D))}} }{\sum_{j=2}^K \frac{1}{\widehat{\text{Var}_{\mathbb{P}}(\phi^{j}(D))}}}.
\end{equation*}
Let $\hat{\sigma}_{\text{bet}}$ be the weighted between-estimator variance of $K-1$ estimators defined as  
\begin{equation*}
   \hat \sigma^2_{\text{bet}} = \sum_{k=2}^{K} \hat \alpha_k (\hat{\theta}^k - \hat \theta^W)^2.
\end{equation*}
Then for any $\alpha \in (0, 1)$, it holds that as $n \xrightarrow[]{} \infty$,
\begin{equation*}
   \lim \inf_{n \rightarrow \infty}  \mathrm{P} \left(\theta^0 \in \Big[\hat{\theta}^1 \pm t_{K-2,1-\alpha/2} \cdot \sqrt{\sum_{j=2}^K \frac{\widehat{\text{Var}_{\mathbb{P}}(\phi^{1}(D))}}{\widehat{\text{Var}_{\mathbb{P}}(\phi^{j}(D))}}}\frac{\hat{\sigma}_{\text{bet}}}{\sqrt{K-2}} \Big]\right) \geq 1-\alpha,
\end{equation*}
where $t_{K-2, 1-\alpha/2}$ is the $1-\alpha/2$ quantile of the $t$ distribution with $K-2$ degrees of freedom. To be clear, here we marginalize over both the randomness due to sampling and the randomness due to the distributional perturbation.
\end{theorem}

The resulting confidence intervals are expected to be conservative. Firstly, we lose one degree of freedom of the $t$-distribution. Secondly, we get an overcoverage if $\theta^k \neq \theta^0$ for $k \ge 2$. 

\quad

\begin{proof}
If $\theta^k \neq \theta^0$ for some $k \ge 2$, then by asymptotic linearity $\hat \sigma^2_{bet}$ converges to some $\tau^2> 0$. As in the proof of Theorem~\ref{theorem:newci}, we get $\hat{\theta}^1 - \theta^0 = \mathcal{N}(0, \delta^2\text{Var}_\mathbb{P}(\phi^1)/n) + o_P(1/\sqrt{n})$. Since the variance estimates are consistent,
\begin{equation*}
    \mathrm{P} \left(\theta^0 \in \Big[\hat{\theta}^1 \pm t_{K-2,1-\alpha/2} \cdot \sqrt{\sum_{j=2}^K \frac{\widehat{\text{Var}_{\mathbb{P}}(\phi^{1}(D))}}{\widehat{\text{Var}_{\mathbb{P}}(\phi^{j}(D))}}}\frac{\hat{\sigma}_{\text{bet}}}{\sqrt{K-2}} \Big]\right) \rightarrow 1.
\end{equation*}
Now let us consider the case $\theta^0 = \theta^1 = \ldots = \theta^K$. From the proof of Theorem \ref{theorem:newci}, we know that 
\begin{equation*}
    \frac{\sqrt{n}(\hat \theta^1 - \theta^0)}{\sqrt{\text{Var}_{\mathbb{P}}(\phi^{1}(D))}} \stackrel{d}{=} \delta Z + o_p(1),
\end{equation*}
where $Z \sim N(0, 1)$. Moreover,
\begin{align*}
    n\hat{\sigma}^2_{\text{bet}} \stackrel{d}{=} \delta^2 \frac{1}{\sum_{j=2}^{K}\frac{1}{\text{Var}_{\mathbb{P}}(\phi^j(D))}} \cdot L_{K-1} + o_p(1),
\end{align*}
where $L_{K-1}$ follows the  chi-square distribution with $K-2$ degrees of freedom. 
Note that $Z$ and $L_{K-1}$ are independent. Then,
we get
\begin{equation*}
   \frac{\frac{\hat{\theta}^1 - \theta^0}{\sqrt{\widehat{\text{Var}_{\mathbb{P}(\phi_1)}}}}}{\hat \sigma_{bet}\sqrt{\sum_{j=2}^{K}\frac{1}{\widehat{\text{Var}_{\mathbb{P}(\phi_k)}}}}/\sqrt{K-2}} \xrightarrow[]{d} \frac{Z }{ \sqrt{L_{K-1}/(K-2)}}.
\end{equation*}
Thus, we get
\begin{equation*}
    \lim _{n \xrightarrow[]{} \infty}P\left( \frac{\frac{\hat{\theta}^1 - \theta^0}{\sqrt{\widehat{\text{Var}_{\mathbb{P}(\phi_1)}}}}}{\hat \sigma_{bet}\sqrt{\sum_{j=2}^{K}\frac{1}{\widehat{\text{Var}_{\mathbb{P}(\phi_k)}}}}/\sqrt{K-2}} \leq x\right) 
    = P\left( t_{K-2} \leq x\right),
\end{equation*}
where $t_{K-2}$ is a $t$-distributed random variable with $K-2$ degrees of freedom. This completes the proof. 
\end{proof}

\section{Additional Simulation Results}\label{sec:add_simulations} 

In this section, we include additional simulation results.

\paragraph{Accuracy of $\hat{\delta}^2$.} The estimation accuracy of $\hat{\delta}$ compared to the ground truth $\delta$ is illustrated in Figure \ref{fig:delta-accuracy}. Recall that the distribution of $\hat{\delta}^2$ follows $\delta^2 \cdot \frac{\chi^2(K-1)}{K-1}$ as $n, m \xrightarrow[]{} \infty$.

\begin{figure}[]
    \centering
    \includegraphics[scale = 0.75]{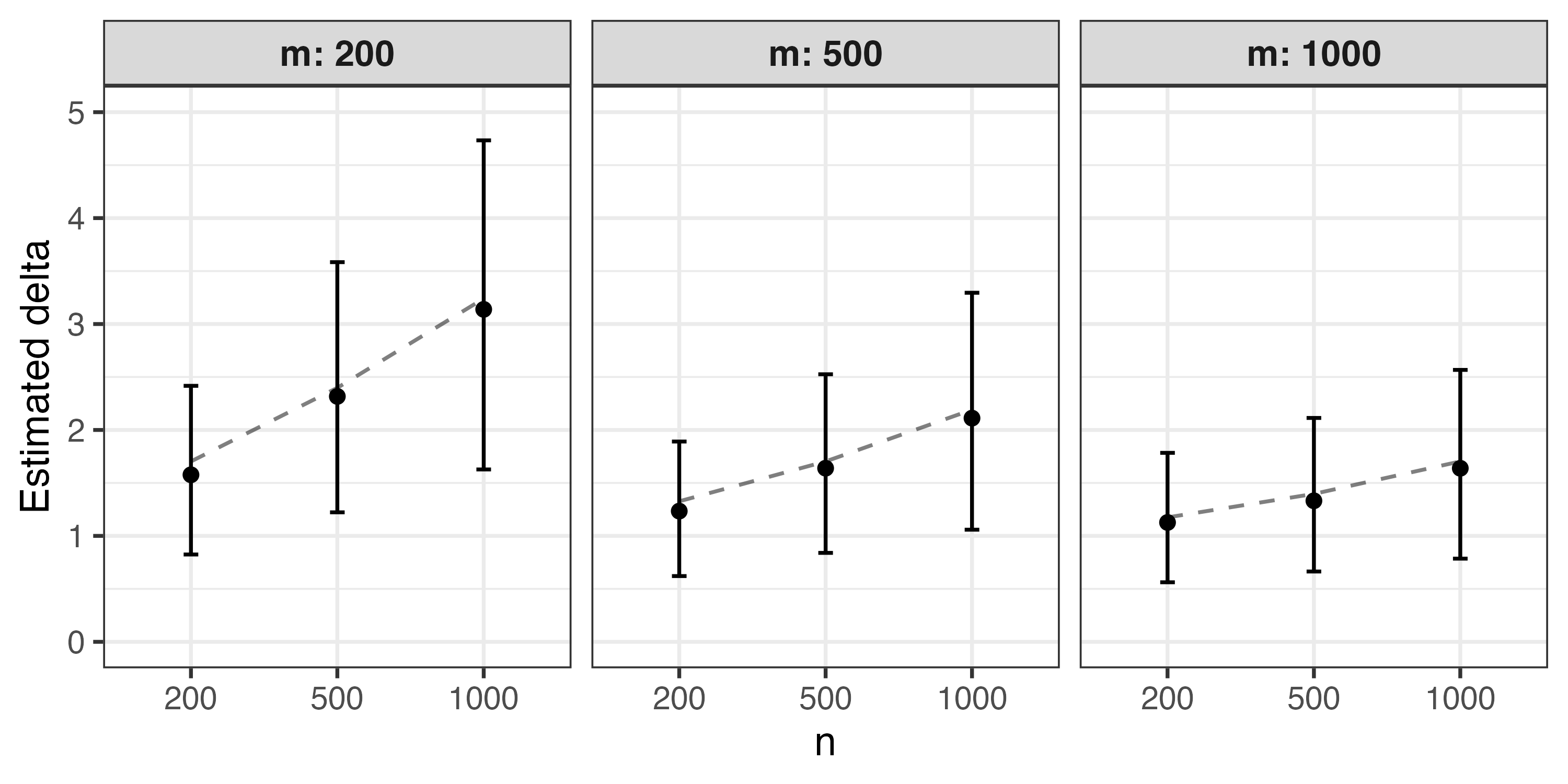}
    \includegraphics[scale = 0.75]{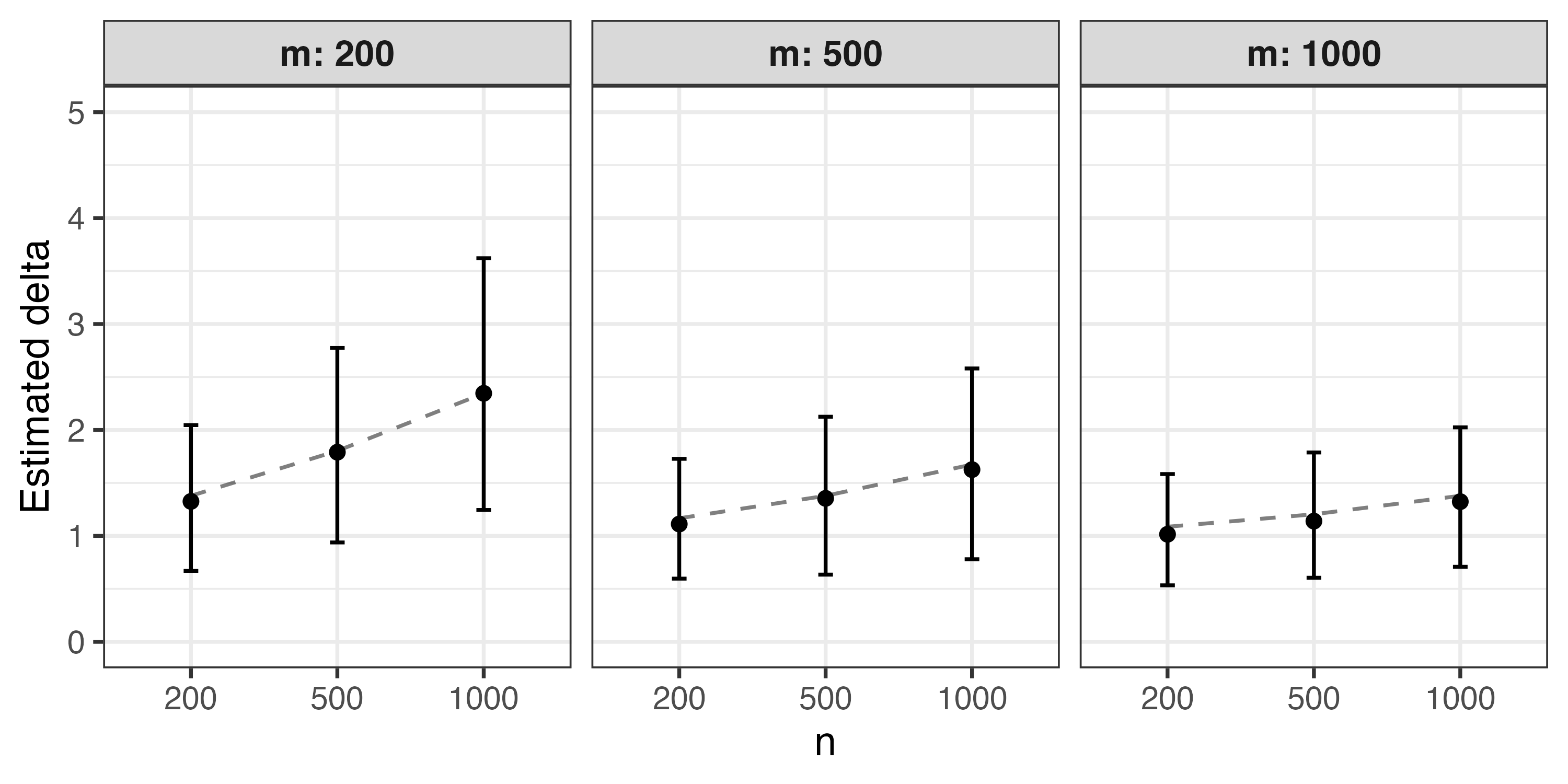}
    \caption{Accuracy of $\hat{\delta}$: The panel above shows the results under the perturbation model in Lemma~\ref{lemma:random_weight} and the panel below shows the results under the perturbation model in Example~\ref{example:samp} in the Appendix. Mean, $5\%$, and $95\%$ quantiles of the estimated $\hat{\delta}$ for each $m = 200, 500, 1000$ and $n = 200, 500, 1000$ are provided. The dashed lines indicate the true values of $\delta$.}
    \label{fig:delta-accuracy}
\end{figure}

\paragraph{Mariginal Coverages of Calibrated Confidence Intervals with Highly Correlated Estimators.}  Instead of using $K = 6$ adjustment sets in the main text, we consider the following $K= 8$ adjustment sets; $\{X_1, X_2\}$, 
$\{X_1, X_2, X_3\}$, 
$\{X_1, X_2, X_4\}$, 
$\{X_1, X_2, X_5\}$, $\{X_1, X_2, X_3, X_4\}$, $\{X_1, X_2, X_3, X_5\}$, $\{X_1, X_2, X_4, X_5\}$, $\{X_1, X_2, X_3, X_4, X_5\}$. The results are depicted in Figure~\ref{fig:coverage2}. In this case, some estimators are highly correlated, resulting in slight undercoverage of calibrated confidence intervals.

\begin{figure}[ht]
    \centering
    \includegraphics[scale = 0.8]{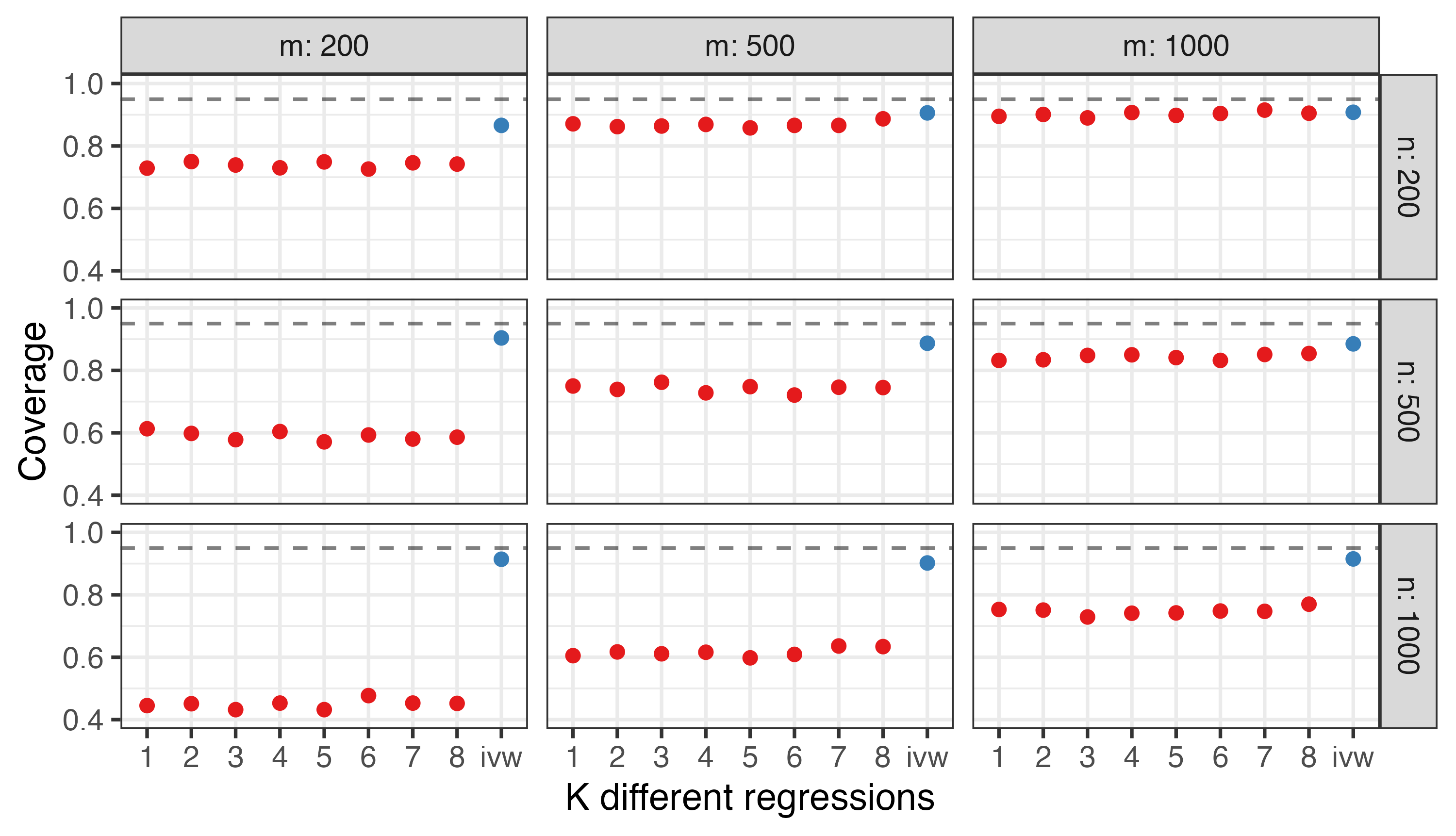}
   \includegraphics[scale = 0.8]{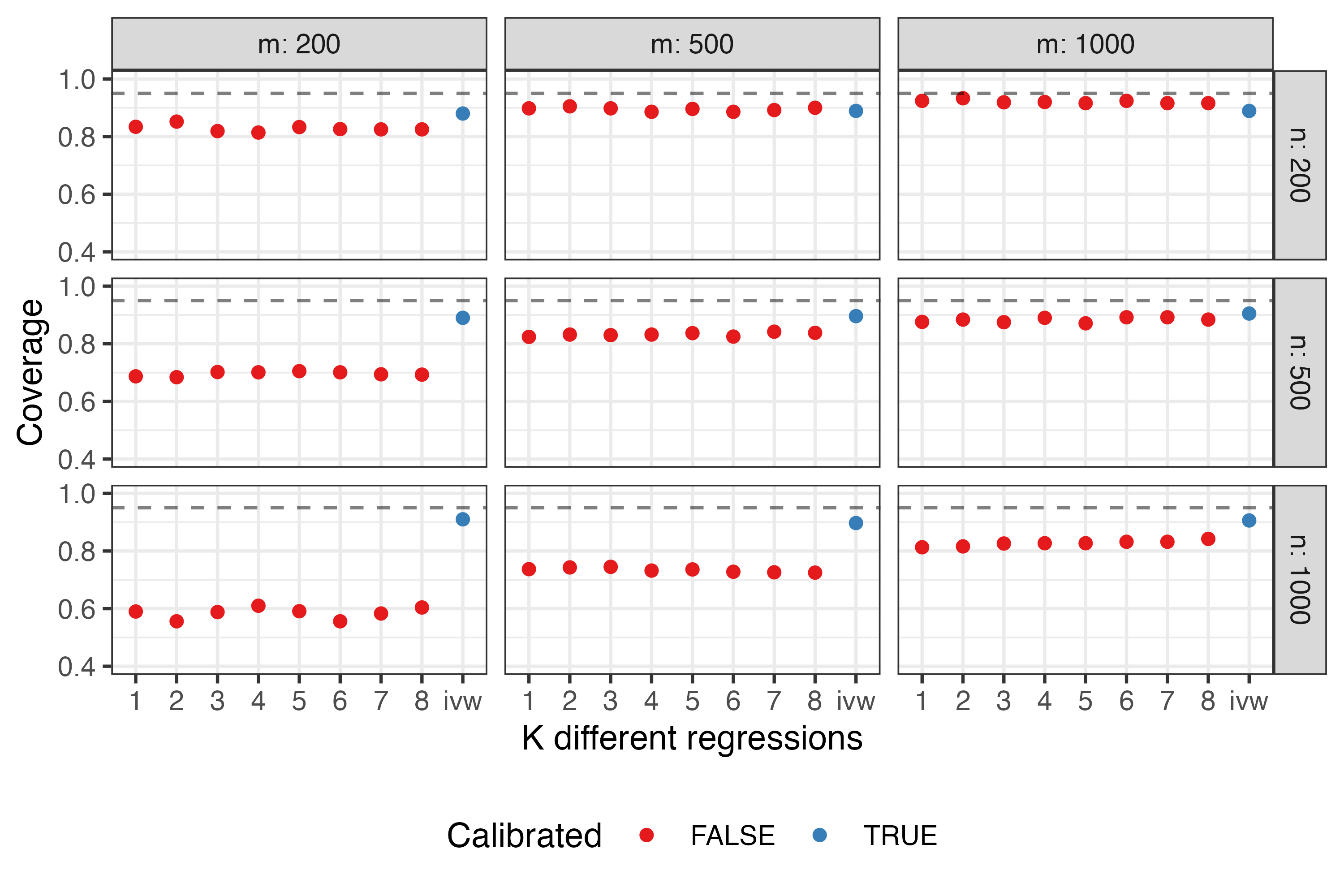}
    \caption{Marginal coverages of calibrated confidence intervals with $K=8$ different adjustment sets: The panel above shows the results under the perturbation model in Lemma~\ref{lemma:random_weight} and the panel below shows the results under the perturbation model in Example~\ref{example:samp} in the Appendix. Marginal coverages of non-calibrated confidence intervals for each regression-adjusted estimator and calibrated confidence intervals for the inverse-variance weighted estimator are presented for $m = 200, 500, 1000$ and $n = 200, 500, 1000$. The dashed lines indicate the nominal coverage 0.95.}
    \label{fig:coverage2}
\end{figure}

\section{Additional Data Analysis Results}\label{sec:add-hist}

In this section, we present additional results from real-world data analysis. Below are the figures showing the histograms of the lengths of confidence interval from Section \ref{sec:real-world}.

\begin{figure}
    \centering
    \includegraphics[width=0.9\linewidth]{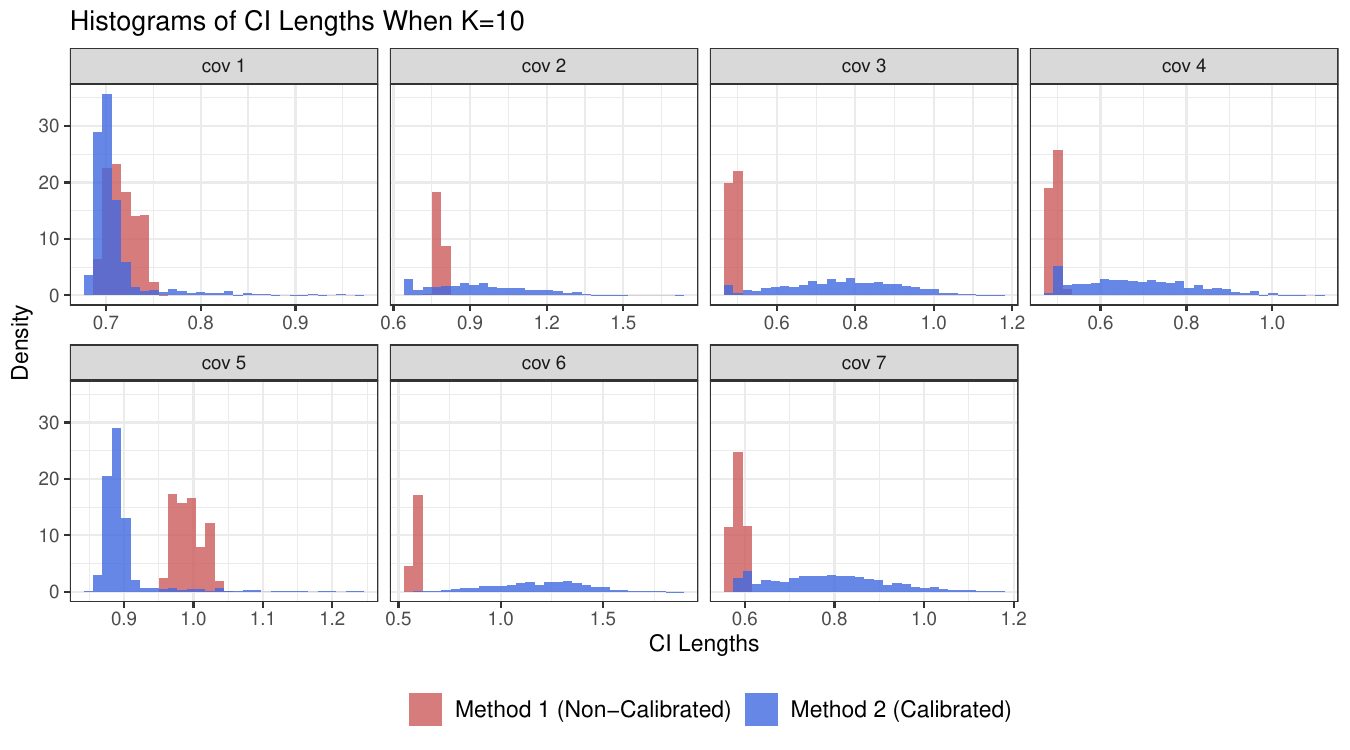}
    \includegraphics[width=0.9\linewidth]{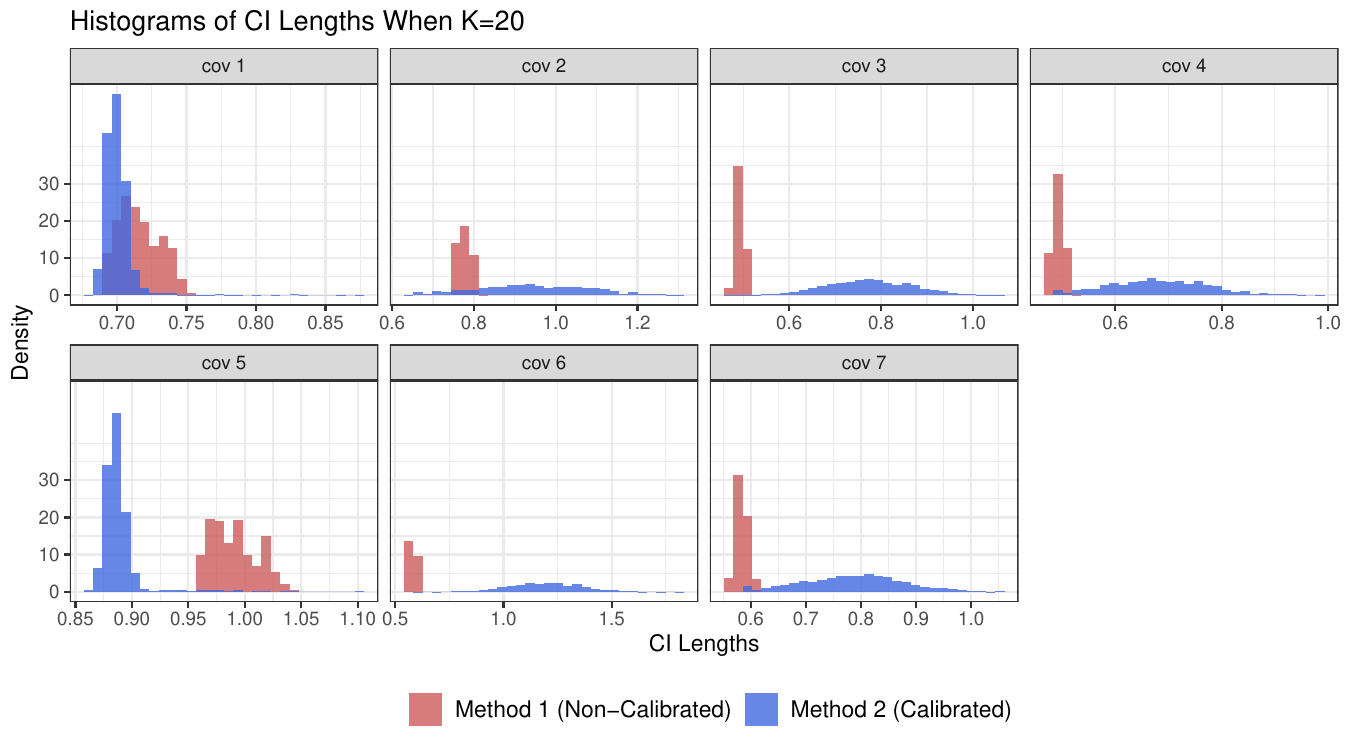}
    \caption{The histograms of the lengths of calibrated and non-calibrated confidence intervals for each selected binary covariate, based on $N = 1000$ iterations, are provided for $K = 10$ and $K = 20$.}
    \label{fig:enter-label}
\end{figure}

\end{document}